\documentclass[10pt]{article}

\usepackage{xspace}
\usepackage{url}
\usepackage{mathtools}
\usepackage{amssymb}
\usepackage{amsthm}
\usepackage{empheq}
\usepackage{latexsym}
\usepackage{enumitem}
\usepackage{eurosym}
\usepackage{dsfont}
\usepackage{appendix}
\usepackage{color} 
\usepackage[unicode,hypertexnames=false]{hyperref}
\usepackage{frcursive}
\usepackage[utf8]{inputenc}
\usepackage[T1]{fontenc}
\usepackage{geometry}
\usepackage{multirow}
\usepackage{lmodern}
\usepackage{anyfontsize}
\usepackage{stmaryrd}
\usepackage{natbib}
\usepackage{cleveref}
\usepackage[english]{babel}
\usepackage[english=british]{csquotes}
\usepackage{accents}
\usepackage{bm}
\usepackage[most]{tcolorbox}
\usepackage{graphicx}
\usepackage{todonotes}
\graphicspath{{figures/}}
\usepackage{float}
\usepackage{aliascnt}
\usepackage{microtype}

\setcitestyle{numbers,open={[},close={]}}

\definecolor{green}{rgb}{0,0.5,0}
\definecolor{blue}{rgb}{0,0,0.7}
\definecolor{purple}{rgb}{0.5,0,0.5}
\hypersetup{colorlinks, linkcolor={blue}, citecolor={green}, urlcolor={blue}}
			
\makeatletter
\@addtoreset{equation}{section}

\makeatother

\newtheorem{theorem}{Theorem}[section]

\newaliascnt{assumption}{theorem}
\newtheorem{assumption}[assumption]{Assumption}
\aliascntresetthe{assumption}

\newaliascnt{proposition}{theorem}
\newtheorem{proposition}[proposition]{Proposition}
\aliascntresetthe{proposition}

\newaliascnt{definition}{theorem}
\newtheorem{definition}[definition]{Definition}
\aliascntresetthe{definition}

\newaliascnt{lemma}{theorem}
\newtheorem{lemma}[lemma]{Lemma}
\aliascntresetthe{lemma}

\newaliascnt{example}{theorem}

\aliascntresetthe{example}

\newaliascnt{corollary}{theorem}
\newtheorem{corollary}[corollary]{Corollary}
\aliascntresetthe{corollary}

\newaliascnt{remark}{theorem}
\newtheorem{remark}[remark]{Remark}
\aliascntresetthe{remark}

\newaliascnt{condition}{theorem}

\aliascntresetthe{condition}

\crefname{theorem}{theorem}{theorems}
\Crefname{theorem}{Theorem}{Theorems}

\crefname{assumption}{assumption}{assumptions}
\Crefname{assumption}{Assumption}{Assumptions}

\crefname{proposition}{proposition}{propositions}
\Crefname{proposition}{Proposition}{Propositions}

\crefname{definition}{definition}{definitions}
\Crefname{definition}{Definition}{Definitions}

\crefname{lemma}{lemma}{lemmas}
\Crefname{lemma}{Lemma}{Lemmas}

\crefname{example}{example}{examples}
\Crefname{example}{Example}{Examples}

\crefname{corollary}{corollary}{corollaries}
\Crefname{corollary}{Corollary}{Corollaries}

\crefname{remark}{remark}{remarks}
\Crefname{remark}{Remark}{Remarks}

\crefname{equation}{equation}{equations}
\Crefname{equation}{Equation}{Equations}

\crefname{condition}{condition}{conditions}
\Crefname{condition}{Condition}{Conditions}

\DeclareUnicodeCharacter{014D}{\=o}
\newcommand\bg{\mathbf{g}}

\def\Fc{\mathcal{F}}

\newcommand{\smallertext}[1]{\text{\fontsize{6}{6}\selectfont$#1$}}
\newcommand{\smalltext}[1]{\text{\fontsize{4}{4}\selectfont$#1$}}

\begin{document}

\title{Equilibrium prices under hidden Markov fundamentals}
\author{Henri \textsc{Pag\`es}\footnote{Consultant at Banque de France, henripages@sfr.fr.} \and Dylan \textsc{Possama\"i}\footnote{ETH Z\"urich, Mathematics department, Switzerland, dylan.possamai@math.ethz.ch. This author gratefully acknowledges partial support by the SNF project MINT 205121-219818.}\and Mateo \textsc{Rodriguez Polo}\footnote{ETH Z\"urich, Mathematics department, Switzerland, mateo.rodriguezpolo@math.ethz.ch. This author gratefully acknowledges partial support by the SNF project MINT 205121-219818.}}

\date{18 September 2026}
\maketitle

\begin{abstract}
We study a representative-agent Epstein--Zin economy with geometric dividends and a hidden finite-state Markov drift. We allow the price--dividend ratio to contain an additional positive, absolutely continuous valuation factor and, within the class $\mathfrak C$ defined below and under the regularity, admissibility, and positivity conditions of our main theorem, equilibrium forces this factor to be constant, yielding belief--Markovian prices. In the two-state case, under the stated positivity condition and strictly positive transition intensities, we prove existence, uniqueness, endpoint smoothness, interior analyticity, and uniform bounds for the positive solution of the pricing equation. For $0<\theta\leq1$, this solution supports an equilibrium under any continuous short rate satisfying the model's one-sided portfolio condition. Finally, in the two-state subregion $\eta>0$, we derive belief-dependent stock volatility, a European option-pricing PDE, and a leading short-maturity conditional risk-neutral log-return skewness expansion. 
\end{abstract}

\section{Introduction}

Asset prices depend not only on current cash flows, but also on what investors believe about economic fundamentals that they cannot observe directly. When investors learn over time and have recursive preferences, it is not obvious whether current dividends and beliefs are enough to determine equilibrium prices or whether other price components may also arise. We study this question in a continuous-time model. Aggregate dividends satisfy
\[
  \frac{\mathrm dD_t}{D_t}=g_t\mathrm dt+\sigma\mathrm dW_t,
\]
where $W$ is a Brownian motion, $\sigma>0$ is constant, and $g$ is an unobserved finite-state Markov chain driving the different states the economy can be in. The agent learns about the current growth regime from dividend observations. Investors observe only dividends, and thus we write $\mathbb F^\smallertext{D}=(\mathcal F_t^{\smallertext{D}})_{t\geq 0}$ for the completion of the natural filtration of $D$. The representative agent has infinite-horizon Epstein--Zin preferences over consumption. That is, for a consumption process $c$, the corresponding continuation utility $V^c$ satisfies the recursive equation
\[
V_t^c = \frac{\delta\theta}{1-\gamma} \mathbb E\bigg[\int_t^{\smallertext{+}\infty}\mathrm e^{-\delta\theta(s-t)}c_s^{1-1/\psi}\big((1-\gamma)V_s^c\big)^{1-1/\theta}\mathrm ds\bigg|\mathcal F_t^{\smallertext{D}}\bigg],\; t\geq0,
\]
where $\delta>0$ is the discount rate, $\gamma$ is the relative risk aversion, and $\psi$ is the intertemporal elasticity of substitution. Their interaction is summarised by
\[
\theta\coloneqq\frac{1-\gamma}{1-\psi^{-1}},
\]
a parameter which will turn out to be qualitatively crucial for our results.

\medskip
Our paper first relates to the literature on learning and incomplete information in asset pricing. \citeauthor*{detemple1986asset} \cite{detemple1986asset} provides an early general-equilibrium analysis in which investors filter an unobserved state, and \citeauthor*{pastor2009learning} \cite{pastor2009learning} survey the broader learning literature. Our closest time-additive antecedents are \citeauthor*{veronesi1999overreaction} \cite{veronesi1999overreaction,veronesi2000information}, and \citeauthor*{david2002option} \cite{david2002option}. The closest recursive hidden-regime benchmarks are \citeauthor*{brandt2004equilibrium} \cite{brandt2004equilibrium} and \citeauthor*{brevik2010information} \cite{brevik2010information}.

\medskip
\citeauthor*{veronesi1999overreaction} \cite{veronesi1999overreaction} studies two unobserved dividend-growth regimes with additive dividends and CARA preferences. The equilibrium stock price is increasing and convex in the posterior probability of the high-growth state. Consequently, bad news in good times has a larger price effect than comparable good news in bad times, while learning also generates belief-dependent return volatility. This belief-driven asymmetry motivates
our analysis of the shape of the price--dividend ratio. A closer time-additive benchmark for our pure-exchange and geometric-dividend setting is \citeauthor*{veronesi2000information} \cite{veronesi2000information}. That paper considers finitely many hidden growth states under time-additive isoelastic preferences and derives an affine price--dividend ratio in posterior beliefs. Its main question is how the quality of investors' information affects risk premia and stock-return moments. \citeauthor*{david2002option} \cite{david2002option} provide the closest pure-diffusion option-pricing benchmark. They consider lognormal fundamentals and finitely many hidden regimes under time-additive isoelastic preferences. Their price--dividend ratio is affine in posterior beliefs, and
they show how the common innovation to fundamentals and beliefs generates belief-dependent stochastic volatility, a changing return--volatility covariance, a market price of belief risk, and time variation in option-implied distributions. For two states, they also develop a Fourier option-pricing method.

\medskip
Learning-generated option skews also precede our analysis. \citeauthor*{guidolin2003option} \cite{guidolin2003option} derive closed-form European option prices in an equilibrium model in which the probabilities governing dividend news are unknown and updated recursively. They show that Bayesian learning can generate asymmetric implied-volatility skews and systematic patterns across maturities. These learning, volatility, and option-pricing channels are therefore not new to our paper.

\medskip
The second ingredient we use is recursive utility. Epstein--Zin preferences originate with \citeauthor*{epstein1989substitution} \cite{epstein1989substitution}, while
\citeauthor*{duffie1992asset} \cite{duffie1992asset} provide the continuous-time stochastic differential utility formulation used here. \citeauthor*{weil1989equity} \cite{weil1989equity} and \citeauthor*{epstein1991substitution} \cite{epstein1991substitution} study the asset-pricing implications of separating risk aversion from intertemporal substitution. Epstein--Zin preferences have already been combined with learning about hidden regimes. \citeauthor*{brandt2004equilibrium}
\cite{brandt2004equilibrium} study a discrete-time Lucas-tree economy with a hidden finite-state endowment regime, Bayesian and distorted learning rules,
and Epstein--Zin preferences. Their equilibrium price--dividend ratio satisfies a generally nonlinear expectation equation and is computed as a function of posterior beliefs. \citeauthor*{brevik2010information} \cite{brevik2010information} consider a hidden two-state economy with an additional noisy signal and focus on how information quality affects expected returns. They obtain the price--consumption ratio
\[
  w(p)=\big(p w_b^\kappa+(1-p)w_r^\kappa\big)^{1/\kappa},
\]
which is also the price--dividend ratio in their endowment economy. Their parameter
\[
  \kappa=\frac{1-\gamma}{1-\psi^{-1}},
\]
is the counterpart of our $\theta$. They show that this ratio is strictly convex in the belief when $\kappa<1$ and strictly concave when $\kappa>1$. Related work by \citeauthor*{ju2012ambiguity} \cite{ju2012ambiguity} combines learning about hidden Markov growth regimes with generalised recursive smooth-ambiguity preferences. 

\medskip
Relative to \citeauthor*{veronesi1999overreaction} \cite{veronesi1999overreaction}, \citeauthor*{david2002option} \cite{david2002option}, \citeauthor*{brandt2004equilibrium} \cite{brandt2004equilibrium}, and \citeauthor*{brevik2010information} \cite{brevik2010information}, our contribution goes in several directions. The continuous-time Epstein--Zin aggregator used here separates risk aversion from intertemporal substitution and leads, away from $\theta=1$, to a generally nonlinear pricing equation on the belief simplex. The nonlinearity itself is not new in view of the preceding recursive-utility models. Our first contribution is instead to study whether the belief--Markovian representation used in those models can be obtained as an equilibrium conclusion when a further valuation state is allowed at the outset.

\medskip
More specifically, the earlier papers work with belief-Markovian price representations, whereas we begin with a broader class in which the price--dividend ratio may contain an additional finite-variation valuation state and may therefore retain information about the history of beliefs. We consider a class $\mathfrak C$ of candidate systems which, for every initial belief $p$, generate a stock price of the form
\[
  S_t^p=D_t^pH_t^p{\varphi(P_t^p)},\; t\geq0,\; H_{\smallertext{0}}^p=h(p),
\]
where $P^p$ is the posterior-belief process, $\varphi$ is a positive function on the simplex, and $H^p$ is a positive, absolutely continuous valuation factor. The term $\varphi(P_t^p)$ captures the effect of current beliefs, whereas $H_t^p$ may be path-dependent.

\medskip
The class $\mathfrak C$ is broader than a belief--Markovian specification, but it does not contain arbitrary adapted path-dependent prices. The process $H^p$ must be the unique global positive solution of an ordinary differential equation driven by $(P^p,H^p)$, and the Borel maps generating $H^p$ and the short rate are common to all initial-belief model copies. We call a candidate a $\mathfrak C$-equilibrium when, for every $p$, the representative agent holds the outstanding stock, holds no risk-free asset in net supply, and consumes the dividend, as in \citeauthor*{veronesi1999overreaction} \cite{veronesi1999overreaction} and \citeauthor*{david2002option} \cite{david2002option}.

\medskip
Our main result shows that, under the regularity, admissibility, and positivity conditions of \Cref{thm:main}, including a classical $C^2$ state representation of the agent's value function, equilibrium eliminates the apparent additional freedom. Even when $h$ is assumed only Borel measurable, there is a constant $h_{\smallertext{0}}>0$ such that
\[
  h(p)=h_{\smallertext{0}}, \; H_t^p=h_{\smallertext{0}},\; t\geq0,
\]
for every initial belief. Hence, the stock price is Markovian in dividends and posterior beliefs. The belief--Markovian representation used in the earlier literature is therefore an equilibrium conclusion within $\mathfrak C$, not an assumption. This conclusion is deliberately stated within $\mathfrak C$: it does not assert that every conceivable adapted equilibrium price, or every solution of an infinite-horizon recursive valuation equation, must be belief--Markovian.

\medskip
This question is related to the literature on existence and uniqueness in infinite-horizon recursive models. More broadly, \citeauthor*{santos1997rational} \cite{santos1997rational} give conditions that rule out rational asset-pricing bubbles in infinite-horizon economies. \citeauthor*{borovicka2020necessary}
\cite{borovicka2020necessary,borovicka2021stability} derive necessary and sufficient conditions for existence and uniqueness in classes of discrete-time recursive-utility and equilibrium asset-pricing problems. \citeauthor*{pohl2024existence} \cite{pohl2024existence} study recursive wealth--consumption ratios, including applications with learning, and select the limit of finite-horizon solutions, which they characterise as the unique homogeneous solution without a bubble. These papers show why the existence of a stationary pricing equation does not, by itself, settle the economic selection problem. Our result addresses a different and more specific
question: equilibrium and admissibility eliminate the additional absolutely continuous state admitted by $\mathfrak C$ in the present continuous-time hidden-regime economy.

\medskip
Second, once the belief-Markovian representation is established, we derive the pricing equation for any finite number of hidden regimes and analyse it in detail in the two-state case. The resulting boundary-value problem is degenerate at the endpoints of the belief interval. Under the stated positivity condition and strictly positive transition intensities, we prove existence and uniqueness of its positive classical solution for every $\theta>0$. We further prove that the solution is uniformly bounded above and away from zero, smooth on the closed belief interval, and real analytic in its interior.

\medskip
For $0<\theta\leq1$, the two-state solution supports a $\mathfrak C$-equilibrium under any continuous short-rate specification satisfying the one-sided portfolio bound. This condition does not select a unique short rate. The direction and curvature of the belief response must also be distinguished. Denoting by $g^{\smallertext{1}}$ and $g^{\smallertext{2}}$ the high- and low-growth states, respectively, monotonicity is governed by
\[
\eta\coloneqq(1-\gamma)\big(g^{\smallertext{1}}-g^{\smallertext{2}}\big),
\]
the price--dividend ratio is increasing when $\eta>0$ and decreasing when $\eta<0$. Curvature depends on $\theta$ and on an additional sufficient condition. When the latter holds, the ratio is convex for $0<\theta<1$ and concave for $\theta>1$, while for $\theta=1$ it is affine without this additional condition. The increasing--convex case echoes the qualitative news asymmetry of \citeauthor*{veronesi1999overreaction} \cite{veronesi1999overreaction}, while $\theta=1$ yields an affine belief dependence analogous to that in \citeauthor*{david2002option} \cite{david2002option}. As explained above, \citeauthor*{brevik2010information} \cite{brevik2010information} already identify the corresponding preference-dependent change in curvature in discrete time; our contribution is the global existence, uniqueness, endpoint regularity, and shape analysis of the continuous-time degenerate boundary-value problem, together with its connection to the equilibrium-selection result.

\medskip
The maintained equilibrium analysis covers $0<\theta\leq1$. We do not analyse $\theta<0$, which includes the often-used parameter combination $\gamma>1$ and $\psi>1$. We study $\theta>1$ at the level of the positive solution of the Markovian pricing equation, but do not provide the additional recursive-utility selection rule required to give those results an equilibrium interpretation.

\medskip
A further related literature combines recursive preferences and option pricing. \citeauthor*{garcia2003empirical} \cite{garcia2003empirical} develop and estimate an Epstein--Zin option-pricing model with Markov latent variables, stochastic volatility, and time-varying skewness. The state history enters the information set on
which their Euler equations are conditioned, so their latent-variable specification is not the same Bayesian filtering problem as ours. \citeauthor*{benzoni2011explaining} \cite{benzoni2011explaining} study a continuous-time Epstein--Zin equilibrium in which investors update beliefs about a hidden regime governing rare-jump risk. Their mechanism generates changes in stock valuations and in the implied-volatility smirk. 

\medskip
Finally, in the two-state subregion $\eta>0$ (equivalently, $\gamma<1$), with the short rate set to its marginal upper bound and under the stated martingale and regularity conditions, we obtain the risk-neutral stock and belief dynamics and the European option-pricing equation. Although belief-dependent volatility and the basic option-pricing mechanism are already dealt with in the literature, we derive an explicit leading expansion for conditional risk-neutral log-return skewness. Under the short-time remainder condition in \Cref{prop:short-maturity-skew-sign}, conditional on $p_t=p$, and writing $\sigma_\smallertext{S}(p)$ for stock volatility and $\chi(p)$ for the
diffusion coefficient of the posterior, we show that
\[
 \operatorname{Skew}^{\mathbb Q}_t(\tau) =\frac{3\chi(p)\sigma_\smallertext{S}^{\prime}(p)}{\sigma_\smallertext{S}(p)}\sqrt{\tau}+o(\sqrt{\tau}).
\]
Thus the direction in which a belief revision changes stock volatility determines the leading sign of short-maturity conditional return skewness. When stock volatility is unimodal in the belief, the leading coefficient changes sign at its maximum. This result is an expansion for return skewness; it does not by itself establish the slope of implied volatility. The numerical exercise uses the equilibrium price--dividend ratio and marginal short rate to examine the corresponding finite-maturity implied-volatility curves.

\medskip
The rest of the paper is organised as follows. \Cref{sec:model-equilibrium} presents the economy, learning, preferences, financial markets, and the class $\mathfrak C$. \Cref{sec:equilibrium-characterisation} characterises equilibrium prices, derives the pricing equation on the belief simplex, and records the resulting uniform bounds. \Cref{sec:markovian-shape} studies the two-state pricing equation, while \Cref{sec:option-pricing} derives the physical and risk-neutral stock and belief dynamics and studies European options. \Cref{app:proof-main-equilibrium,app:proof-two-state-analysis} contain the proofs of the equilibrium characterisation and the two-state results, and \Cref{app:characteristic-functions} gives the Fourier pricing representation.

\bigskip
\begingroup
\footnotesize
\noindent\textbf{Notation:}
We write $\mathbb N\coloneqq\{0,1,2,\ldots\}$,
$\mathbb N^\star\coloneqq\{1,2,\ldots\}$,
$\mathbb{R}_\smallertext{+}\coloneqq[0,+\infty)$, and use $\mathbb C$ for the
complex plane. For $x\in\mathbb R$, we set
$x^{\smallertext{+}}\coloneqq\max\{x,0\}$. The symbol $|\cdot|$ denotes absolute value,
the Euclidean norm, or set cardinality according to context;
$\mathbf1_A$ denotes the indicator of a set $A$, $\mathbf 1_n$ is the
vector in $\mathbb R^n$ whose entries are all one, $e_i$ is the $i$-th
canonical basis vector of $\mathbb R^n$, and ${}^\top$ denotes transposition.

\medskip
For a compact set $K$, we write $C(K)$ for the Banach space of continuous
real-valued functions on $K$, endowed with the supremum norm
$\|\cdot\|_\infty$. Let $E$ be a finite-dimensional normed space, let
$O\subset\mathbb R^d$ be open, and let $k\in\mathbb N$.
$C^k(O;E)$ consists of the maps $f:O\longrightarrow E$ that are $k$ times
continuously differentiable. The space $C^k_{\mathrm{loc}}(O;E)$ consists
of the maps that are $k$ times continuously differentiable on every
relatively compact open subset of $O$. We omit $E$ for real-valued
functions. For real numbers $a<b$, $C^k([a,b];E)$ consists of the maps in
$C^k((a,b);E)$ whose derivatives up to order $k$ extend continuously to
$[a,b]$. If $I\subset\mathbb R$ and $O\subset\mathbb R^d$ are open, then
$C^{1,2}(I\times O;E)$ denotes the space of functions that are once
continuously differentiable in the first variable and twice continuously
differentiable jointly in the variables in $O$.

\medskip
For $n\in\mathbb N^\star$, let us define the probability simplex on $n$
states and its relative interior
\[
  \Delta^n
  \coloneqq
  \big\{p\in[0,1]^n:\mathbf 1_n^\top p=1\big\},
  \;
  \Delta_\circ^n
  \coloneqq
  \operatorname{ri}(\Delta^n)
  =\big\{p\in\Delta^n:p^i>0\text{ for every }i\in\{1,\dots,n\}\big\},
\]
 and its tangent space
\[
  T\Delta^n
  \coloneqq
  \big\{v\in\mathbb R^n:\mathbf 1_n^\top v=0\big\}.
\]

For a finite-dimensional normed space $E$, $f\in C^k(\Delta^n;E)$ means
that $f$ is the restriction to $\Delta^n$ of a $C^k$ function $\widetilde f$ defined on
an open neighbourhood of $\Delta^n$ in its affine hull. If
$O\subset\mathbb R^d$ is open, $f\in C^k(\Delta^n\times O;E)$ means that
$f$ is the restriction to $\Delta^n\times O$ of a $C^k$ function defined on
$U\times O$, where $U$ is an open neighbourhood of $\Delta^n$ in its affine
hull. For a real-valued $f\in C^2(\Delta^n)$,
$\partial_p f(p)\in T\Delta^n$ denotes the gradient along the affine hull,
and $\partial_{pp}f(p)$ denotes the Hessian on $T\Delta^n$. In other words, for
$p\in\Delta^n$ and $(v,w)\in T\Delta^n\times T\Delta^n$
\[
  v^\top\partial_p f(p)
  =
  \frac{\mathrm d}{\mathrm dt}\widetilde f(p+tv)\bigg|_{t=0},
  \;
  v^\top\partial_{pp}f(p)w
  =
  \frac{\partial^2}{\partial s\partial t}
  \widetilde f(p+sv+tw)\bigg|_{s=t=0},
\]
for any such extension $\widetilde f$. Note that these quantities do not depend on
the chosen extension.

\medskip
Let $(\Omega,\mathcal F,\mathbb G,\mathbb Q)$ be a filtered probability
space, and let $\mathbb E^{\mathbb Q}$ denote expectation under
$\mathbb Q$. The space $\mathbb L^1(\mathbb Q)$ consists of the
real-valued $\mathcal F$-measurable random variables $\xi$ such that
$\mathbb E^{\mathbb Q}[|\xi|]<+\infty$. We write
$\mathbb H^1(\mathbb G,\mathbb Q)$ for the space of real-valued
$\mathbb G$-predictable processes $Z$ such that
\[
  \mathbb E^{\mathbb Q}\Bigg[
    \Bigg(\int_{\smallertext{0}}^{\smallertext{+}\infty} |Z_t|^2\mathrm dt\Bigg)^{1/2}
  \Bigg]<+\infty.
\]
When the filtration is unambiguous, we abbreviate this space to
$\mathbb H^1(\mathbb Q)$. A real-valued, $\mathbb G$--progressively
measurable process $R$ is called pathwise locally Lebesgue-integrable under
$\mathbb Q$ if  

\[
  \int_{\smallertext{0}}^{\smallertext{T}}|R_t|\mathrm dt<+\infty,
  \; \mathbb Q\text{--a.s.},\; \text{for every }T<+\infty.
\]

We write $\mathrm dt\otimes\mathbb Q$ for the product of Lebesgue measure
and $\mathbb Q$. Finally, for a continuous local martingale $M$ with
$M_{\smallertext{0}}=0$,
$\mathcal E(M)$ denotes its stochastic exponential,
$\mathcal E(M)_t=\exp(M_t-\frac12\langle M\rangle_t)$, $t\geq 0$.
\endgroup

\section{Model and equilibrium problem}\label{sec:model-equilibrium}

\subsection{The model}

Fix an integer $n\geq1$ representing the number of states of the economy,
and use the simplex and differential conventions fixed above.

\medskip

For every initial belief $p\in\Delta^n$, we work on a complete filtered
probability space
$(\Omega^p,\mathcal F^p,\mathbb F^p,\mathbb P^p)$ carrying a Brownian
motion $W^p$ and an independent right-continuous Markov chain $g^p$. The
chain has initial distribution $p$ under $\mathbb P^p$ and takes values in the finite set
$G\coloneqq\{g^{\smallertext{1}},\dots,g^n\}$, where the regime-value vector
$\bg\coloneqq(g^1,\dots,g^n)^\top\in\mathbb R^n$ has pairwise distinct
entries. Its
infinitesimal generator is the matrix
\begin{equation*}
\Lambda\coloneqq \big( \lambda^{ij}\big)_{(i,j)\in\{1,\dots,n\}^\smalltext{2}} , 
\end{equation*}
where, for $(i,j)\in\{1,\dots,n\}^2$ with $i\neq j$, $\lambda^{ij}\geq0$ is the transition intensity
from state $g^i$ to state $g^j$, while
\[
  \lambda^{ii}=-\sum_{j\in\{ 1,\dots,n\}\setminus\{i\}}\lambda^{ij}.
\]

We write $\mathbb E^p$ for expectation under $\mathbb P^p$. Every process
and market indexed by $p$ is defined on this $p$-model copy, and all
almost-sure, integrability, admissibility, and optimality statements about it
are understood under $\mathbb P^p$. 

\medskip
We recall the model of \cite{david2002option}. For each $p\in\Delta^n$, define the
dividend process as the unique solution $(D_t^p)_{t\geq 0}$ of the following
stochastic differential equation 
\begin{equation*}
D_t^p=D_{\smallertext{0}}+\int_{\smallertext{0}}^tD_s^pg_{s}^p\mathrm{d}s
+\sigma \int_{\smallertext{0}}^tD_s^p\mathrm{d}W_{s}^p,\; t\geq 0,\;
\mathbb P^p\text{\rm--a.s.},
\end{equation*}
where $\sigma>0$ is a constant volatility parameter, and $D_{\smallertext{0}}\in(0,+\infty)$.
We denote by $\mathbb F^{\smallertext{D},p}$ the completed,
right-continuous natural filtration of $D^p$. We have
\begin{equation*}
D_t^p=D_{\smallertext{0}}\exp\bigg(\int_{\smallertext{0}}^t\Big(g_s^p-\frac{\sigma^2}{2}\Big)\mathrm{d}s
+\sigma W_t^p\bigg),\; t\geq 0,\; \mathbb P^p\text{\rm--a.s.}
\end{equation*}

Investors observe only dividends, so their information on the $p$-model
copy is represented by $\mathbb F^{\smallertext{D},p}$. They use dividend
innovations to update their beliefs about the hidden growth state. For a
fixed $p$, we denote by $\mathcal C^p$ the set of positive,
$\mathbb F^{\smallertext{D},p}$-adapted and measurable processes. For $n=2$, this is the hidden-regime learning environment used by
\citeauthor*{veronesi1999overreaction} \cite{veronesi1999overreaction} and
\citeauthor*{david2002option} \cite{david2002option}: investors observe
dividends, infer the current growth regime, and prices depend on the
posterior probability of the high-growth state. The formulation above keeps
the same economic mechanism but allows for an arbitrary finite number of
hidden growth regimes.

\subsection{The dynamics of investors' beliefs}

The maximisation problem is a control problem with partial observations
because investors do not observe the current growth state $g^p$ directly. For
each $p\in\Delta^n$, let
$P^p\coloneqq(P^{p,i})_{i\in\{1,\dots,n\}}$ denote the posterior associated
with the law $\mathbb P^p$. Put differently,
\begin{equation*}
P_t^{p,i}\coloneqq \mathbb P^p\big[g_t^p=g^i\big|
\mathcal F_t^{\smallertext{D},p}\big],\; t\geq 0,\;
i\in\{1,\dots,n\}.
\end{equation*}
In particular, $P_{\smallertext{0}}^p=p$ under $\mathbb P^p$.
 The assumptions of \citeauthor*{lipster1977statistics} \cite[Theorems
 9.1 and 9.2]{lipster1977statistics} hold in our setting, and the posterior
 vector is the unique solution of the Wonham filtering
 equation, that is, the finite-state Kushner--Stratonovich equation; see
 \citeauthor*{wonham1964applications} \cite{wonham1964applications}.

\begin{proposition}
Under $\mathbb P^p$, the posterior probabilities
$(P^{p,i})_{i\in\{1,\dots,n\}}$ satisfy the system
of equations, for $i\in\{1,\dots,n\}$ 
\begin{equation*}
P _{t}^{p,i}=p^i+\int_{\smallertext{0}}^t(\Lambda^\top P_s^p)^{i}\mathrm{d}s-\int_{\smallertext{0}}^t\sigma
^{-1}P_{s}^{p,i}\big(\overline{g}(P_s^p)-g^{i}\big)\mathrm{d}\overline{W}_{s}^p,\; t\geq 0,\; 
\mathbb P^p\text{\rm--a.s.}, 
\end{equation*}
where the map $\overline g:\mathbb{R}^n\longrightarrow\mathbb{R}$ is defined
by 
\begin{equation*}
\overline g(y)\coloneqq \sum_{i=1}^ny^ig^i,\; y\in\mathbb{R}^n,
\end{equation*}
and where $\overline{W}^p$ is an
$(\mathbb{F}^{\smallertext{D},p},\mathbb P^p)$--Brownian motion, representing
the innovation process, and defined by 
\begin{equation*}
\overline{W}_{t}^p
=\sigma^{-1}\bigg(
\int_{\smallertext{0}}^t\frac{\mathrm{d}D_s^p}{D_s^p}
-\int_{\smallertext{0}}^t\overline{g}(P_s^p)\mathrm{d}s
\bigg),\; t\geq 0, \; \mathbb P^p\text{\rm--a.s.} 
\end{equation*}
\end{proposition}

Although posterior beliefs are not directly observed, investors compute them
from dividend history. The model therefore has an equivalent fully observed formulation with $D^p$ and $P^p$ as state variables
\begin{equation*}
\begin{cases}
\displaystyle D_t^p=D_{\smallertext{0}}+\int_{\smallertext{0}}^tD_s^p\overline{g}(P_{s}^p)\mathrm{d}s+\sigma
\int_{\smallertext{0}}^tD_s^p\mathrm{d}\overline W_{s}^p,\; t\geq 0,\; \mathbb P^p\text{\rm--a.s.}, \\[0.8em] 
\displaystyle P_{t}^{p,i}=p^i+\int_{\smallertext{0}}^t(\Lambda^\top
P_s^p)^{i}\mathrm{d}s-\int_{\smallertext{0}}^t\sigma ^{-1}P_{s}^{p,i}\big(\overline{g}(P_s^p)-g^{i}\big)\mathrm{d}
\overline{W}_{s}^p,\; t\geq 0,\; \mathbb P^p\text{\rm--a.s.},\; i\in\{1,\dots,n\}.
\end{cases}
\end{equation*}

\subsection{Epstein--Zin preferences in continuous time}

We fix $p\in\Delta^n$ and consider investors with recursive preferences of Epstein--Zin type. The
following assumption collects the parameter restrictions maintained in the
equilibrium analysis.

\begin{assumption}[Preferences]
\label{assum:pref} Investors discount the future at rate $
\delta>0 $. They have the recursive preferences described below, with
relative risk aversion $\gamma\in(0,+\infty)
\setminus\{1\} $, and an intertemporal elasticity of substitution $
\psi\in(0,+\infty)\setminus\{1\} $. The preference parameter $\theta $ for
early resolution of uncertainty is defined by
\begin{equation*}
\theta\coloneqq \frac{1-\gamma}{1-\psi^{-1}}.
\end{equation*}
Unless explicitly stated otherwise, we assume that
$\theta\in(0,1]$. Equivalently, under the preceding restrictions on
$\gamma$ and $\psi$, this means
\[
  \gamma\in[1/\psi\wedge 1,1/\psi\vee 1].
\]
\end{assumption}

\medskip

For an infinite-horizon consumption strategy $c\in\mathcal C^p$, the utility
process $V^{p,c}$ should be an
$(\mathbb{F}^{\smallertext{D},p},\mathbb P^p)$--semi-martingale satisfying
\begin{equation*}
V_t^{p,c}=\frac{\delta\theta}{1-\gamma}\mathbb{E}^p\bigg[\int_t^{\smallertext{+}\infty}\mathrm{e}^{-
\delta\theta (s-t)}(c_s)^{1-\frac1\psi}((1-\gamma)V_s^{p,c})^{1-\frac1\theta}\mathrm{d}s
\bigg|\mathcal{F}^{\smallertext{D},p}_t\bigg],\; t\geq0,\; \mathbb P^p\text{\rm--a.s.}
\end{equation*}
The well-posedness of $V^{p,c}$ is not immediate. Formally, it is related to
the following infinite-horizon backward stochastic differential equation
(BSDE)
\begin{equation} \label{eq:bsde}
Y_t^{p,c}=\int_t^{\smallertext{+}\infty}F_s\big(c_s,Y_s^{p,c}\big)\mathrm{d}s-\int_t^{\smallertext{+}\infty}Z_s^{p,c}\mathrm{d}\overline W_s^p,\;
t\geq0,\; \mathbb{P}^p\text{\rm--a.s.},
\end{equation}
where the generator $F:\mathbb{R}_\smallertext{+}\times(0,+\infty)^2
\longrightarrow\mathbb{R}$ is defined by 
\begin{equation*}
F_t(c,y)\coloneqq \delta\theta \mathrm{e}^{-\delta t}c^{1-\frac1\psi}y^{1-\frac1\theta},\;
(t,c,y)\in\mathbb{R}_\smallertext{+}\times(0,+\infty)^2,
\end{equation*}
the formal relationship being given by 
\begin{equation*}
V_t^{p,c}=\frac{\mathrm{e}^{\delta\theta t}}{1-\gamma}Y_t^{p,c},\; t\geq 0,\; \mathbb P^p\text{\rm--a.s.},
\end{equation*}
which follows from It\^o's formula.

\begin{remark}[Scope of the case $\theta>1$]
\label{rem:theta-greater-one-scope}
{\rm\Cref{assum:pref}} restricts the equilibrium analysis to
$0<\theta\leq1$. We nevertheless study $\theta>1$ later at the level of the
Markovian pricing equation. This distinction is necessary because, when
$\theta>1$, the exponent $q=1-1/\theta$ lies in $(0,1)$ and the transformed
\emph{BSDE} generator has a continuous extension to $y=0$ for which the trivial zero solution is also available, so uniqueness of the utility process
is no longer automatic. An economically meaningful utility process must therefore be selected among its non-negative solutions. Proper-solution selection for related---but different---infinite-horizon Epstein--Zin problems has been explored and developed by \emph{\citeauthor*{herdegen2023infiniteI} \cite{herdegen2023infiniteI,herdegen2025proper}}. Since we do not extend their theory to the present hidden-state model, results stated for $\theta>1$ concern only the unique positive solution of the Markovian pricing equation. They acquire an equilibrium interpretation only under an appropriate utility-selection rule, whose construction is beyond the scope of this paper.
\end{remark}

\begin{remark}[{The case $\theta<0$}]
\label{rem:theta-negative-scope}
The restriction to $\theta>0$ is not merely technical. For infinite-horizon
Epstein--Zin stochastic differential utility, the region $\theta<0$,
including the empirically common case $\gamma>1$, $\psi>1$, raises a separate
well-foundedness issue; see \cite[Section~7.3]{herdegen2023infiniteI}.
\citeauthor{shigeta2026negative} \cite{shigeta2026negative} addresses this
issue using a transformed utility index and suitable admissibility
conditions. We do not develop this alternative formulation here and
therefore do not attach an equilibrium interpretation to a formal pricing
equation in that region.
\end{remark}

\medskip

The limiting case $\theta=1$ gives the standard time-separable utility
specification
\begin{equation*}
V_t^{p,c}=\frac{\delta}{1-\gamma}\mathbb{E}^p\bigg[\int_t^{\smallertext{+}\infty}\mathrm{e}^{-\delta
(s-t)}(c_s)^{1-\frac1\psi}\mathrm{d}s\bigg|\mathcal{F}^{\smallertext{D},p}_t\bigg],\; t\geq 0,\; \mathbb P^p\text{\rm--a.s.}
\end{equation*}

This is the preference specification closest to the representative-agent
settings in \citeauthor*{veronesi1999overreaction}
\cite{veronesi1999overreaction} and \citeauthor*{david2002option}
\cite{david2002option}. Away from $\theta=1$, the Epstein--Zin aggregator
is the source of the nonlinear term in the Markovian price--dividend equation
derived below.

\medskip
Simple general conditions for existence and uniqueness of $V^{p,c}$ are not
available; see \citeauthor*{kraft2013consumption}
\cite[Definition~2.1 and Remark]{kraft2013consumption} and
\citeauthor*{schroder1999optimal}
\cite[Theorem~1]{schroder1999optimal} for sufficient conditions. We therefore restrict our attention to
consumption strategies $c$ for which $V^{p,c}$ is well defined.

\begin{definition}
We let $\mathcal C_a^p$ denote the subset of $\mathcal C^p$ consisting of $c$ such that there exists a unique pair $(Y^{p,c},Z^{p,c})$ such that \eqref{eq:bsde} holds, with $Y^{p,c}$ a
positive, continuous, $\mathbb{F}^{\smallertext{D},p}$-adapted process, and
$Z^{p,c}\in\mathbb H^1(\mathbb P^p)$.
\end{definition}

\subsection{The financial market and the investors' problem}

Fix $p\in\Delta^n$ and consider the corresponding model copy. In addition to consuming, the investors in our model have the opportunity to invest in a financial market. For simplicity, we assume that the market
consists of one risky asset, whose value is given by an
$\mathbb{F}^{\smallertext{D},p}$-adapted, positive and continuous
semi-martingale $(S_t^p)_{t\geq 0}$, and a risk-free asset with value
$(S^{o,p}_t)_{t\geq 0}$ given by
\begin{equation*}
S^{o,p}_t\coloneqq \exp\bigg(\int_{\smallertext{0}}^tr_s^p\mathrm{d}s\bigg),\;
t\geq 0,\; \mathbb{P}^p\text{\rm--a.s.},
\end{equation*}
where the process $(r_t^p)_{t\geq 0}$ is
$\mathbb F^{\smallertext{D},p}$-adapted and measurable, and
represents the interest rate in the market. We require
$\int_{\smallertext{0}}^{\smallertext{T}}|r_s^p|\mathrm ds<+\infty$,
$\mathbb P^p$--almost surely, for every $T>0$.

\medskip
An investor can thus choose an investment strategy $(\pi_t)_{t\geq
0}$, that is to say a $[0,1]$-valued,
$\mathbb F^{\smallertext{D},p}$-adapted and measurable process,
representing the fraction of wealth held in the risky asset by the
investor at each time. We let the set of investment strategies be denoted by 
$\mathcal{U}$. Then, the wealth of the investor following an
investment--consumption strategy $(\pi,c)\in\mathcal{U}\times\mathcal{C}_a^p$,
and starting with initial wealth $x>0$, is given by the
process $X^{p,x,\pi,c}$ defined by
\begin{align*}
X^{p,x,\pi,c}_t&= x+\int_{\smallertext{0}}^t\big(r_s^p(1-\pi_s)X^{p,x,\pi,c}_s-c_s\big)
\mathrm{d}s+\int_{\smallertext{0}}^t\pi_sX^{p,x,\pi,c}_s\bigg(
\frac{\mathrm{d}S_s^p}{S_s^p}+\frac{D_s^p}{S_s^p}\mathrm{d}s\bigg)
,\; t\geq 0,\; \mathbb{P}^p\text{\rm--a.s.}
\end{align*}

\begin{definition}[Admissible strategies]
\label{def:admissibility}
For a fixed market $(S^p,r^p)$, write $Q^p\coloneqq S^p/D^p$. Given initial
wealth $x>0$ and $(\pi,c)\in\mathcal U\times\mathcal C^p_a$, suppose the
wealth equation has a strictly positive solution $X^{p,x,\pi,c}$ and define
\begin{equation*}
  \widehat Y_t^{p,x,\pi,c}
  \coloneqq
  \delta^\theta\mathrm e^{-\delta\theta t}
  \big(X_t^{p,x,\pi,c}\big)^{1-\gamma}(Q_t^p)^{\theta/\psi},\;
  t\geq 0,\; \mathbb{P}^p\text{\rm--a.s.}
\end{equation*}
Let $\widehat Z^{p,x,\pi,c}$ denote the predictable integrand in the local
martingale part of $\widehat Y^{p,x,\pi,c}$ with respect to
$\overline W^p$. An investment--consumption strategy is admissible, written
$(\pi,c)\in\mathcal A(x;S^p,r^p)$, if
\begin{enumerate}
\item[$(i)$] $(\pi,c)\in\mathcal U\times\mathcal C_a^p$, the wealth equation
has a solution, and $X_t^{p,x,\pi,c}>0$ for all $t\geq0$,
$\mathbb P^p\text{\rm--a.s.};$

\item[$(ii)$] $\widehat Z^{p,x,\pi,c}
\in\mathbb H^1(\mathbb P^p)$.

\item[$(iii)$]
$\liminf_{t\to\smallertext{+}\infty}\widehat Y_t^{p,x,\pi,c}=0$,
$\mathbb P^p\text{\rm--a.s.}$
\end{enumerate}
\end{definition}

\begin{remark}[The verification integrand]
\label{rem:verification-integrand}
For a fixed initial belief, the Wonham equation has a pathwise-unique strong
solution driven by $\overline W^p$, while the innovation formula recovers
$\overline W^p$ from the dividend filtration. Hence
$\mathbb F^{\smallertext{D},p}$ is the completed Brownian filtration generated by
$\overline W^p$ and has the martingale representation property. The integrand
$\widehat Z^{p,x,\pi,c}$ in {\rm\Cref{def:admissibility}} is therefore well defined
and unique up to $\mathrm dt\otimes\mathbb P^p$--null sets.
\end{remark}

For any fixed $x>0$ and initial belief $p\in\Delta^n$, the investor
maximises over $(\pi,c)\in\mathcal A(x;S^p,r^p)$ the initial utility
$V_{\smallertext{0}}^{p,\pi,c}(x;S^p,r^p)$. Recalling the definition of $\mathcal C_a^p$,
the process $(V_t^{p,\pi,c}(x;S^p,r^p))_{t\geq0}$ is the unique continuous,
$\mathbb F^{\smallertext{D},p}$-adapted process satisfying
\begin{equation*}
V^{p,\pi,c}_t(x;S^p,r^p)
=\frac{\delta\theta}{1-\gamma}\mathbb{E}^p\bigg[
\int_t^{\smallertext{+}\infty}\mathrm{e}^{-\delta\theta(s-t)}
(c_s)^{1-\frac1\psi}
\big((1-\gamma)V^{p,\pi,c}_s(x;S^p,r^p)\big)^{1-\frac1\theta}
\mathrm{d}s
\bigg|\mathcal{F}^{\smallertext{D},p}_t\bigg],
\; t\geq0,
\; \mathbb P^p\text{\rm--a.s.}
\end{equation*}

An investment--consumption strategy
$(\pi^\star,c^{\star})(x;S^p,r^p)\in\mathcal A(x;S^p,r^p)$
is optimal if it attains the supremum, denoted by $V_{\smallertext{0}}^p(x;S^p,r^p)$.

\subsection{Market equilibria}
The aim of this paper, however, is not merely to solve the preceding optimisation problem, but to determine equilibrium prices. In equilibrium, the representative agent holds the single outstanding stock share, the money-market account is in zero net supply, and aggregate consumption equals the dividend endowment.

\begin{definition}[$p$-equilibrium]
\label{def:p-equilibrium}
Fix $p\in\Delta^n$. A $p$-equilibrium is a pair $(S^p,r^p)$ of risky-asset
price and interest rate on the $p$-model copy for which there exists a strategy
$(\pi^\star,c^{\star})$ satisfying the conditions below. Admissibility,
utility, and optimality are evaluated under $\mathbb P^p$.
\begin{itemize}
\item[$(i)$] $(\pi^\star,c^{\star})\in
\mathcal A(S_{\smallertext{0}}^p;S^p,r^p)$ and attains
$V_{\smallertext{0}}^p(S_{\smallertext{0}}^p;S^p,r^p);$

\item[$(ii)$] $\pi_t^\star=1$,
$\mathrm dt\otimes\mathbb P^p\text{\rm--a.e.};$

\medskip

\item[$(iii)$] $c_t^{\star}=D_t^p$,
$\mathrm dt\otimes\mathbb P^p\text{\rm--a.e.}$
\end{itemize}
\end{definition}

We now introduce a class in which the price--dividend ratio depends, as expected, on the posterior beliefs, but may also contain an
additional path-dependent factor.
\begin{definition}[The class $\mathfrak C$]
\label{def:class-C}
The class $\mathfrak C$ consists of specifications
$(\varphi,h,\mathcal H,\mathcal R)$, where
\[
  \varphi\in C^2(\Delta^n;(0,+\infty)),
  \;
  h:\Delta^n\longrightarrow(0,+\infty),
  \;
  \mathcal H:\Delta^n\times(0,+\infty)\longrightarrow\mathbb R,
  \;
  \mathcal R:\Delta^n\times(0,+\infty)\longrightarrow\mathbb R,
\]
are independent of the initial belief. The maps $h$, $\mathcal H$, and
$\mathcal R$ are Borel measurable. Belonging to $\mathfrak C$ requires that, for every $p\in\Delta^n$, on the $p$-model copy the equation
\begin{equation*}
  H_t^p
  =
  h(p)
  +
  \int_{\smallertext{0}}^t
    \mathcal H(P_s^p,H_s^p) \mathrm{d}s,
  \; t\geq0,\; \mathbb P^p\text{\rm--a.s.},
\end{equation*}
admits a unique global, positive,
$\mathbb F^{\smallertext{D},p}$-adapted, absolutely continuous solution.
Given this solution, define
\begin{equation*}
  Q^p
  \coloneqq
  H^p\varphi(P^p),
  \;
  S^p=D^pQ^p,
  \;
  r^p
  =
  \mathcal R(P^p,H^p).
\end{equation*}
We also require, for every $T>0$,
\[
  \int_{\smallertext{0}}^\smallertext{T}|r_t^p|\mathrm dt<+\infty,
  \; \mathbb P^p\text{\rm--a.s.}
\]
Thus one specification generates the family
$((S^p,r^p))_{p\in\Delta^n}$ across the corresponding model copies.
\end{definition}

Finally, we introduce what will be our target in the rest of the paper.
\begin{definition}[$\mathfrak C$-equilibrium]
\label{def:C-equilibrium}
A specification in $\mathfrak C$ is a $\mathfrak C$-equilibrium if, for every
$p\in\Delta^n$, its generated pair $(S^p,r^p)$ under
$\mathbb P^p$ is a $p$-equilibrium in the
sense of {\rm\Cref{def:p-equilibrium}}.
\end{definition}

Within this class, we can write the stock and wealth dynamics explicitly.
For $y\in\Delta^n$, set
\[
  b(y)\coloneqq\Lambda^\top y,
  \;
  A(y)\coloneqq
  \operatorname{diag}(y)\big(\overline g(y)\mathbf 1_n-\bg\big).
\]
Both vectors are tangent to the simplex. Indeed, $\mathbf 1_n^\top b(y)=0,
  \;
  \mathbf 1_n^\top A(y)=0.$ For every specification in $\mathfrak C$ and every $p\in\Delta^n$, under
$\mathbb P^p$ the candidate price $S^p$ is an It\^o process satisfying
$S_{\smallertext{0}}^p=D_{\smallertext{0}}h(p)\varphi(p)$ and
\begin{align*}
\frac{\mathrm dS_t^p}{S_t^p}
&=\bigg(
  \overline g(P_t^p)
  +\frac{\mathcal H(P_t^p,H_t^p)}{H_t^p}
  +\frac{
    \big(b(P_t^p)-A(P_t^p)\big)^\top \partial_p\varphi(P_t^p)
    +\frac{1}{2\sigma^\smalltext{2}}
    A(P_t^p)^\top\partial_{pp}\varphi(P_t^p)A(P_t^p)
   }{\varphi(P_t^p)}
\bigg)\mathrm dt\\
&\quad+\bigg(
  \sigma-
  \frac{A(P_t^p)^\top \partial_p\varphi(P_t^p)}
     {\sigma\varphi(P_t^p)}
 \bigg){\mathrm d\overline W_t^p},
\; t\geq0,
\; \mathbb P^p\text{\rm--a.s.}
\end{align*}

It\^o's formula gives the associated candidate wealth
dynamics
\begin{align*}
X_{t}^{p,x,\pi ,c}&= x+\int_{0}^{t}\bigg(\mathcal{R}
(P_{s}^p,H_{s}^p)(1-\pi _{s})X_{s}^{p,x,\pi ,c}-c_{s}+\pi
_{s}X_{s}^{p,x,\pi ,c}\bigg(\overline{g}(P_{s}^p)+\frac{\mathcal{H}
(P_{s}^p,H_{s}^p)+\varphi(P_s^p)^{-1}}{H_{s}^p}\bigg)\bigg)\mathrm{d}s
 \\
&\quad +\int_{0}^{t}\pi _{s}X_{s}^{p,x,\pi ,c}
\frac{
  \big(b(P_s^p)-A(P_s^p)\big)^\top \partial_p\varphi(P_s^p)
  +\frac{1}{2\sigma^\smalltext{2}}
   A(P_s^p)^\top\partial_{pp}\varphi(P_s^p)A(P_s^p)
}{\varphi(P_s^p)} \mathrm ds\\
&\quad +\int_{0}^{t}\pi _{s}X_{s}^{p,x,\pi ,c}\bigg(
\sigma-
\frac{A(P_s^p)^\top \partial_p\varphi(P_s^p)}
   {\sigma\varphi(P_s^p)}
\bigg){\mathrm d\overline W_s^p},
\; t\geq0,
\; \mathbb P^p\text{\rm--a.s.}
\end{align*}

\subsection{The formal Hamilton--Jacobi--Bellman equation}

For a candidate generated by a specification in $\mathfrak C$, the investor
faces an optimal-control problem
with state variables $X^{p,x,\pi,c}$, $H^p$, and $P^p$, and controls $\pi$
and $c$. We use $u$ for the current value of the auxiliary factor, reserving
$h$ for the initial-value map in the definition of $\mathfrak C$.

\medskip
For $(\pi,c)\in[0,1]\times(0,+\infty)$ and a smooth map
$\vartheta:\Delta^n\times(0,+\infty)^2\longrightarrow\mathbb R$, the infinitesimal
generator is
\begin{align*}
\big(\mathcal L^{\pi,c}\vartheta\big)(p,x,u)
\coloneqq{}&
\bigg(
  \mathcal R(p,u)(1-\pi)x-c
  +\pi x\bigg(\overline g(p)
    +\frac{\mathcal H(p,u)+\varphi(p)^{-1}}{u}\bigg)
\bigg)\partial_x\vartheta(p,x,u)\\
&+\pi x\frac{
   \big(b(p)-A(p)\big)^\top \partial_p\varphi(p)
   +\frac{1}{2\sigma^2}A(p)^\top\partial_{pp}\varphi(p)A(p)
  }{\varphi(p)}\partial_x\vartheta(p,x,u)\\
&+\mathcal H(p,u)\partial_u\vartheta(p,x,u)
+b(p)^\top \partial_p\vartheta(p,x,u)
+\frac{1}{2\sigma^2}
 A(p)^\top\partial_{pp}\vartheta(p,x,u)A(p)\\
&-\frac{\pi x}{\sigma}
 \bigg(\sigma-
   \frac{A(p)^\top \partial_p\varphi(p)}{\sigma\varphi(p)}\bigg)
 A(p)^\top \partial_{xp}\vartheta(p,x,u)+\frac{\pi^2x^2}{2}
 \bigg(\sigma-
   \frac{A(p)^\top \partial_p\varphi(p)}{\sigma\varphi(p)}\bigg)^2
 \partial_{xx}\vartheta(p,x,u).
\end{align*}

\medskip

The Hamilton--Jacobi--Bellman equation is
\begin{align*}
0&=\delta\theta v(p,x,u)
-\sup_{(\pi,c)\in[0,1]\times(0,\smallertext{+}\infty)}
\bigg\{
  \big(\mathcal L^{\pi,c}v\big)(p,x,u)
  +\frac{\delta c^{1-1/\psi}}{1-1/\psi}
   \big((1-\gamma)v(p,x,u)\big)^{1-1/\theta}
\bigg\},
\end{align*}
or, equivalently
\begin{align}\label{eq:hjb}
0&=\delta\theta v(p,x,u)
-\mathcal R(p,u)xv_x(p,x,u)
-{b(p)^\top \partial_p v(p,x,u)}
-\mathcal H(p,u)v_u(p,x,u)
-{\frac{1}{2\sigma^2}
 A(p)^\top \partial_{pp}v(p,x,u)A(p)}\notag\\
&\quad-\mathcal G\big(
  p,x,u,v(p,x,u),v_x(p,x,u),v_{xx}(p,x,u),
  \partial_{xp}v(p,x,u)
\big),
\end{align}
where, for $p\in\Delta^n$, $x,u,z>0$, $y\in\mathbb R$ such that
$(1-\gamma)y>0$, $m\in\mathbb R$, and $k\in T\Delta^n$, we define
\begin{align*}
\mathcal G(p,x,u,y,z,m,k)
&\coloneqq
\sup_{\pi\in[0,1]}
\bigg\{
z\pi x\bigg(
  \overline g(p)-\mathcal R(p,u)
  +\frac{\mathcal H(p,u)+\varphi(p)^{-1}}{u}
\bigg)\\
&\quad+z\pi x\frac{
  \big(b(p)-A(p)\big)^\top \partial_p\varphi(p)
  +\frac{1}{2\sigma^\smalltext{2}}A(p)^\top\partial_{pp}\varphi(p)A(p)
 }{\varphi(p)}\\
&\quad-\frac{\pi x}{\sigma}
   \bigg(\sigma-
    \frac{A(p)^\top \partial_p\varphi(p)}{\sigma\varphi(p)}\bigg)
   A(p)^\top  k+\frac{\pi^2x^2m}{2}
 \bigg(\sigma-
  \frac{A(p)^\top \partial_p\varphi(p)}{\sigma\varphi(p)}\bigg)^2
\bigg\}\\
&\quad+\sup_{c>0}
\bigg\{-cz+
\frac{\delta c^{1-1/\psi}}{1-1/\psi}((1-\gamma)y)^{1-1/\theta}
\bigg\}.
\end{align*}

Notice that the wealth equation and preferences are homogeneous in wealth and
consumption. If $(\pi,c)\in\mathcal A(x;S^p,r^p)$ and $\lambda>0$, then
$(\pi,\lambda c)\in\mathcal A(\lambda x;S^p,r^p)$ and
$X^{p,\lambda x,\pi,\lambda c}=\lambda X^{p,x,\pi,c}$. Hence the investor's
value function is non-decreasing in $x$. It is non-negative when
$\psi>1$ and non-positive when $\psi<1$. In the smooth homothetic class
considered below, the relevant range is $z=v_x>0$, and $(1-\gamma)y$ is
strictly positive. For $z>0$, the consumption supremum is attained at
\[
c^\star(y,z)
\coloneqq
\bigg(\frac{\delta((1-\gamma)y)^{1-1/\theta}}{z}\bigg)^\psi.
\]
When $m<0$ and the risky asset has non-zero return volatility, the portfolio supremum is strictly concave and its maximiser is the projection of the unconstrained optimum onto $[0,1]$. When $m\geq0$, a maximiser lies at an endpoint unless the objective is constant in $\pi$. The proof below uses only the one-sided optimality condition at the clearing portfolio $\pi=1$.

\medskip
In the remainder of the paper, we say that the representative
agent's value function admits a classical $C^2$ state representation if it is
represented by a function
$v\in C^2\bigl(\Delta^n\times(0,+\infty)^2;\mathbb R\bigr)$ that solves
\eqref{eq:hjb} pointwise and whose Hamiltonian is attained by the
market-clearing controls along the equilibrium state process.

\medskip
The next lemma expresses the equilibrium restrictions on the auxiliary drift
and the short rate. To state those restrictions compactly, let
\[
  \kappa
  \coloneqq
  \delta\theta
  +\frac{\gamma(1-\gamma)\sigma^2}{2},
\]
set
\begin{gather}
\label{eq:zeta-definition}
\zeta(p)
\coloneqq
  \frac{\kappa}{\theta}-\frac{1}{2\sigma^2}
  \Bigg(
    \frac{A(p)^\top\partial_{pp}\varphi(p)A(p)}{\varphi(p)}
    +(\theta-1)
    \bigg(
      \frac{A(p)^\top\partial_p\varphi(p)}{\varphi(p)}
    \bigg)^2
  \Bigg)-
  \frac{\big(b(p)-(1-\gamma)A(p)\big)^\top\partial_p\varphi(p)}{\varphi(p)}-
  \bigg(1-\frac1\psi\bigg)\overline g(p),\\
  \notag
\overline r(p)
\coloneqq
  \frac{\overline g(p)}{\psi}
  +(\theta-1)
   \bigg(
    \frac{1}{2\sigma^2}
    \bigg(
      \frac{A(p)^\top \partial_p\varphi(p)}{\varphi(p)}
    \bigg)^2
    -
    \frac{A(p)^\top \partial_p\varphi(p)}{\varphi(p)}
   \bigg)
  +\frac{\kappa}{\theta}
  -\gamma\sigma^2.
\end{gather}

Both $\zeta$ and $\overline r$ depend on the candidate function $\varphi$, though we suppress this dependence to alleviate the notation.

\begin{lemma}[Equilibrium reduction]
\label{lem:equilibrium-reduction}
Let the specification be a $\mathfrak C$-equilibrium for
which the representative agent's value function admits a classical $C^2$
state representation. Then,
for every initial belief $p\in\Delta^n$
\begin{equation}
\label{eq:on-path-equilibrium-reduction}
  \mathcal H(P_t^p,H_t^p)
  =
  H_t^p\zeta(P_t^p)-\varphi(P_t^p)^{-1},\;
  \mathcal R(P_t^p,H_t^p)
  \leq
  \overline r(P_t^p),
\;
\mathrm dt\otimes\mathbb P^p\text{\rm--a.e.}
\end{equation}
\end{lemma}

The proof is given in \Cref{app:proof-main-equilibrium}.

\section{Equilibrium characterisation}
\label{sec:equilibrium-characterisation}

\subsection{Main result}

We use a verification argument to characterise the $\mathfrak C$-equilibria
for which the representative agent's value function admits a classical
$C^2$ state representation. The class $\mathfrak C$ allows candidate
prices to depend on the auxiliary state $H^p$ and therefore
need not be Markovian in dividends and beliefs.
The central result of the paper is that this additional
state cannot survive in such an equilibrium: every such
$\mathfrak C$-equilibrium has a constant auxiliary process
$H^p$.
Consequently,
$Q^p=h_{\smallertext{0}}\varphi(P^p)$ and the stock price is Markovian in
$(D^p,P^p)$. Thus Markovianity is an equilibrium conclusion, not a
restriction imposed on candidate prices. The proof of the next result is given in \Cref{app:proof-main-theorem}.

\begin{theorem}[Equilibrium characterisation]
\label{thm:main}
Let {\rm\Cref{assum:pref}} hold. If a specification is a $\mathfrak C$-equilibrium for which the agent's value function admits a classical
$C^2$ state representation and
$\min_{p\in\Delta^\smalltext{n}}\zeta(p)>0$, then there is a constant $h_{\smallertext{0}}>0$ such that
\begin{equation}\label{eq:equilibrium-identity-short}
  h(p)=h_{\smallertext{0}},
  \; p\in\Delta^n,\; \text{\rm and}\;
  h_{\smallertext{0}}\varphi(p)\zeta(p)=1,
  \; p\in\Delta^n.
\end{equation}
For this $\mathfrak C$-equilibrium, for every
$p\in\Delta^n$, on the $p$-model copy
\begin{equation}
\label{eq:belief-markovian-conclusion}
  H_t^p=h_{\smallertext{0}},
  \;
  Q_t^p
  =h_{\smallertext{0}}\varphi(P_t^p),
  \;
  S_t^p
  =D_t^pQ_t^p,
  \;
  r_t^p
  =\mathcal R(P_t^p,h_{\smallertext{0}}),
  \; t\geq0,
  \;\mathbb P^p\text{\rm--a.s.}
\end{equation}
Conversely, suppose that there is a constant $h_{\smallertext{0}}>0$ satisfying
\eqref{eq:equilibrium-identity-short}, and that the specification in
$\mathfrak C$ satisfies
\begin{equation}
\label{eq:converse-conditions}
  h(p)=h_{\smallertext{0}},
  \;
  \mathcal H(p,h_{\smallertext{0}})=0,
  \;
  \mathcal R(p,h_{\smallertext{0}})\leq\overline r(p),
  \; p\in\Delta^n.
\end{equation}
Then the specification is a $\mathfrak C$-equilibrium and satisfies
\eqref{eq:belief-markovian-conclusion}.
\end{theorem}

Before we continue, let us make two short remarks on \Cref{thm:main}.

\begin{remark}[Role of the strict lower bound on $\zeta$]
\label{rem:two-state-endogenous-zeta}
For an arbitrary number of states, the condition
$\min_{p\in\Delta^\smalltext{n}}\zeta(p)>0$ provides the uniform tail estimate used in
the proof. In the two-state model with
$\lambda^{\smallertext{1}\smallertext{2}}>0$ and
$\lambda^{\smallertext{2}\smallertext{1}}>0$, however, this condition
follows from the positivity requirement on $H$ and need not be imposed
separately. Hence the necessity conclusion of {\rm\Cref{thm:main}} remains valid
without it in that case. A formal statement and proof are provided in
{\rm\Cref{prop:two-state-endogenous-zeta}}.
\end{remark} 

\begin{remark}[Initial wealth and scaled equilibria]
By {\rm\Cref{def:p-equilibrium}}, the representative agent's equilibrium problem
is evaluated at $x=S_{\smallertext{0}}^p$. This initial wealth is the value of the single
outstanding stock share. With this wealth, the agent holds that share and
consumes the dividend endowment $D^p$, so the original economy clears.

\medskip

Homogeneity also gives the solution of the investor's problem for any
initial wealth $x>0$. The value is
\begin{equation}
\label{eq:value-short}
  V_{\smallertext{0}}^p(x)
  =
  \delta^\theta
  \frac{x^{1-\gamma}}{1-\gamma}
  \big(h_{\smallertext{0}}\varphi(p)\big)^{\theta/\psi},
\end{equation}
and the optimal controls and associated wealth process are
\begin{equation}
\label{eq:controls-short}
  \pi_t^\star=1,
  \;
  X_t^\star=\frac{x}{S_{\smallertext{0}}^p}S_t^p,
  \;
  c_t^\star=\frac{x}{S_{\smallertext{0}}^p}D_t^p,\; t\geq0,\; \mathbb P^p\text{\rm--a.s.}
\end{equation}

If $x\neq S_{\smallertext{0}}^p$, this strategy and its associated wealth do not clear the original economy, and the
agent's wealth and consumption are scaled by $x/S_{\smallertext{0}}^p$. They can instead be
interpreted as the clearing allocation of an otherwise identical economy in
which the stock price and aggregate dividend endowment are both multiplied
by $x/S_{\smallertext{0}}^p$. 
\end{remark}

\subsection{Equilibrium stability}
\label{subsec:uniform-valuation-bounds}

The auxiliary factor was introduced to allow prices outside the
belief-Markovian class. We now introduce valuation stability. It is satisfied
by the equilibria covered by \Cref{thm:main} and violated by the non-Markovian
formal candidates below.

\begin{definition}[Valuation stability]
\label{def:valuation-stability}
A specification in $\mathfrak C$ is \emph{valuation-stable} if, for every initial
belief $p\in\Delta^n$
\[
  \sup_{t\geq0}\mathbb E^p[Q_t^p]<+\infty,
  \;
  \sup_{t\geq0}\mathbb E^p[(Q_t^p)^{-1}]<+\infty.
\]
\end{definition}

The main theorem immediately implies valuation stability for the equilibria
it covers.

\begin{corollary}[Valuation stability]
\label{cor:C-equilibrium-stability}
Every $\mathfrak C$-equilibrium covered by {\rm\Cref{thm:main}} is
valuation-stable.
\end{corollary}

\begin{proof}
\Cref{thm:main} gives
$Q_t^p=h_{\smallertext{0}}\varphi(P_t^p)$ for every $t\geq0$. Since $\varphi$ is positive and continuous on the
compact simplex
\[
  \underline\varphi
  \coloneqq \min_{y\in\Delta^n}\varphi(y),
  \;
  \overline\varphi
  \coloneqq \max_{y\in\Delta^n}\varphi(y),
  \;
  0<\underline\varphi\leq\overline\varphi<+\infty.
\]
It follows that, for every initial belief $p$ and every $t\geq 0$
\[
  h_{\smallertext{0}}\underline\varphi
  \leq Q_t^p
  \leq h_{\smallertext{0}}\overline\varphi,
  \;
  \frac{1}{h_{\smallertext{0}}\overline\varphi}
  \leq (Q_t^p)^{-1}
  \leq \frac{1}{h_{\smallertext{0}}\underline\varphi},
  \; \mathbb P^p\text{\rm--a.s.}
\]
These pathwise bounds imply both conditions in
\Cref{def:valuation-stability}.
\end{proof}

\begin{remark}[Formal non-Markovian HJB candidates]
\label{rem:formal-unstable-candidates}
The {\rm HJB} equation is a local condition and admits candidates that are not
infinite-horizon equilibria. To see this, let
\[
  \underline\varphi
  \coloneqq\min_{p\in\Delta^\smalltext{n}}\varphi(p),
  \;
  \underline\zeta
  \coloneqq\min_{p\in\Delta^\smalltext{n}}\zeta(p) > 0,
\]
and choose $a>(\underline\varphi\underline\zeta)^{-1}$. Set
\[
  h(p)=a,
  \;
  \mathcal H(p,u)=u\zeta(p)-\varphi(p)^{-1},
  \; p\in\Delta^n.
\]
If $\mathcal R(p,u)\leq\overline r(p)$ for every
$(p,u)\in\Delta^n\times(0,+\infty)$, the function
\[
  w(p,x,u)
  =
  \delta^\theta\frac{x^{1-\gamma}}{1-\gamma}
  \big(u\varphi(p)\big)^{\theta/\psi},\; (p,x,u)\in\Delta^n\times (0,+\infty)^2,
\]
is a classical solution of the formal {\rm HJB} equation \eqref{eq:hjb}, with
pointwise optimal controls
\[
  \pi^\star=1,
  \;
  c^\star=\frac{x}{u\varphi(p)}.
\]
This is the same substitution as in the converse verification in
{\rm\Cref{app:proof-main-theorem}}, now applied at arbitrary $u>0$. At a
market state $x=S=D u\varphi(p)$, the controls give $\pi^\star=1$ and
$c^\star=D$, and hence clear pointwise.

\medskip
The parameter $a$ does not enter the {\rm HJB} calculation, but it determines the
auxiliary process. If
$\Gamma_\cdot^p\coloneqq\int_{\smallertext{0}}^\cdot\zeta(P_s^p)\mathrm ds$, then
\[
  H_t^p
  =
  \mathrm e^{\smallertext{\Gamma}_\smalltext{t}^\smalltext{p}}
  \bigg(
    a-
    \int_{\smallertext{0}}^t
      \mathrm e^{-\smallertext{\Gamma}_\smalltext{s}^\smalltext{p}}
      \varphi(P_s^p)^{-1}\mathrm ds
  \bigg),\; t\geq 0,\; \mathbb{P}^p\text{\rm--a.s.}
\]
Let
$\varepsilon_a\coloneqq
a-(\underline\varphi\underline\zeta)^{-1}>0$. Since for any $s\geq 0$ we have 
$\Gamma_s^p\geq\underline\zeta s$ and
$\varphi(P_s^p)^{-1}\leq\underline\varphi^{-1}$, we deduce
\[
  \int_{\smallertext{0}}^t
    \mathrm e^{-\smallertext{\Gamma}_\smalltext{s}^\smalltext{p}}
    \varphi(P_s^p)^{-1}\mathrm ds
  \leq
  \frac{1}{\underline\varphi}
  \int_{\smallertext{0}}^t\mathrm e^{-\underline\zeta s}\mathrm ds
  \leq
  \frac{1}{\underline\varphi\underline\zeta}, \; t\geq 0,\; \mathbb{P}^p\text{\rm--a.s.}
\]
It follows that
\[
  H_t^p\geq\varepsilon_a\mathrm e^{\underline\zeta t},
  \;
  Q_t^p\geq
  \underline\varphi\varepsilon_a\mathrm e^{\underline\zeta t},
  \; t\geq 0,\; \mathbb{P}^p\text{\rm--a.s.}
\]
In particular,
$\int_{\smallertext{0}}^{\smallertext{+}\infty}
(Q_s^p)^{-1}\mathrm ds<+\infty$. Under {\rm\Cref{def:admissibility}}, the relevant process to check for admissibility is
\[
  \widehat Y_t^p
  \coloneqq
  \delta^\theta\mathrm e^{-\delta\theta t}
  (S_t^p)^{1-\gamma}(Q_t^p)^{\theta/\psi}, \; t\geq 0,\; \mathbb{P}^p\text{\rm--a.s.}
\]
This is exactly the process in {\rm\Cref{def:admissibility}} after
substituting the clearing wealth $X^p=S^p$. Let $\widehat Z^p$ denote the
coefficient of $\mathrm d\overline W^p$ in its It\^o decomposition. A direct
calculation gives
\begin{gather*}
  \lambda_\smallertext{U}(q)
  \coloneqq(1-\gamma)\sigma
  -
  \frac{\theta}{\sigma}
  A(q)^\top\frac{\partial_p\varphi(q)}{\varphi(q)},\\
  \mathrm d\widehat Y_t^p
  =
  -\frac{\theta}{Q_t^p}\widehat Y_t^p\mathrm dt
  +
  \lambda_\smallertext{U}(P_t^p)\widehat Y_t^p
  \mathrm d\overline W_t^p,
  \;
  \widehat Z_t^p
  =
  \lambda_\smallertext{U}(P_t^p)\widehat Y_t^p.
\end{gather*}
Thus $\widehat Z^p$ is determined by the candidate clearing strategy. Let
\[
  \widehat L_t^p
  \coloneqq\mathcal E\bigg(
    \int_{\smallertext{0}}^\cdot
      \lambda_\smallertext{U}(P_s^p)\mathrm d\overline W_s^p
  \bigg)_t, \; t\geq 0,\; \mathbb{P}^p\text{\rm--a.s.}
\]
Since $\lambda_\smallertext{U}$ is bounded, $\widehat L^p$ is a martingale on every finite horizon. Solving the preceding equation gives
\[
  \widehat Y_t^p
  =
  \widehat Y_{\smallertext{0}}^p\widehat L_t^p
  \exp\bigg(
    -\theta\int_{\smallertext{0}}^t(Q_s^p)^{-1}\mathrm ds
  \bigg), \; t\geq 0,\; \mathbb{P}^p\text{\rm--a.s.}
\]
Because the integral is bounded uniformly in $t$
\[
  \inf_{t\geq0}\mathbb E^p\big[\widehat Y_t^p\big]>0.
\]
If both the terminal condition in {\rm\Cref{def:admissibility}} and
$\widehat Z^p\in\mathbb H^1(\mathbb P^p)$ held,
{\rm\Cref{lem:zero-residual}} would instead give
$\widehat Y_t^p\longrightarrow0$ in $\mathbb L^1(\mathbb P^p)$, a
contradiction. Thus these candidates cannot satisfy all the admissibility
conditions. In particular, whenever the terminal condition holds,
$\widehat Z^p\notin\mathbb H^1(\mathbb P^p)$.

\medskip

Independently, they are not valuation-stable. Indeed, the preceding bound directly
gives
\[
  Q_t^p
  \geq
  \underline\varphi\varepsilon_a
  \exp\bigg(
    t\min_{y\in\Delta^\smalltext{n}}\zeta(y)
  \bigg),
  \; t\geq0,\; \mathbb P^p\text{\rm--a.s.},
\]
and therefore $\sup_{t\geq0}\mathbb E^p[Q_t^p]=+\infty$. Thus the {\rm HJB}
equation admits additional formal candidates, but these candidates are neither equilibria nor
valuation-stable. Even under a weakened equilibrium definition, the valuation-stability criterion would select only the Markovian candidates considered here.\end{remark}

A simple two-state simulation illustrates the implications of the
Markovian representation. The simulation sets $\theta=1$ and
uses $g^{\smallertext{1}}=0.25$, $g^{\smallertext{2}}=0$, $\gamma=0.8$,
$\sigma=1.25$,
$\lambda^{\smallertext{1}\smallertext{2}}
=\lambda^{\smallertext{2}\smallertext{1}}=0.04$, and $\delta=0.025$. These
parameters will be used again in {\rm\Cref{subsec:markovian-shape-numerics}}. Under
this calibration, posterior
beliefs are persistent and mean reverting. The resulting equilibrium
price--dividend ratio is bounded, and its conditional mean converges towards
a long-run level; see
\Cref{fig:equilibrium-ratio-mean-reversion}.

\begin{figure}[H]
  \centering
  \includegraphics[width=0.78\textwidth]
    {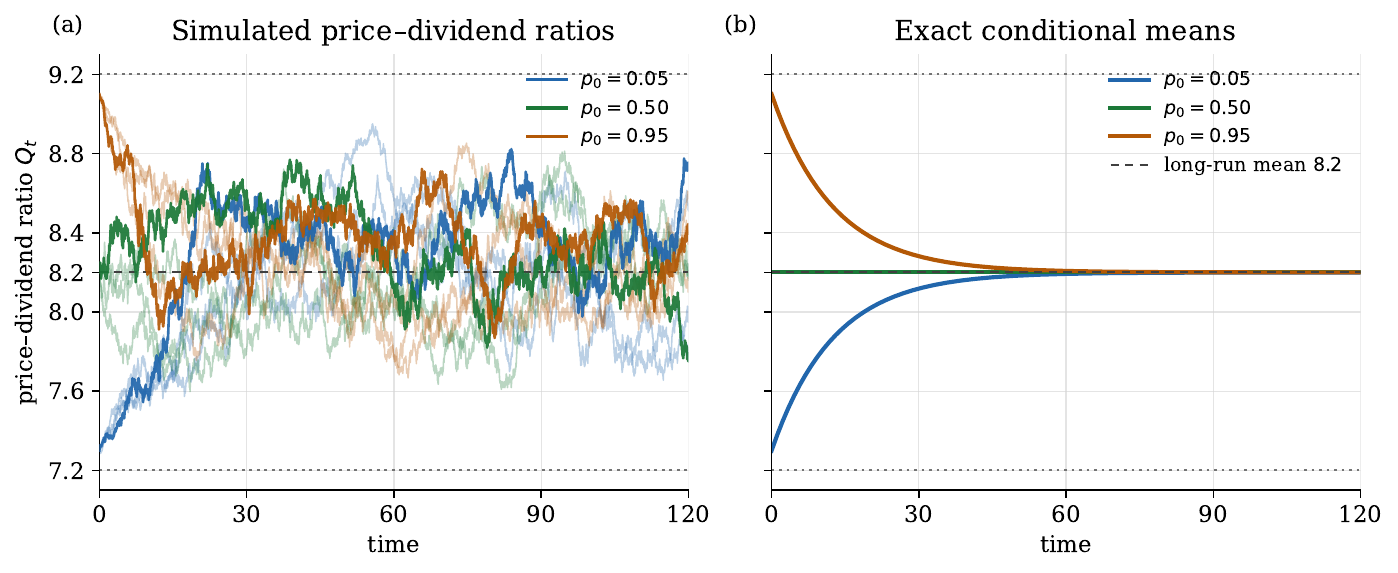}
  \caption{\small Twelve simulated paths of the equilibrium
  ratio in Panel~(a) and the corresponding exact conditional means in
  Panel~(b). The ratio remains in $[7.2,9.2]$.}
  \label{fig:equilibrium-ratio-mean-reversion}
\end{figure}

The qualitative conclusion illustrated in
\Cref{fig:equilibrium-ratio-mean-reversion} is similar to that of the
empirical asset-pricing literature: aggregate valuation ratios display
persistent but mean-reverting variation; see
\citeauthor*{fama1988dividend} \cite{fama1988dividend},
\citeauthor*{campbell1988dividend} \cite{campbell1988dividend}, and
\citeauthor*{cochrane2008dog} \cite{cochrane2008dog}. In our model, this
qualitative pattern arises from learning about the hidden growth state.

\subsection{The Markovian pricing equation on the simplex}

The equilibrium characterisation makes the auxiliary factor constant. Since
$\zeta$ is unchanged when $\varphi$ is multiplied by a positive constant,
without loss of generality we absorb $h_{\smallertext{0}}$ into $\varphi$ in this section and
continue to denote $h_{\smallertext{0}}\varphi$ by $\varphi$. Then
\[
  Q^p=\varphi(P^p),
  \;
  \varphi(p)\zeta(p)=1,
  \; p\in\Delta^n.
\]
We now derive the equation satisfied by this price--dividend ratio. Set
\[
  c(p)
  \coloneqq
  \kappa-(1-\gamma)\overline g(p),
  \; p\in\Delta^n.
\]
Because $\varphi\in C^2(\Delta^n)$, substituting the definition of $\zeta$
into the normalised identity $\varphi\zeta=1$ leads to the following equation in the simplex.

\begin{definition}[Markovian pricing equation]
\label{def:markovian-pricing-equation}
A positive function $\varphi\in C^2(\Delta^n)$ is a solution of the
Markovian pricing equation if, for every $p\in\Delta^n$
\begin{equation}
\label{eq:markovian-simplex-pde}
  -\frac{1}{2\sigma^2}
  \bigg(
    A(p)^\top\partial_{pp}\varphi(p)A(p)
    +(\theta-1)
    \frac{\big(A(p)^\top\partial_p\varphi(p)\big)^2}{\varphi(p)}
  \bigg)
  -\big(b(p)-(1-\gamma)A(p)\big)^\top\partial_p\varphi(p)
  +
  \frac{c(p)}{\theta}\varphi(p)
  =
  1,
  \; p\in\Delta^n.
\end{equation}
\end{definition}

We do not claim well-posedness of this problem for an arbitrary number of
states. However, in \Cref{sec:markovian-shape}, we manage to prove existence and uniqueness in
the two-state case.

\medskip

In viscosity theory, the analogous formulation without separately prescribed
boundary values is called a state-constraint problem; see
\cite[Section~7.C$'$]{crandall1992user}. Here we
use the stronger classical formulation in which the equation holds throughout
the closed simplex. For it to be well defined, it remains to verify that its boundary evaluation involves
only tangent or inward derivatives. This is the content of the next
proposition. Recall
that the boundary of the simplex is
\[
  \partial\Delta^n
  =\bigl\{p\in\Delta^n:\exists j\in\{1,\dots,n\},\; p^j=0\bigr\}.
\]

\begin{proposition}[Boundary derivative directions]
Let $\mu(p)\coloneqq b(p)-(1-\gamma)A(p)$, $p\in\Delta^n$. For every
$p\in\partial\Delta^n$, there exists $\varepsilon>0$ such that
\[
  p+tA(p)\in\Delta^n,\; \text{\rm for}\; |t|<\varepsilon,
  \;
  p+t\mu(p)\in\Delta^n,\;\text{\rm for}\; 0\leq t<\varepsilon.
\]
Thus the $A(p)$-direction is tangent to the smallest face containing $p$,
whereas the $\mu(p)$-direction is tangent or inward.
\end{proposition}

\begin{proof}
Let $J(p)\coloneqq\{j\in\{1,\dots,n\}:p^j=0\}$. For every $j\in J(p)$
\[
  A^j(p)=0,
  \;
  \mu^j(p)=b^j(p)=\sum_{i\in\{1,\dots,n\}\setminus\{j\}}p^i\lambda^{ij}\geq0.
\]
Moreover, $\mathbf 1_n^\top A(p)=\mathbf 1_n^\top\mu(p)=0$.
Hence the displayed paths have coordinates summing to one, and their
coordinates that are positive at $p$ remain positive for sufficiently small
$|t|$. The result follows.
\end{proof}

\begin{remark}\label{rem:no-inductive-sol}
The boundary restrictions are not generally autonomous equations on the
faces. If, for some $p\in\partial\Delta^n$ and some $j\in\{1,\dots,n\}$, $p^j=0$ but $b^j(p)>0$, the equation contains a derivative from
that face into the interior, which cannot be recovered from the restriction
of $\varphi$ to the face. Thus one cannot in general solve first at the
vertices and then proceed successively through the higher-dimensional faces.
\end{remark}

At a vertex $e_i$, we have
\[
  A(e_i)=0,
  \;
  b(e_i)=\Lambda^\top e_i
  =
  \sum_{j\in\{1,\dots,n\}\setminus\{i\}}\lambda^{ij}(e_j-e_i),\; i\in\{1,\dots,n\}.
\]

The natural vertex relation is therefore
\begin{equation}
\label{eq:simplex-vertex-bc}
  -\big(\Lambda^\top e_i\big)^\top\partial_p\varphi(e_i)
  +
  \frac{\kappa-(1-\gamma)g^i}{\theta}\varphi(e_i)
  =
  1,
  \; i\in\{1,\dots,n\}.
\end{equation}

Thus every derivative appearing in
\eqref{eq:simplex-vertex-bc} is taken along an edge entering the simplex. The relation has a direct financial interpretation. At $e_i$, investors are
certain that the current state is $i$, but they still price future switches.
The derivative term therefore links $\varphi(e_i)$ to nearby mixed-belief
valuations, so the vertex value is not an independently prescribed constant.

\medskip

As an example, suppose that $\varphi$ is linear in beliefs, 
\[
  \varphi(p)=\sum_{k=1}^n a_kp^k.
\]
In the two-state model, this is precisely the form obtained when $\theta=1$;
see \Cref{prop:theta-one-affine-solution} in \Cref{sec:markovian-shape}.
Since $a_k=\varphi(e_k)$, we have
\[
  \big(\Lambda^\top e_i\big)^\top\partial_p\varphi(e_i)
  =
  \sum_{j\in\{1,\dots,n\}\setminus\{i\}}\lambda^{ij}(a_j-a_i).
\]
This is the familiar regime-switching coupling appearing in
\citeauthor*{david2002option}
\cite[Proposition~1(a)]{david2002option}, where the price--dividend ratio
is affine in posterior beliefs and the regime-specific valuations satisfy
a coupled linear system. Likewise, the two-state posterior in
\citeauthor*{veronesi1999overreaction}
\cite[Equation~(3)]{veronesi1999overreaction} has vanishing diffusion but
inward drift at the endpoint beliefs.

\medskip
Finally, the nonlinear Epstein--Zin derivative term in
\eqref{eq:markovian-simplex-pde} is removed by the transformation
\[
  \Phi\coloneqq\varphi^\theta,
  \;
  q\coloneqq1-\frac1\theta,
  \;
  \Psi(x)\coloneqq\theta x^q,
  \; x>0.
\]
Then \eqref{eq:markovian-simplex-pde} becomes
\[
  -\frac{1}{2\sigma^2}
  A(p)^\top\partial_{pp}\Phi(p)A(p)
  -\big(b(p)-(1-\gamma)A(p)\big)^\top\partial_p\Phi(p)
  +c(p)\Phi(p)
  =
  \Psi(\Phi(p)),
  \; p\in\Delta^n.
\]

\section{A closer analysis of the Markovian price--dividend ratio for two states}
\label{sec:markovian-shape}

In the previous section, the Markovian condition reduced the characterisation of the price--dividend ratio to a scalar nonlinear boundary-value problem. For simplicity, we study the qualitative properties of its positive solution for the particular case of $n=2$: existence and uniqueness, monotonicity in the posterior belief, and conditional curvature results indexed by the preference-combination parameter~$\theta$.

\medskip
The proofs of all results in this section are presented in
\Cref{app:proof-two-state-analysis}.

\subsection{Setup and notation}
\label{subsec:markovian-shape-setup}

We identify $[0,1]$ with $\Delta^2$ through
\[
  \iota(p)\coloneqq(p,1-p)^\top,\; p\in[0,1],
\]
and, with a slight abuse of notation, write
$\varphi(p)\coloneqq\varphi(\iota(p))$. The equations below follow by
substituting this parametrisation into the simplex equation
\eqref{eq:markovian-simplex-pde} and its vertex relations.

\medskip

Let state $1$ be the high-growth state and state $2$ the low-growth
state.
Under the identification above, $p\in[0,1]$ denotes the first
coordinate of a generic belief vector $\iota(p)$, and hence the posterior
probability assigned to the high-growth state. Set
\[
  \Delta g\coloneqq g^{\smallertext{1}}-g^{\smallertext{2}}>0,\;
  g(p)\coloneqq\overline g(\iota(p))=g^{\smallertext{2}}+p\Delta g,
  \; p\in[0,1].
\]
We also write
\begin{gather*}
  \eta\coloneqq(1-\gamma)\Delta g,\;
  c(p)\coloneqq\kappa-(1-\gamma)g(p),\; p\in[0,1],\\
  A(p)\coloneqq-p(1-p)\Delta g,\;
  \alpha(p)\coloneqq\frac{A(p)^2}{2\sigma^2},\;
  \beta(p)\coloneqq
  \lambda^{\smallertext{2}\smallertext{1}}(1-p)-\lambda^{\smallertext{1}\smallertext{2}}p+\eta p(1-p),\; p\in[0,1].
\end{gather*}

Under this parametrisation, derivatives with respect to $p$ are ordinary
derivatives in the scalar coordinate and are therefore scalar-valued. They
should be distinguished from the $T\Delta^n$-valued simplex gradient used
earlier in the paper.

\begin{assumption}[Two-state shape conditions]
\label{assum:two-state-shape}
Assume that
\begin{equation*}
  c(p)>0,\; p\in[0,1].
\end{equation*}
\end{assumption}

Substitution into \eqref{eq:markovian-simplex-pde}, followed by collecting
the scalar first- and second-derivative terms, gives
\begin{equation}
\label{eq:two-state-varphi-ode}
  -\alpha(p)
  \bigg(
    \varphi^{\prime\prime}(p)
    +
    (\theta-1)\frac{(\varphi^{\prime}(p))^2}{\varphi(p)}
  \bigg)
  -
  \beta(p)\varphi^{\prime}(p)
  +
  \frac{c(p)}{\theta}\varphi(p)
  =
  1,
  \; p\in(0,1).
\end{equation}
Evaluating the same simplex equation at
$\iota(0)=e_{\smallertext{2}}$ and $\iota(1)=e_{\smallertext{1}}$, equivalently applying
\eqref{eq:simplex-vertex-bc}, gives the state-constraint conditions
\begin{equation}
\label{eq:two-state-varphi-boundary}
  -\lambda^{\smallertext{2}\smallertext{1}}\varphi^{\prime}(0)
  +
  \frac{c(0)}{\theta}\varphi(0)
  =
  1,
  \;
  \lambda^{\smallertext{1}\smallertext{2}}\varphi^{\prime}(1)
  +
  \frac{c(1)}{\theta}\varphi(1)
  =
  1.
\end{equation}
If both transition intensities vanish, the endpoints are absorbing and
these relations reduce to
\[
  \varphi(0)=\frac{\theta}{c(0)},\;
  \varphi(1)=\frac{\theta}{c(1)}.
\]
These are the price--dividend ratios of the corresponding constant-state
economies. When switching regimes is possible, the derivative terms account for
the possibility of leaving the current state.

\medskip

Finally, the transformation
\[
  \Phi\coloneqq\varphi^\theta,\;
  q\coloneqq1-\frac1\theta,
\]
turns \eqref{eq:two-state-varphi-ode} into
\begin{equation}
\label{eq:two-state-Phi-ode}
  -\alpha(p)\Phi^{\prime\prime}(p)
  -
  \beta(p)\Phi^{\prime}(p)
  +
  c(p)\Phi(p)
  =
  \theta\Phi(p)^q,
  \; p\in(0,1),
\end{equation}
with boundary conditions
\begin{equation}
\label{eq:two-state-Phi-boundary}
  -\lambda^{\smallertext{2}\smallertext{1}}\Phi^{\prime}(0)
  +
  c(0)\Phi(0)
  =
  \theta\Phi(0)^q,
  \;
  \lambda^{\smallertext{1}\smallertext{2}}\Phi^{\prime}(1)
  +
  c(1)\Phi(1)
  =
  \theta\Phi(1)^q.
\end{equation}
Here $q<1$ for every $\theta>0$, which is the sub-linearity used in the
existence and uniqueness analysis below.

\subsection{Existence and uniqueness}
\label{subsec:markovian-shape-existence}

Set
\[
 c_\smallertext{-}\coloneqq \min_{p\in[0,1]}c(p),
 \;
 c_\smallertext{+}\coloneqq \max_{p\in[0,1]}c(p).
\]

\medskip

The existence proof separates the linear and non-linear parts of the
boundary-value problem. First, \Cref{lem:markovian-linear-resolvent} shows
that the linear equation
\[
 -\alpha u^{\prime\prime}-\beta u^{\prime}+cu=f,
\]
defines a positive compact solution operator $R:f\longmapsto u$. The non-linear
equation is then written as the fixed-point problem
\[
 \Phi=R\big(\theta\Phi^q\big),
\]
and a bounded truncation permits an application of Schauder's fixed-point
theorem. The resulting \emph{a priori} bounds make the truncation inactive.
The separate regularity arguments in
\Cref{lem:linear-regularity-scale} and \Cref{thm:markovian-smoothness} then exploit
the inward boundary drift to bootstrap the solution up to the degenerate
endpoints.

\begin{theorem}[Existence of a positive Markovian solution]
\label{thm:markovian-existence}
Assume $\theta>0$, that {\rm\Cref{assum:two-state-shape}} holds, and that
 $\lambda^{\smallertext{1}\smallertext{2}}>0$ and $\lambda^{\smallertext{2}\smallertext{1}}>0$. Then
\eqref{eq:two-state-Phi-ode}--\eqref{eq:two-state-Phi-boundary} admits a
positive solution $\Phi\in C^1([0,1])\cap C^2((0,1))$ satisfying
\begin{equation}
\label{eq:Phi-existence-bounds}
  \bigg(\frac{\theta}{c_\smallertext{+}}\bigg)^\theta
  \le\Phi(p)\le
  \bigg(\frac{\theta}{c_\smallertext{-}}\bigg)^\theta,
  \; p\in[0,1].
\end{equation}
Consequently, $\varphi=\Phi^{1/\theta}$ is a positive solution of
\eqref{eq:two-state-varphi-ode}--\eqref{eq:two-state-varphi-boundary} and
\[
 \frac{\theta}{c_\smallertext{+}}
 \le\varphi(p)\le
 \frac{\theta}{c_\smallertext{-}},
 \; p\in[0,1].
\]
\end{theorem}

\begin{theorem}[Uniqueness]
\label{thm:markovian-uniqueness}
Assume $\theta>0$, that {\rm\Cref{assum:two-state-shape}} holds, and that
 $\lambda^{\smallertext{1}\smallertext{2}}\geq0$ and $\lambda^{\smallertext{2}\smallertext{1}}\geq0$. Then
\eqref{eq:two-state-varphi-ode}--\eqref{eq:two-state-varphi-boundary} has at
most one positive classical solution
\[
 \varphi\in C^1([0,1])\cap C^2((0,1)).
\]
\end{theorem}

\begin{corollary}[Existence and uniqueness]
\label{cor:markovian-existence-uniqueness}
Assume $\theta>0$, that {\rm\Cref{assum:two-state-shape}} holds, and that
 $\lambda^{\smallertext{1}\smallertext{2}}>0$ and $\lambda^{\smallertext{2}\smallertext{1}}>0$. Then the Markovian boundary-value problem
has exactly one positive classical solution.
\end{corollary}

The solution is in fact smooth up to the endpoints. The proof in
\Cref{app:proof-two-state-analysis} uses the regularity scale in
\Cref{lem:linear-regularity-scale}: the linear resolvent gains one
derivative on the closed interval, while the nonlinear forcing
$\Psi(\Phi)=\theta\Phi^q$ inherits the full regularity of the positive
solution. This gain can therefore be iterated.

\begin{theorem}[Smoothness of the Markovian solution]
\label{thm:markovian-smoothness}
Assume $\theta>0$, {\rm\Cref{assum:two-state-shape}}, 
$\lambda^{\smallertext{1}\smallertext{2}}>0$, and $\lambda^{\smallertext{2}\smallertext{1}}>0$. Then the unique positive solution
$\Phi$ of
\eqref{eq:two-state-Phi-ode}--\eqref{eq:two-state-Phi-boundary} belongs
to $C^\infty([0,1])$ and is real analytic on $(0,1)$.
Consequently,
$\varphi=\Phi^{1/\theta}\in C^\infty([0,1])$ and is real analytic on
$(0,1)$.

\medskip

More precisely, set $\mathcal E_{\smallertext{0}}\coloneqq c\Phi-\Psi(\Phi)$, and, for every
integer $j\in\mathbb N^\star$, define
\begin{equation}
\label{eq:hj-def}
  \mathcal E_j
  \coloneqq
  \mathcal E_{\smallertext{0}}^{(j)}
  -\sum_{i=2}^{j}\binom{j}{i}
    \alpha^{(i)}\Phi^{(j+2-i)}
  -\sum_{i=1}^{j}\binom{j}{i}
    \beta^{(i)}\Phi^{(j+1-i)}.
\end{equation}
The one-sided derivatives at the endpoints satisfy
\begin{equation}
\label{eq:endpoint-hierarchy}
  \lambda^{\smallertext{2}\smallertext{1}}\Phi^{(j+1)}(0)=\mathcal E_j(0),
  \;
  \lambda^{\smallertext{1}\smallertext{2}}\Phi^{(j+1)}(1)=-\mathcal E_j(1),
  \; j\in\mathbb N^\star.
\end{equation}
The order-zero relations are precisely
\eqref{eq:two-state-Phi-boundary}.
\end{theorem}

Under the conditions of the next proposition, the unique solution generates
a $\mathfrak C$-equilibrium whose price--dividend ratio and reciprocal are
uniformly bounded.

\begin{proposition}[Two-state Markovian $\mathfrak C$-equilibrium]
\label{thm:two-state-markovian-equilibrium}
Assume $0<\theta\leq1$, {\rm\Cref{assum:two-state-shape}}, 
$\lambda^{\smallertext{1}\smallertext{2}}>0$, and $\lambda^{\smallertext{2}\smallertext{1}}>0$. Let $\varphi$ be the unique positive
solution in {\rm\Cref{cor:markovian-existence-uniqueness}}, and set
\[
  d(p)\coloneqq A(p)\frac{\varphi^{\prime}(p)}{\varphi(p)},
  \; p\in[0,1].
\]
Suppose the Markovian short rate $r:[0,1]\longrightarrow\mathbb R$ is
continuous and satisfies
\[
  r(p)
  \leq
  \frac{g(p)}{\psi}
  +
  (\theta-1)
  \bigg(
    \frac{d(p)^2}{2\sigma^2}-d(p)
  \bigg)
  +
  \frac{\kappa}{\theta}
  -
  \gamma\sigma^2,
  \; p\in[0,1].
\]
For an initial belief $p\in[0,1]$, write $p_t^p\coloneqq P_t^{\iota(p),\smallertext{1}},
  \; t\geq0.$ Then the processes defined by
\[
  Q_t^{\iota(p)}=\varphi(p_t^p),
  \;
  S_t^{\iota(p)}=D_t^{\iota(p)}Q_t^{\iota(p)},
  \;
  r_t^{\iota(p)}=r(p_t^p),
\]
constitute an $\iota(p)$-equilibrium.
Hence the resulting specification is a $\mathfrak C$-equilibrium. Its
price--dividend ratio and reciprocal are bounded uniformly over time. Its
clearing controls at $x=S_{\smallertext{0}}^{\iota(p)}$ are
\[
\pi^\star=1, \;X^\star=S^{\iota(p)},\; {\rm and}\; c^{\iota(p),\star}=D^{\iota(p)}.
\]
\end{proposition}

\begin{proof}[Proof of \Cref{thm:two-state-markovian-equilibrium}]
By \Cref{thm:markovian-smoothness}, $\varphi$ is positive and smooth on
$[0,1]$. Set $h(y)\equiv1$ and define the corresponding
function on the simplex by
\[
  \widehat\varphi(y)=\varphi(y^1),
  \; y\in\Delta^2.
\]
The two-state equation \eqref{eq:two-state-varphi-ode}, together with its
endpoint conditions, is precisely the restriction of
\eqref{eq:markovian-simplex-pde} to this parametrisation. By the
equivalence preceding \eqref{eq:markovian-simplex-pde}, it therefore gives
\[
  \widehat\varphi(\iota(p))\zeta(\iota(p))=1,
  \; p\in[0,1].
\]
Since $\widehat\varphi$ is positive and continuous, this identity also gives
$\min_{y\in\Delta^2}\zeta(y)>0$.
Choose
\[
  \mathcal H(y,u)=u\zeta(y)-\widehat\varphi(y)^{-1},
  \;
  \mathcal R(y,u)=r(y^1),
  \; (y,u)\in\Delta^2\times(0,+\infty).
\]
The preceding identity gives $\mathcal H(y,1)=0$. Hence the conditions in
\eqref{eq:converse-conditions} hold with $h_{\smallertext{0}}=1$, so the converse part of
\Cref{thm:main} applies. The displayed condition on $r$ is exactly the
required short-rate condition. The equilibrium and clearing claims follow.

\medskip

Since $\varphi$ is positive and continuous on the compact interval
$[0,1]$, both $\varphi$ and $\varphi^{-1}$ are bounded. Uniqueness
of the price--dividend ratio follows from \Cref{cor:markovian-existence-uniqueness}.
\end{proof}

\begin{remark}[The marginal short rate]
\label{rem:short-rate-multiplicity}
The portfolio constraint $\pi\in[0,1]$ places the clearing portfolio $\pi=1$
at the boundary of the feasible set. The upper bound
$\overline r(\iota(p))$ in the proposition is the marginal short rate:
at that rate, the portfolio Hamiltonian has zero derivative at $\pi=1$.
Lower rates remain compatible with the investor holding no risk-free asset.
The inequality is therefore a Kuhn--Tucker condition. For option pricing in
{\rm\Cref{sec:option-pricing}}, we select the marginal rate.
\end{remark}

\begin{remark}[Vanishing transition intensities]
\label{rem:absorbing-endpoint-caution}
The strict positivity of the transition intensities is used only in the
existence proof. The ratio-based uniqueness proof remains valid whenever
a positive classical solution exists.
\end{remark}

\subsection{Monotonicity and curvature}
\label{subsec:markovian-shape-monotonicity}
\label{subsec:markovian-shape-curvature}

We now study the monotonicity and curvature of the price--dividend ratio.

\begin{remark}[Scope when $\theta>1$]
For $\theta>1$, the statements below concern the unique positive solution
of the Markovian pricing equation. They become equilibrium statements only
under an appropriate selection of the recursive-utility solution.
\end{remark}

\begin{theorem}[Monotonicity]
\label{thm:markovian-monotonicity}
Let $\theta>0$, let {\rm\Cref{assum:two-state-shape}} hold, assume
$\lambda^{\smallertext{1}\smallertext{2}}>0$, $\lambda^{\smallertext{2}\smallertext{1}}>0$, and let $\varphi$ be the unique positive
solution from {\rm\Cref{cor:markovian-existence-uniqueness}}. Then
\begin{enumerate}
  \item[$(i)$] if $\eta>0$, $\varphi$ is strictly increasing;
  \item[$(ii)$] if $\eta<0$, $\varphi$ is strictly decreasing.
\end{enumerate}
\end{theorem}

Thus the sign of $\eta=(1-\gamma)(g^{\smallertext{1}}-g^{\smallertext{2}})$ determines the direction in
which beliefs affect the price--dividend ratio. When $\eta>0$, assigning
greater probability to the high-growth state raises the value of future
dividends relative to the current dividend; when $\eta<0$, the
continuation-value effect is reversed.

\medskip

This result is related to, but distinct from, the analysis of
\citeauthor*{veronesi1999overreaction}
\cite[Proposition~2, see also pages~976--978]
{veronesi1999overreaction}.
In that CARA model with additive dividends, the stock price is decomposed
into an affine present-value component and a negative, convex uncertainty
discount, and the resulting equilibrium price is increasing and convex in
the posterior belief. Here the object is the equilibrium price--dividend ratio
under Epstein--Zin preferences, and \Cref{thm:markovian-monotonicity}
provides an analytical characterisation for both possible signs of $\eta$.

\begin{proposition}[Closed-form solution for $\theta=1$]
\label{prop:theta-one-affine-solution}
Suppose $\theta=1$ and let {\rm\Cref{assum:two-state-shape}} hold. Then
the unique positive Markovian solution is
\[
 \varphi(p)=a+bp,
\]
where
\[
  b
  =
  \frac{c(0)-c(1)}
  {c(0)c(1)+c(0)\lambda^{\smallertext{1}\smallertext{2}}+c(1)\lambda^{\smallertext{2}\smallertext{1}}}
  =
  \frac{\eta}
  {c(0)c(1)+c(0)\lambda^{\smallertext{1}\smallertext{2}}+c(1)\lambda^{\smallertext{2}\smallertext{1}}},
  \;
  a=\frac{1+\lambda^{\smallertext{2}\smallertext{1}}b}{c(0)}.
\]
Consequently
\[
  \operatorname{sign}(b)=\operatorname{sign}(\eta).
\]
\end{proposition}

The affine form can be compared directly with
\citeauthor*{david2002option} \cite[Proposition~1(a)]{david2002option}. In their time-additive model, the
price--dividend ratio is the posterior-weighted average of the present values
obtained by conditioning on each possible current hidden state. It is therefore
affine in posterior beliefs.
\Cref{prop:theta-one-affine-solution} recovers this belief dependence in the
two-state, time-separable case.

\medskip

For the curvature result, define
\[
  m_\theta(p)
  \coloneqq
  c(p)-\alpha^{\prime\prime}(p)-2\beta^{\prime}(p)
  -\frac{\theta-1}{\varphi(p)},\; p\in[0,1].
\]
This coefficient is the zero-order term in the differential equation
satisfied by $\Phi^{\prime\prime}$, where
$\Phi=\varphi^\theta$. We now move on to study the curvature of the price-dividend ratio.

\begin{theorem}[Curvature]
\label{thm:markovian-curvature}
Let $\theta>0$, let {\rm\Cref{assum:two-state-shape}} hold, and assume
$\lambda^{\smallertext{1}\smallertext{2}}>0$, $\lambda^{\smallertext{2}\smallertext{1}}>0$, and let $\varphi$ be the unique
positive solution from
{\rm\Cref{cor:markovian-existence-uniqueness}}. Then
\begin{enumerate}
  \item[$(i)$] if $0<\theta<1$ and $m_\theta>0$ on $[0,1]$,
  then $\varphi$ is convex;
  \item[$(ii)$] if $\theta=1$, then $\varphi$ is affine;
  \item[$(iii)$] if $\theta>1$ and $m_\theta>0$ on $[0,1]$,
  then $\varphi$ is concave.
\end{enumerate}
\end{theorem}

Monotonicity determines the direction of the valuation response, while
curvature determines how its magnitude varies with beliefs. In the
increasing and convex regime, $\varphi^{\prime}$ is larger at favourable
beliefs, so bad news in good times has a larger valuation effect than equally
sized good news in bad times. This is the asymmetry highlighted by
\citeauthor*{veronesi1999overreaction}
\cite[pages~976--978]{veronesi1999overreaction}. The present result proves
the corresponding shape for the Epstein--Zin price--dividend ratio under
the stated sufficient condition and also identifies the affine and concave
regimes generated by $\theta$.

\medskip

The condition $m_\theta>0$ depends on the unknown solution through
$1/\varphi$. The following stronger conditions involve only primitive
parameters.

\begin{proposition}[Primitive sufficient curvature conditions]
\label{prop:primitive-curvature-conditions}
Let {\rm\Cref{assum:two-state-shape}} hold, assume $\eta>0$, 
$\lambda^{\smallertext{1}\smallertext{2}}>0$, and $\lambda^{\smallertext{2}\smallertext{1}}>0$. Let $\varphi$ be the unique positive
Markovian solution from
{\rm\Cref{cor:markovian-existence-uniqueness}}. If
$0<\theta<1$ and
\[
 \frac{c(1)}{\theta}+2(\lambda^{\smallertext{1}\smallertext{2}}+\lambda^{\smallertext{2}\smallertext{1}})
 >
 (1-\gamma)(g^{\smallertext{1}}-g^{\smallertext{2}})+\frac{(g^{\smallertext{1}}-g^{\smallertext{2}})^2}{\sigma^2},
\]
then $\varphi$ is convex. If $\theta>1$ and
\[
 \frac{c(0)}{\theta}+2(\lambda^{\smallertext{1}\smallertext{2}}+\lambda^{\smallertext{2}\smallertext{1}})
 >
 2(1-\gamma)(g^{\smallertext{1}}-g^{\smallertext{2}})+\frac{(g^{\smallertext{1}}-g^{\smallertext{2}})^2}{\sigma^2},
\]
then $\varphi$ is concave.
\end{proposition}

\subsection{Numerical illustration}
\label{subsec:markovian-shape-numerics}

We illustrate the three curvature regimes at the level of the Markovian pricing equation using the primitive calibration
\[
  g^{\smallertext{1}} = 0.25,\; g^{\smallertext{2}} = 0,\; \gamma = 0.8,\; \sigma = 1.25,
  \; \lambda^{\smallertext{1}\smallertext{2}} = \lambda^{\smallertext{2}\smallertext{1}} = 0.04,\; \delta=0.025,
\]
and the preference values $\theta\in\{0.25,1,10\}$. Consequently
\[
  \kappa(\theta)=0.125+0.025\theta,
  \;
  c_\theta(p)=0.125+0.025\theta-0.05p.
\]
The corresponding elasticity parameter is
$\psi(\theta)=(1-0.2/\theta)^{-1}>1$.
In particular, $c_\theta(p)>0$ for every
$(p,\theta)\in[0,1]\times[0.25,10]$, and the primitive sufficient
curvature conditions in
\Cref{prop:primitive-curvature-conditions} hold throughout the displayed
parameter range. The values $\theta\leq1$ fall within the equilibrium theory developed above, while the value $\theta=10$ illustrates only the positive solution of the pricing equation.

\medskip
The primitive curvature conditions are sufficient but need not be
necessary. \Cref{fig:theta-both-clean-shape} uses a separate parameter
specification for which both primitive inequalities fail, whereas the
computed coefficient $m_\theta$ remains positive and the predicted convex
and concave profiles persist.

\medskip

\begin{figure}[H]
  \centering
  \includegraphics[width=0.78\textwidth]{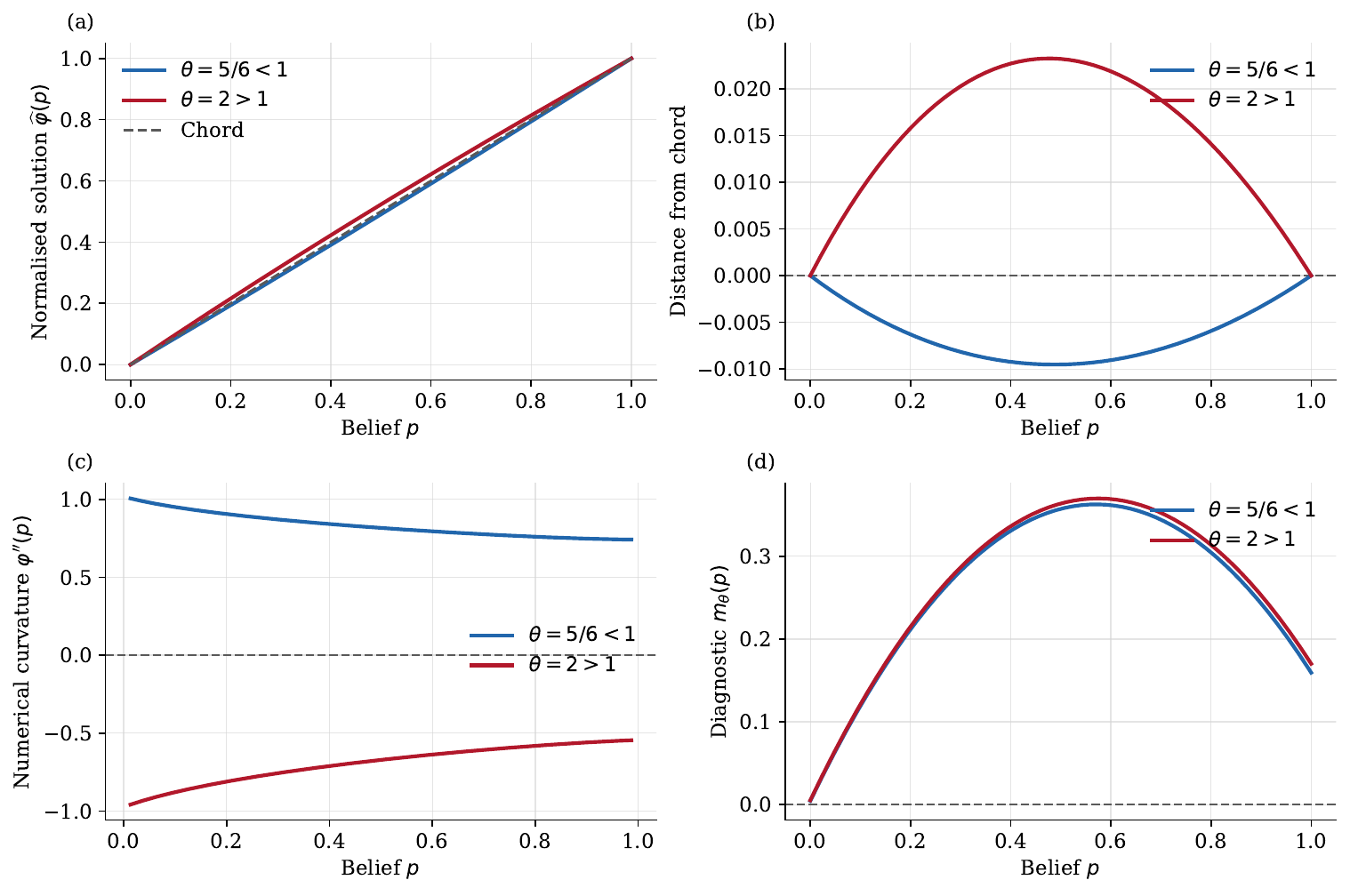}
  \caption{\small
  Finite-difference solutions of
  \eqref{eq:two-state-Phi-ode}--\eqref{eq:two-state-Phi-boundary} for
  $g^{\smallertext{1}}=0.21$, $g^{\smallertext{2}}=0$, $\gamma=0.75$, $\sigma=0.49$,
  $\lambda^{\smallertext{1}\smallertext{2}}=0.052$, $\lambda^{\smallertext{2}\smallertext{1}}=0.059$, and $\delta=0.049$.
  The primitive sufficient inequalities in
  \Cref{prop:primitive-curvature-conditions} fail for both displayed values
  of $\theta$, whereas the computed diagnostic $m_\theta$ remains
  positive. The resulting convex and concave profiles are therefore
  numerical evidence that the primitive conditions are sufficient but not
  necessary.}
  \label{fig:theta-both-clean-shape}
\end{figure}

\medskip

Returning to the baseline specification, each curve below is obtained from
a finite-difference solution of
\Cref{eq:two-state-Phi-ode,eq:two-state-Phi-boundary}. To highlight the shape, we plot the normalised solution
\[
 \varphi_{\mathrm{norm}}(p)
 \coloneqq \frac{\varphi(p)-\varphi(0)}{\varphi(1)-\varphi(0)},
\]
together with the deviation from its chord,
\[
 \Delta_{\mathrm{ch}}(p)
 \coloneqq \frac{\varphi(p)-\big[(1-p)\varphi(0)+p\varphi(1)\big]}
 {\varphi(1)-\varphi(0)}.
\]
Convexity corresponds to $\Delta_{\mathrm{ch}}(p)<0$ (below the chord),
affinity to $\Delta_{\mathrm{ch}}(p)\equiv 0$, and concavity to
$\Delta_{\mathrm{ch}}(p)>0$.

\begin{figure}[H]
  \centering
  \includegraphics[width=\textwidth]
    {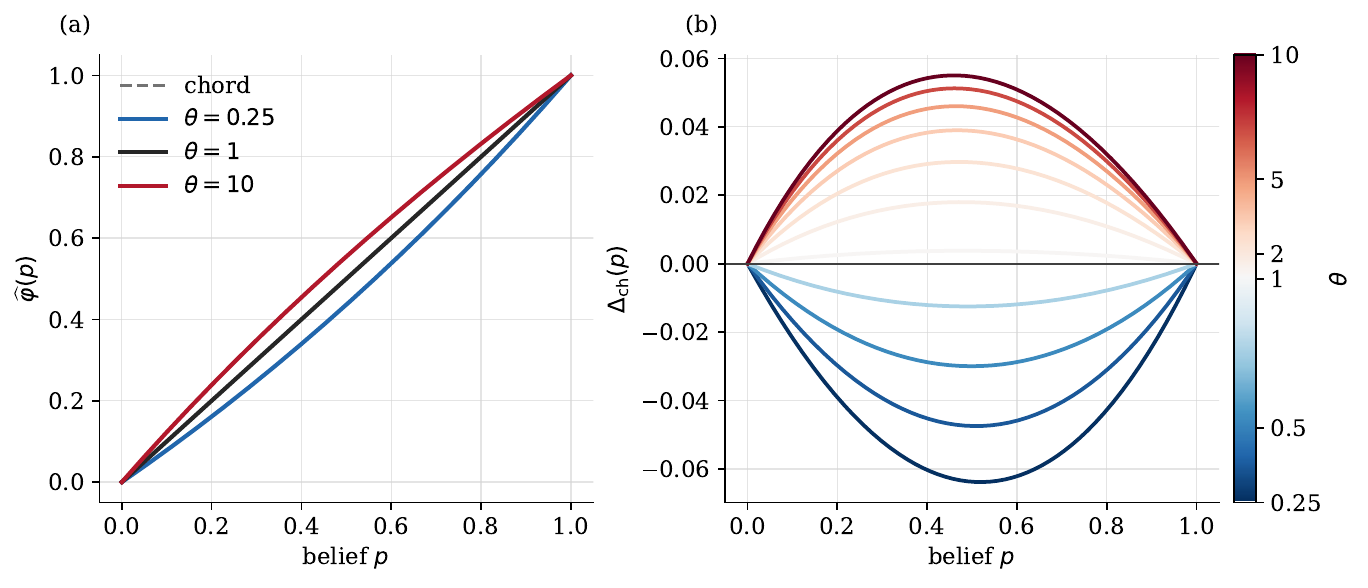}
  \caption{\small Shape of the Markovian price--dividend ratio under the
    calibration of \Cref{subsec:markovian-shape-numerics}.
    Panel~(a) plots the normalised solution
    $\varphi_{\mathrm{norm}}(p) = (\varphi(p)-\varphi(0))/(\varphi(1)-\varphi(0))$
    for the three preference regimes
    $\theta \in \{0.25,1,10\}$, together with the chord
    $p \longmapsto p$ (dashed). Panel~(b) plots the chord deviation
    $\Delta_{\mathrm{ch}}(p) = (\varphi(p) - (1-p)\varphi(0) - p\varphi(1))
    / (\varphi(1) - \varphi(0))$
    for a log-spaced sweep $\theta \in [0.25,10]$, with curves
    coloured by~$\theta$ on a diverging scale centred at the affine
    case $\theta = 1$. Across the full range of preferences, the
    family interpolates smoothly between the convex and concave
    regimes predicted by \Cref{thm:markovian-curvature}.}
  \label{fig:markovian-shape-sweep}
\end{figure}

\section{Option prices and the implied-volatility surface}
\label{sec:option-pricing}

We use the two-state notation introduced in
\Cref{subsec:markovian-shape-setup}. Throughout this section, we maintain the parameter and shape hypotheses of \Cref{thm:two-state-markovian-equilibrium} and additionally assume $\eta>0$, or equivalently $\gamma<1$. Hence $0<\theta\leq1$, $\psi>1$, and $1/\psi\leq\gamma<1$; the option analysis does not cover the high-risk-aversion region $\gamma>1$, $\psi>1$, for which $\theta<0$. We fix an initial belief $p_{\smallertext{0}}\in[0,1]$ and work on the $\iota(p_{\smallertext{0}})$-model copy. To simplify the option-pricing formulas, only in this section we suppress the fixed superscript
$\iota(p_{\smallertext{0}})$ from $D$, $P$, $\overline W$, $Q$, and $S$. We also write
$\mathbb P$, $\mathbb E$, and $\mathbb F^\smallertext{D}$ for the physical
law, expectation, and dividend filtration on that model copy. Along a
filtering path, we write
$p_t=P_t^{\iota(p_\smalltext{0}),\smallertext{1}}$, $t\geq 0$.

\medskip
Throughout this section and
\Cref{app:characteristic-functions}, we set the short rate equal to the
upper bound in \Cref{thm:two-state-markovian-equilibrium}:
\begin{equation}
\label{eq:marginal-short-rate-option}
  r_f(p)\coloneqq\overline r(\iota(p))
  =\frac{g(p)}{\psi}
  +(\theta-1)\bigg(\frac{d(p)^2}{2\sigma^2}-d(p)\bigg)
  +\frac{\kappa}{\theta}-\gamma\sigma^2,
  \; p\in[0,1],
\end{equation}
where $d(p)=A(p)\varphi^{\prime}(p)/\varphi(p)$.
As explained in \Cref{rem:short-rate-multiplicity}, this is the rate at
which the investor is marginally indifferent between the stock and the
risk-free account at the clearing portfolio.

\medskip

We price European options written on the \emph{ex-dividend} stock
\[
 Q_t=\varphi(p_t),
 \;
  S_t=D_tQ_t,
 \; t\geq 0,
\]
where $p_t$ is the posterior probability of the high-growth state.
Here $\varphi$ denotes the unique positive solution of the two-state
Markovian boundary-value problem in
\Cref{cor:markovian-existence-uniqueness}. For $0<\theta\leq1$,
\Cref{thm:two-state-markovian-equilibrium} identifies it with the equilibrium
price--dividend ratio, which is bounded above and away from zero.
Because dividend news also changes beliefs, stock volatility depends on the
current belief even though dividend volatility is constant.

\medskip
This is the same economic channel emphasised by
\citeauthor*{david2002option} \cite{david2002option}: beliefs about hidden
fundamental growth move both prices and future volatility, and this changes
the risk-neutral distribution priced by options.
Option valuation with a hidden Markov state is also studied by
\citeauthor*{guo2001information} \cite{guo2001information}, while
\citeauthor*{buraschi2006model} \cite{buraschi2006model} connect model
uncertainty and heterogeneous beliefs to option prices and trading volume.
In our representative-agent economy, the relevant state is the filtered
growth belief and the stock's state-dependent volatility is generated
endogenously by equilibrium valuation.

\medskip

The equilibrium characterisation and the two-state boundary-value problem
determine the Markovian price--dividend ratio $\varphi$, while preferences
determine its level and belief sensitivity. These effects enter option
prices through $\varphi^{\prime}/\varphi$, which determines the learning
component of stock volatility.

\medskip

Under the physical measure, the scalar belief process satisfies
\begin{equation}
\label{eq:filtering-p-option}
 \mathrm{d}p_t
 =
 B(p_t)\mathrm{d}t
 +
 \chi(p_t)\mathrm{d}\overline W_t,
\end{equation}
with
\[
 B(p)\coloneqq \lambda^{\smallertext{2}\smallertext{1}}(1-p)-\lambda^{\smallertext{1}\smallertext{2}}p,
 \;
 \chi(p)\coloneqq \frac{p(1-p)\Delta g}{\sigma},\; p\in[0,1].
\]
Hence $\chi(p)>0$ for interior beliefs: a positive dividend innovation
raises the posterior probability of the high-growth state. The observed
dividend dynamics are
\begin{equation}
\label{eq:dividend-dynamics-option}
 \frac{\mathrm{d}D_t}{D_t}
 =
 g(p_t)\mathrm{d}t+\sigma\mathrm{d}\overline W_t.
\end{equation}
Equations \eqref{eq:filtering-p-option} and
\eqref{eq:dividend-dynamics-option} are the key inputs for option pricing:
one Brownian innovation moves both cash flows and beliefs.

\subsection{Stock dynamics}

Let $\ell(p)\coloneqq \log(\varphi(p))$. Since $\log S_t=\log D_t+\ell(p_t)$, Itô's formula gives the
following stock dynamics.

\begin{lemma}[Physical stock dynamics]
\label{lem:physical-stock-dynamics}
Under the physical measure, the stock price satisfies
\[
 \frac{\mathrm{d}S_t}{S_t}
 =
 \mu_\smallertext{S}(p_t)\mathrm{d}t
 +\sigma_\smallertext{S}(p_t)\mathrm{d}\overline W_t,
\]
where
\begin{equation}
\label{eq:stock-volatility-option}
 \sigma_\smallertext{S}(p)
 =
 \sigma+\ell^{\prime}(p)\chi(p),
\end{equation}
and
\begin{equation}
\label{eq:stock-drift-option}
  \mu_\smallertext{S}(p)
  =
  g(p)
  +
  \ell^{\prime}(p)B(p)
  +
  \frac12\ell^{\prime\prime}(p)\chi^2(p)
  +
  \frac12\big(\sigma_\smallertext{S}^2(p)-\sigma^2\big).
\end{equation}
%Equivalently,
%\[
% \mu_\smallertext{S}(p)
% =
% g(p)
% +
% \frac{\varphi^{\prime}(p)}{\varphi(p)}B(p)
% +
% \frac12\frac{\varphi^{\prime\prime}(p)}{\varphi(p)}\chi^2(p)
% +
% \sigma\frac{\varphi^{\prime}(p)}{\varphi(p)}\chi(p).
%\]
\end{lemma}

\begin{proof}
From \eqref{eq:dividend-dynamics-option}, we have
\[
 \mathrm{d}(\log D_t)
 =
 \bigg(g(p_t)-\frac12\sigma^2\bigg)\mathrm{d}t
 +
 \sigma\mathrm{d}\overline W_t.
\]
Also
\[
 \mathrm{d}\ell(p_t)
 =
 \bigg(
 \ell^{\prime}(p_t)B(p_t)
 +
 \frac12\ell^{\prime\prime}(p_t)\chi^2(p_t)
 \bigg)\mathrm{d}t
 +
 \ell^{\prime}(p_t)\chi(p_t)\mathrm{d}\overline W_t.
\]
Therefore
\[
 \mathrm{d}(\log S_t)
 =
 \bigg(
 g(p_t)-\frac12\sigma^2
 +
 \ell^{\prime}(p_t)B(p_t)
 +
 \frac12\ell^{\prime\prime}(p_t)\chi^2(p_t)
 \bigg)\mathrm{d}t
 +
 \big(
 \sigma+\ell^{\prime}(p_t)\chi(p_t)
 \big)\mathrm{d}\overline W_t.
\]
Thus the log-stock volatility is $\sigma_\smallertext{S}(p)$ as in
\eqref{eq:stock-volatility-option}. Passing from log dynamics to level
dynamics adds the correction $\frac12\sigma_\smallertext{S}^2(p)$, giving
\eqref{eq:stock-drift-option}.
\end{proof}

By \Cref{subsec:markovian-shape-monotonicity},
$\varphi^{\prime}(p)\geq0$. Hence a positive dividend innovation raises
both current dividends and the price--dividend ratio, and
$\sigma_\smallertext{S}(p)\geq\sigma$, with strict inequality wherever
$\varphi^{\prime}(p)>0$.

\medskip
This is the volatility-amplification mechanism studied by
\citeauthor*{veronesi1999overreaction}
\cite[Section~3, especially pages~988--991]{veronesi1999overreaction} and
\citeauthor*{david2002option}
\cite[equation~(10) and pages~10--14]{david2002option}: dividend news changes
both cash flows and posterior beliefs, and the two effects are driven by the
same innovation. In the present model, the belief-revision component of
stock volatility is
$\chi(p)\varphi^{\prime}(p)/\varphi(p)$, whose magnitude is determined by
the equilibrium Epstein--Zin price--dividend ratio.

\medskip
The learning premium is belief-dependent. Since
$\chi(0)=\chi(1)=0$, it vanishes at the degenerate beliefs, and therefore
$\sigma_\smallertext{S}(0)=\sigma_\smallertext{S}(1)=\sigma$. In the
interior, its magnitude is determined jointly by
$\chi(p)$ and $\varphi^{\prime}(p)/\varphi(p)$.

\begin{figure}[ht!]
  \centering
  \includegraphics[width=0.68\textwidth]{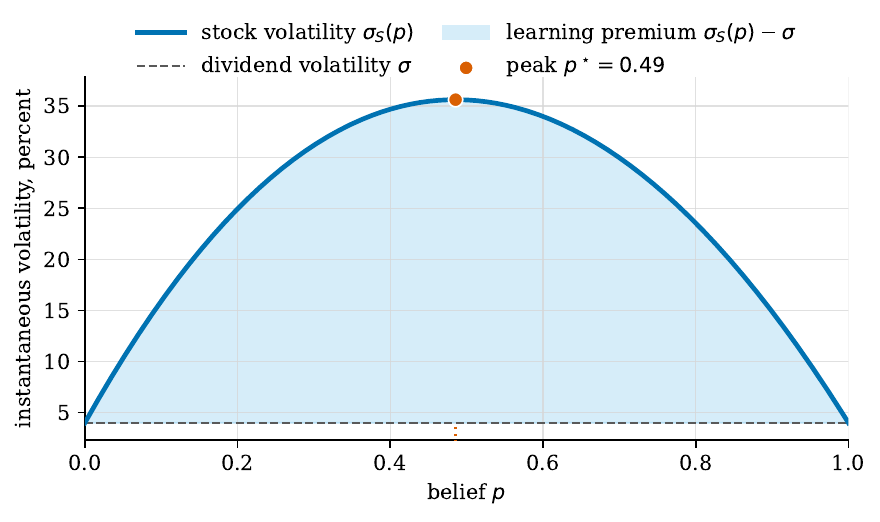}
  \caption{
  Instantaneous stock return volatility as a function of
  the posterior belief under the equilibrium calibration of
  \Cref{subsec:option-numerics}. The dashed line is the dividend volatility
  $\sigma$, and the shaded area is the learning premium
  $\sigma_\smallertext{S}(p)-\sigma$.
  }
  \label{fig:stock-volatility-belief}
\end{figure}

\subsection{Risk-neutral dynamics}

Because $S$ is the \emph{ex-dividend} stock price, the stockholder receives the
dividend flow $D_t\mathrm{d}t$. The instantaneous dividend yield is
\begin{equation*}
 q_\smallertext{S}(p_t)
 \coloneqq 
 Q_t^{-1}
 =\varphi(p_t)^{-1},\; t\geq 0.
\end{equation*}
With $r_f$ fixed by
\eqref{eq:marginal-short-rate-option}, the total return is
\[
 \frac{\mathrm{d}S_t+D_t\mathrm{d}t}{S_t}
 =
 \big(
 \mu_\smallertext{S}(p_t)+q_\smallertext{S}(p_t)
 \big)\mathrm{d}t
 +
 \sigma_\smallertext{S}(p_t)\mathrm{d}\overline W_t.
\]
Since $\sigma_\smallertext{S}(p)\geq\sigma>0$, we can define
\begin{equation}
\label{eq:market-price-risk-option}
 \lambda(p)
 \coloneqq
 \frac{
 \mu_\smallertext{S}(p)+q_\smallertext{S}(p)-r_f(p)
 }{\sigma_\smallertext{S}(p)}
 =\gamma\sigma+(1-\theta)\chi(p)\ell^{\prime}(p),
 \; p\in[0,1].
\end{equation}
The last equality follows by substituting
\eqref{eq:marginal-short-rate-option} and using
\eqref{eq:two-state-varphi-ode}, with $d(p)=-\sigma\chi(p)\ell^{\prime}(p)$. Define the density process
\[
 L_t
 \coloneqq
 \mathcal E\bigg(
 -\int_{\smallertext{0}}^t \lambda(p_s)\mathrm{d}\overline W_s
 \bigg),
 \; t\geq0.
\]

\begin{assumption}[Finite-horizon measure change]
\label{assum:risk-neutral-martingale}
For every $T>0$, $(L_t)_{0\leq t\leq \smallertext{T}}$ is a true martingale.
\end{assumption}

This condition holds, for example, if $\lambda$ is bounded on $[0,1]$:
Novikov's condition then holds on every finite horizon. Under
\Cref{assum:risk-neutral-martingale}, for each $T>0$ we define the
corresponding risk-neutral measure on
$\mathcal F_\smallertext{T}^\smallertext{D}$ by
\[
\frac{\mathrm{d}\mathbb Q^\smallertext{T}}{\mathrm{d}\mathbb P}
 \bigg|_{\mathcal F_\smalltext{T}^\smalltext{D}}
 =
 L_\smallertext{T},
\]
and the risk-neutral Brownian motion by
$
 \mathrm{d}\widetilde W_t
 \coloneqq
 \mathrm{d}\overline W_t+\lambda(p_t)\mathrm{d}t.
$
For a fixed option maturity, we suppress the superscript $T$ and write
$\mathbb Q$.

\medskip

By \Cref{thm:markovian-smoothness}, $\varphi$ is positive
and smooth on $[0,1]$. The selected short rate in
\eqref{eq:marginal-short-rate-option} is therefore smooth and bounded,
and \eqref{eq:market-price-risk-option} makes $\lambda$ bounded as well.
Thus \Cref{assum:risk-neutral-martingale} follows from Novikov's condition
on every finite horizon.

\medskip

Then, by the standard martingale-pricing representation (see for instance \citeauthor*{shreve2004stochastic}
\cite[Section~5.5.1]{shreve2004stochastic}),
$\mathbb P$ remains the physical measure used for filtering and for the
equilibrium construction, while $\mathbb Q$ is the equivalent martingale measure defined above through $\lambda$. In
particular, the discounted cumulative-dividend gain process
\[
 \frac{S_t}{S_t^o}
 +
 \int_{\smallertext{0}}^t \frac{D_s}{S_s^o}\mathrm{d}s,
\]
is a $\mathbb Q$-martingale; equivalently, under $\mathbb Q$ the total
expected return of the stock is $r_f(p_t)$.
 
\begin{proposition}[Risk-neutral dynamics]
\label{prop:risk-neutral-dynamics-option}
Under {\rm\Cref{assum:risk-neutral-martingale}} and the corresponding
risk-neutral measure $\mathbb Q$, the state variables
$(S,p)$ have dynamics
\begin{gather}
\label{eq:risk-neutral-stock-level}
 \frac{\mathrm{d}S_t}{S_t}
 =
 \big(
 r_f(p_t)-q_\smallertext{S}(p_t)
 \big)\mathrm{d}t
 +
 \sigma_\smallertext{S}(p_t)\mathrm{d}\widetilde W_t,\\
\label{eq:risk-neutral-belief}
 \mathrm{d}p_t
 =
 \widetilde B(p_t)\mathrm{d}t
 +
 \chi(p_t)\mathrm{d}\widetilde W_t,
\end{gather}
where
\[
  \widetilde B(p)\coloneqq B(p)-\chi(p)\lambda(p)
  =B(p)-\gamma\sigma\chi(p)-(1-\theta)\chi(p)^2\ell^{\prime}(p),
  \; p\in[0,1].
\]
\end{proposition}

\begin{proof}
Using the definition of $\lambda(p)$,
\[
  \mu_\smallertext{S}(p)+q_\smallertext{S}(p)-\sigma_\smallertext{S}(p)\lambda(p)
  =
  r_f(p), \; p\in[0,1].
\]
Therefore the \emph{ex-dividend} stock drift under $\mathbb Q$ is
\[
 \mu_\smallertext{S}(p)-\sigma_\smallertext{S}(p)\lambda(p)
 =
 r_f(p)-q_\smallertext{S}(p),\; p\in[0,1].
\]
This gives \eqref{eq:risk-neutral-stock-level}. Similarly, substituting
$\mathrm{d}\overline W_t=\mathrm{d}\widetilde W_t-\lambda(p_t)\mathrm{d}t$ into
\eqref{eq:filtering-p-option} gives
\[
 \mathrm{d}p_t
 =
 \big(
 B(p_t)-\chi(p_t)\lambda(p_t)
 \big)\mathrm{d}t
 +
 \chi(p_t)\mathrm{d}\widetilde W_t.
\]
This is \eqref{eq:risk-neutral-belief}.
\end{proof}

\medskip

The measure change changes the drift of the belief process, but not its
diffusion coefficient. This is important for the option-pricing PDE:
under $\mathbb Q$, stock returns and belief revisions are still driven
by the same Brownian innovation. Writing $X_t\coloneqq \log S_t$, the
risk-neutral log-stock dynamics are
\begin{equation*}
 \mathrm{d}X_t
 =
 b_\smallertext{X}(p_t)\mathrm{d}t
 +
 \sigma_\smallertext{S}(p_t)\mathrm{d}\widetilde W_t,
\end{equation*}
where
\begin{equation*}
 b_\smallertext{X}(p)
 \coloneqq 
 r_f(p)-q_\smallertext{S}(p)-\frac12\sigma_\smallertext{S}^2(p), \; p\in[0,1].
\end{equation*}

\subsection{European option prices}

Let $C^{\mathrm{lev}}(S,p,\tau;K)$ be the time-$t$ price of a European call option
with time to maturity $\tau=T-t$, strike $K$, and payoff
$(S_T-K)^{\smallertext{+}}$. In the remainder of this section, $C$ denotes the same surface
in log-stock coordinates:
\[
 C(X,p,\tau;K)\coloneqq
 C^{\mathrm{lev}}(\mathrm{e}^\smallertext{X},p,\tau;K).
\]
Under \Cref{assum:risk-neutral-martingale}, the arbitrage-free price under
the corresponding measure $\mathbb Q$ is
\begin{equation}
\label{eq:call-price-expectation}
 C(X,p,\tau;K)
 =
 \mathbb E^{\mathbb Q}_{X,p}
 \Bigg[
 \exp\bigg(
 -\int_{\smallertext{0}}^\tau r_f(p_s)\mathrm{d}s
 \bigg)
 \big(\mathrm{e}^{\smallertext{X}_\smalltext{\tau}}-K\big)^{\smallertext{+}}
 \Bigg].
\end{equation}
Here $\mathbb E^{\mathbb Q}_{X,p}$ denotes expectation for the
risk-neutral Markov process $(X_t,p_t)$ started from $(X,p)$ at time zero. Equivalently, in terms of the physical measure and the stochastic discount
factor
\[
 M_t
 \coloneqq
 \exp\bigg(-\int_{\smallertext{0}}^t r_f(p_s)\mathrm{d}s\bigg)L_t,\; t\in[0,T],
\]
the same price is
\[
 C_t\coloneqq C(X_t,p_t,T-t;K)
 =
 \frac{1}{M_t}
 \mathbb E
 \big[
 M_\smallertext{T}(S_\smallertext{T}-K)^{\smallertext{+}}
 \big|\mathcal F_t^\smallertext{D}
 \big],\; t\in[0,T].
\]
Hence the measure change does not remove discounting: the discount factor is
explicit in \eqref{eq:call-price-expectation}, and the dividend yield enters
through the \emph{ex-dividend} stock drift $r_f-q_\smallertext{S}$.

\begin{theorem}[Option-pricing PDE]
\label{thm:option-PDE}
Under {\rm\Cref{assum:risk-neutral-martingale}}, let $C$ be given by
\eqref{eq:call-price-expectation}. If, for each $K$, the map
\[
(0,+\infty)\times\mathbb R\times(0,1)\ni  (\tau,X,p)\longmapsto C(X,p,\tau;K)\in\mathbb R,
\]
belongs to
$C^{1,2}\big((0,+\infty)\times\mathbb R\times(0,1)\big)$, then it satisfies, for
$X\in\mathbb R$,
$p\in(0,1)$, and $\tau>0$
\begin{equation}
\label{eq:option-PDE}
\begin{aligned}
  \partial_\tau C
  &=
  \frac12\sigma_\smallertext{S}^2(p)\partial_{\smallertext{X}\smallertext{X}}C
  +
  \sigma_\smallertext{S}(p)\chi(p)\partial_{Xp}C
  +
  \frac12\chi^2(p)\partial_{pp}C
  +
  \bigg(
    r_f(p)-q_\smallertext{S}(p)-\frac12\sigma_\smallertext{S}^2(p)
  \bigg)\partial_\smallertext{X} C
  +
  \widetilde B(p)\partial_p C
  -
  r_f(p)C,
\end{aligned}
\end{equation}
with boundary condition
\[
  C(X,p,0;K)=(\mathrm{e}^\smallertext{X}-K)^{\smallertext{+}}.
\]

If additionally
\[
  \chi(p)\partial_{\smallertext{X}p}C(X,p,\tau;K)\longrightarrow0,\;
  \chi^2(p)\partial_{pp}C(X,p,\tau;K)\longrightarrow0,\; \text{\rm as $p\downarrow0$ and as $p\uparrow1$,}
\]  
then we have
\begin{align*}
  \partial_\tau C(X,0,\tau;K)
  &=
  \frac12\sigma_\smallertext{S}^2(0)\partial_{\smallertext{X}\smallertext{X}}C(X,0,\tau;K)
  +
  \bigg(
    r_f(0)-q_\smallertext{S}(0)-\frac12\sigma_\smallertext{S}^2(0)
  \bigg)\partial_\smallertext{X} C(X,0,\tau;K)
  +
  \lambda^{\smallertext{2}\smallertext{1}}\partial_p C(X,0,\tau;K)\\
  &\quad
  -
  r_f(0)C(X,0,\tau;K),
  \\
  \partial_\tau C(X,1,\tau;K)
  &=
  \frac12\sigma_\smallertext{S}^2(1)\partial_{\smallertext{X}\smallertext{X}}C(X,1,\tau;K)
  +
  \bigg(
    r_f(1)-q_\smallertext{S}(1)-\frac12\sigma_\smallertext{S}^2(1)
  \bigg)\partial_\smallertext{X} C(X,1,\tau;K)
  -
  \lambda^{\smallertext{1}\smallertext{2}}\partial_p C(X,1,\tau;K)\\
&\quad  -
  r_f(1)C(X,1,\tau;K).
\end{align*}
\end{theorem}

\begin{proof}
The generator of the two-dimensional Markov process $(X,p)$ under
$\mathbb Q$ is
\[
\begin{aligned}
  \mathcal L^{\mathbb Q}F
  &=
  b_\smallertext{X}(p)\partial_X F
  +
  \widetilde B(p)\partial_p F
  +
  \frac12\sigma_\smallertext{S}^2(p)\partial_{XX}F
  +
  \sigma_\smallertext{S}(p)\chi(p)\partial_{Xp}F
  +
  \frac12\chi^2(p)\partial_{pp}F.
\end{aligned}
\]
Under the regularity stated in the theorem, the Feynman--Kac calculation
applied to \eqref{eq:call-price-expectation} gives
\[
  \partial_\tau C
  =
  \mathcal L^{\mathbb Q}C-r_f(p)C,
\]
which yields \eqref{eq:option-PDE}. At the endpoints
\[
  \chi(0)=\chi(1)=0,
  \;
  \widetilde B(0)=\lambda^{\smallertext{2}\smallertext{1}},
  \;
  \widetilde B(1)=-\lambda^{\smallertext{1}\smallertext{2}}.
\]
The additional trace conditions in the theorem remove the second-order
belief and mixed terms and give the two endpoint equations stated above.
\end{proof}

\Cref{eq:option-PDE} has the same filtering-induced covariance structure as
the pricing equations of \citeauthor*{david2002option}
\cite[Equations~(9)--(11), pages~9--11]{david2002option}. In both models,
the stock price and posterior belief are state variables, and the common
innovation generates the mixed derivative. In their time-additive model,
the price--dividend ratio is affine in beliefs
\cite[Proposition~1(a)]{david2002option}; here it is generally nonlinear, and
the equilibrium pricing equation determines the function $\varphi$ that
enters $\sigma_\smallertext{S}$, $q_\smallertext{S}$, and $\widetilde B$.

\medskip
\Cref{app:characteristic-functions} gives the characteristic-function PDE and the Fourier representation of call prices. The remainder of the main text derives a conditional short-maturity return-skewness result and then examines its implied-volatility counterpart numerically using the equilibrium solution and marginal short rate.

\subsection{Short-maturity return skewness and the implied-volatility mechanism}

One channel shaping the option surface is the state dependence of
\[
 \sigma_\smallertext{S}(p)
 =
 \sigma+\ell^{\prime}(p)\chi(p),\; p\in[0,1].
\]
A positive dividend innovation raises both current dividends and the
posterior probability $p$. Under the standing assumptions of this section, the
price--dividend ratio is non-decreasing in $p$, so the two channels
reinforce one another and learning weakly raises stock volatility above
dividend volatility.

\medskip
The skew mechanism is suggested by the fact that the innovation also moves
future volatility. Fix $t\geq0$ and evaluate the conditional expansion at a
current state $p_t=p\in(0,1)$. Write
$\sigma_{\smallertext{0}}\coloneqq\sigma_\smallertext{S}(p)$, $\sigma_{\smallertext{0}}^\prime\coloneqq\sigma_\smallertext{S}^\prime(p)$, and
$\chi_{\smallertext{0}}\coloneqq\chi(p)$. For small maturity $\tau$, the local
Itô--Taylor expansion takes the form
\begin{equation}
\label{eq:ito-taylor-short-time}
 X_{t+\tau}-X_t
 =
 b_\smallertext{X}(p)\tau
 +
 \sigma_{\smallertext{0}}\Delta\widetilde W
 +
 \frac12\sigma_{\smallertext{0}}^{\prime}\chi_{\smallertext{0}}
 \big(
 (\Delta\widetilde W)^2-\tau
 \big)
 +
 R_\tau,
\end{equation}
where $\Delta\widetilde W\coloneqq \widetilde W_{t+\tau}-\widetilde W_t.$ Convergence in probability of the remainder is not enough to expand a third
moment. The proposition below therefore imposes the required
control explicitly.

\begin{proposition}[Leading short-maturity conditional return skewness]
\label{prop:short-maturity-skew-sign}
At a current state $p_t=p\in(0,1)$, under
{\rm\Cref{assum:risk-neutral-martingale}}, suppose that the remainder in
{\rm\Cref{eq:ito-taylor-short-time}} satisfies
\[
 \mathbb E^{\mathbb Q}\big[|R_\tau|^3\big|\Fc_t^\smallertext{D}\big]^{1/3}=o(\tau).
\]
Then
\[
  \mathbb E^{\mathbb Q}\bigg[
    \bigg(
      X_{t+\tau}-X_t
      -\mathbb E^{\mathbb Q}[X_{t+\tau}-X_t|\Fc_t^\smallertext{D}]
    \bigg)^3\bigg|\Fc_t^\smallertext{D}
  \bigg]
  =
  3\sigma_\smallertext{S}^2(p)\chi(p)\sigma_\smallertext{S}^\prime(p)\tau^2+o(\tau^2).
\]
Moreover, the conditional variance is
$\sigma_\smallertext{S}^2(p)\tau+o(\tau)$, and hence the standardised
conditional skewness satisfies
\[
  \operatorname{Skew}^{\mathbb Q}_t(\tau)
  =
  \frac{3\chi(p)\sigma_\smallertext{S}^{\prime}(p)}
  {\sigma_\smallertext{S}(p)}
  \sqrt{\tau}
  +
  o(\sqrt{\tau}).
\]
Thus $\chi(p)\sigma_\smallertext{S}^{\prime}(p)$ determines the leading
sign of short-maturity conditional return skewness.

\medskip
If, in addition,
$\sigma_\smallertext{S}$ is increasing on $(0,p^\star)$ and
decreasing on $(p^\star,1)$, then the leading skewness is positive below
$p^\star$, zero at $p^\star$, and negative above $p^\star$.
\end{proposition}

\medskip

For the marginal short rate
\eqref{eq:marginal-short-rate-option}, the remainder estimate follows
from the standing assumptions. Indeed,
\Cref{thm:markovian-smoothness} gives
$\varphi\in C^\infty([0,1])$, hence $r_f\in C^\infty([0,1])$,
while $\sigma_\smallertext{S}\geq\sigma>0$. Consequently,
$b_\smallertext{X}$, $\widetilde B$,
$\sigma_\smallertext{S}$, and $\chi$ are twice continuously
differentiable with bounded derivatives on $[0,1]$. The order-one
Itô--Taylor expansion then gives
\[
  \mathbb E^{\mathbb Q}
  \big[|R_\tau|^3\big|\Fc_t^\smallertext{D}\big]^{1/3}
  =
  O(\tau^{3/2}),
\]
see \citeauthor*{kloeden1992numerical}
\cite[Chapter~5]{kloeden1992numerical}.

\begin{proof}
Let $\Delta\widetilde W=\sqrt{\tau}\xi$, where $\xi\sim N(0,1)$. Then
after centring,
\[
  X_{t+\tau}-X_t
  -
  \mathbb E^{\mathbb Q}
  [X_{t+\tau}-X_t|\Fc_t^\smallertext{D}]
  =
  \sigma_{\smallertext{0}}\sqrt{\tau}\xi
  +
  \frac12\sigma_{\smallertext{0}}^\prime\chi_{\smallertext{0}}\tau(\xi^2-1)
  +
  \overline R_\tau,
\]
where
$\mathbb E^{\mathbb Q}
[|\overline R_\tau|^3|\Fc_t^\smallertext{D}]^{1/3}=o(\tau)$.
Hölder's inequality makes every term containing $\overline R_\tau$
negligible at order $\tau^2$. The only nonzero term of order
$\tau^2$ is
\[
  3\sigma_{\smallertext{0}}^2\tau
  \bigg(
    \frac12\sigma_{\smallertext{0}}^{\prime}\chi_{\smallertext{0}}\tau
  \bigg)
  \mathbb E[\xi^2(\xi^2-1)]
  =
  3\sigma_{\smallertext{0}}^2\sigma_{\smallertext{0}}^\prime\chi_{\smallertext{0}}\tau^2,
\]
because $\mathbb E[\xi^2(\xi^2-1)]=3-1=2$. The same expansion gives
$\operatorname{Var}^{\mathbb Q}
(X_{t+\tau}-X_t|\Fc_t^\smallertext{D})
=\sigma_{\smallertext{0}}^2\tau+o(\tau)$, proving the standardised formula.
\end{proof}

\citeauthor*{david2002option}
\cite[Section~5.2, pages~18--19]{david2002option} describe the same economic
sign mechanism: where stock volatility decreases with the posterior belief,
bad news lowers the posterior and raises volatility, producing negative
return--volatility dependence and an asymmetric smile. The proposition
above isolates the corresponding infinitesimal log-return effect through
$\chi\sigma_\smallertext{S}^{\prime}$. It does not by itself establish
the slope of implied volatility, which is examined numerically below.

\medskip
For the interpretation of the numerical results, a positive
implied-volatility slope means that out-of-the-money calls are expensive
relative to symmetric out-of-the-money puts, whereas a negative slope means
that the puts are expensive. Under the unimodality condition in
\Cref{prop:short-maturity-skew-sign}, the leading return-skewness coefficient
suggests the first pattern
below $p^\star$ and the second above $p^\star$. Whether this prediction
persists at a given finite maturity depends on the full option-pricing
dynamics. The numerical demonstration below examines this
mechanism under the equilibrium calibration.

\medskip
This price statement can be made explicit with a symmetric risk reversal.
For maturity $T$, define the zero-coupon bond price and forward stock price by
\[
 \mathcal B_{t\smallertext{T}}
 \coloneqq\mathbb E^{\mathbb Q}\bigg[
   \mathrm e^{-\int_t^\smallertext{T} r_f(p_s)\mathrm ds}
   \bigg|\mathcal F_t^\smallertext{D}\bigg],
 \;
 F_{t\smallertext{T}}
 \coloneqq\frac{\mathbb E^{\mathbb Q}\bigg[
   \mathrm e^{-\int_t^\smallertext{T} r_f(p_s)\mathrm ds}S_\smallertext{T}
   \bigg|\mathcal F_t^\smallertext{D}\bigg]}{\mathcal B_{t\smallertext{T}}}.
\]
For $k>0$, define
\[
\mathrm{RR}_k(p,\tau)
\coloneqq
\sigma_{\rm imp}(k,p,\tau)-\sigma_{\rm imp}(-k,p,\tau),
\]
where $\sigma_{\rm imp}(k,p,\tau)$ denotes the Black implied volatility at
current belief $p$, maturity $\tau$, and forward log-moneyness
$k=\log(K/F)$, with $F=F_{t\smallertext{T}}$. Since option prices are
non-decreasing in implied volatility, ${\rm RR}_k>0$ means that the
out-of-the-money call with $K=F\mathrm e^k$ is rich relative to the
symmetric out-of-the-money put with $K=F\mathrm e^{-k}$. Conversely,
${\rm RR}_k<0$ is the usual equity-index configuration: downside puts
carry the larger volatility and price premium.

\begin{figure}[ht!]
 \centering
 \includegraphics[width=0.98\textwidth]{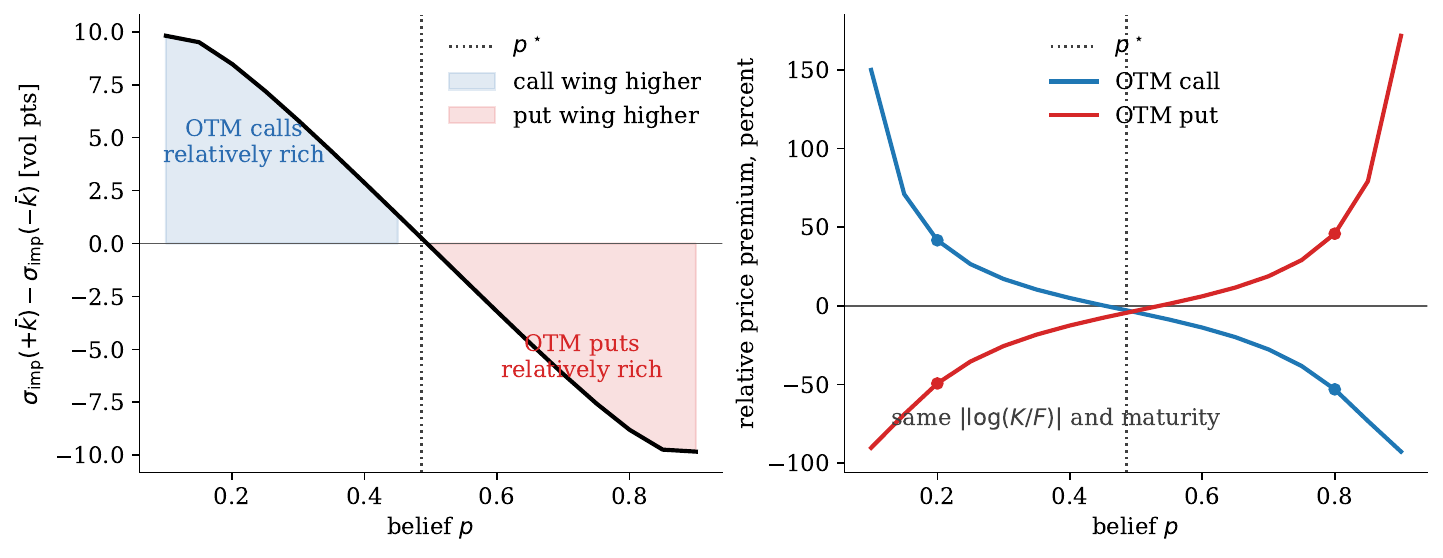}
 \caption{\small Equilibrium risk reversals and option premia at $\tau=0.10$ and
 $\bar k=0.07$. The left panel shows ${\rm RR}_{\bar k}$ computed from
 discounted option prices, in volatility percentage points. The right panel
 shows relative wing-price premia, in percent, over a Black price
 using the same bond price, forward and at-the-money implied volatility.
 The vertical line marks the volatility maximum.
 }
 \label{fig:risk-reversal-call-put}
\end{figure}

\medskip

The point $p^\star$ should be interpreted as the belief at which
instantaneous stock volatility is maximal. It is not, in general, the
mean-reversion point of the risk-neutral belief process. The latter is
determined by the zero of the risk-neutral drift
\[
 \widetilde B(p)=B(p)-\chi(p)\lambda(p),\; p\in(0,1).
\]
Thus the leading short-maturity return skewness is controlled by the
slope of $\sigma_\smallertext{S}(p)$, whereas the finite- and long-maturity
option surfaces also depend on the full risk-neutral dynamics of $p$.

\medskip
This interpretation is consistent with the empirical option-pricing
literature. \citeauthor*{dumas1998implied} \cite{dumas1998implied} document
that implied volatilities vary systematically with strike and maturity,
contrary to the flat Black--Scholes benchmark. For broad equity indices,
the dominant empirical pattern after the 1987 crash is a negative skew:
out-of-the-money puts are expensive because option-implied distributions
assign large probability mass to downside tail events; see
\citeauthor*{jackwerth1996recovering} \cite{jackwerth1996recovering} and \citeauthor*{bates2000post} \cite{bates2000post}. This corresponds in our
model to the region $p>p^\star$, where bad news raises future
volatility. At the same time, order-flow evidence shows that the wing
that becomes expensive depends on the market: index option markets are
especially sensitive to put buying pressure, while single-stock option
markets can be more affected by call buying pressure; see
\citeauthor*{bollen2004net} \cite{bollen2004net}. Consistently, \citeauthor*{bakshi2003stock}
\cite{bakshi2003stock} find that individual equity option-implied
distributions are far less negatively skewed than the market-index
distribution. The sign-changing skew in the model should therefore be
read as a state-dependent mechanism: the usual index put skew appears in
states where downside news increases future volatility, while call-rich
surfaces can arise in recovery states where upside news moves beliefs
towards the high-volatility region.

\subsection{Numerical illustration}
\label{subsec:option-numerics}
For the option-pricing illustration, we set
\[
g^{\smallertext{1}}=0.05,\;
g^{\smallertext{2}}=-0.05,\;
\sigma=0.04,\;
\delta=0.06,\;
\lambda^{\smallertext{1}\smallertext{2}}=
\lambda^{\smallertext{2}\smallertext{1}}=0.05,\;
\gamma=0.8,\; \theta=0.25,\; \psi=5,
\]
and use the maturity $\tau=0.10$.
These parameters satisfy $\min_{p\in[0,1]}c(p)=0.005128>0$.
The price--dividend ratio is the numerical solution
of \eqref{eq:two-state-varphi-ode}--\eqref{eq:two-state-varphi-boundary}.
We compute $r_f$ from \eqref{eq:marginal-short-rate-option}, use the
dividend yield $q_\smallertext{S}=1/\varphi$, and obtain the risk-neutral
belief drift from \eqref{eq:market-price-risk-option} and
\eqref{eq:risk-neutral-belief}. The volatility maximum is attained near
$p^\star=0.486$. Stock volatility ranges from $4\%$ to approximately
$35.6\%$, while the selected short rate ranges from $1.17\%$ to $6.92\%$.
The relatively large growth spread compared with dividend volatility
makes beliefs responsive to news and the learning contribution to stock
volatility substantial. This is a stylised illustration of the mechanism,
not an empirical calibration.

\medskip

We simulate the stock and beliefs jointly under
$\mathbb Q$, using log-Euler steps for the stock and Euler steps for
$\log(p/(1-p))$ to keep beliefs in $(0,1)$. Each path uses 512 time
steps, with antithetic Brownian increments. The smile and return-density
plots use 1,048,576 paths per initial belief; the risk-reversal and
at-the-money-skew plots use 262,144. Discounting is evaluated along each
path using $\exp(-\int_0^\tau r_f(p_s)\mathrm ds)$.

\medskip

Bond and forward prices are computed from the same
simulation, and strikes are expressed in forward log-moneyness
$\log(K/F_{0T})$, with $S_{\smallertext{0}}=1$. A lognormal stock with
coefficients frozen at the initial belief and driven by the same Brownian
increments provides Black-price control variates for option prices and
the prepaid forward, with control coefficients estimated on separate halves
of the antithetic sample. Implied volatilities are obtained by inverting the
Black formula with the simulated bond and forward prices.

\medskip
\Cref{fig:option-smiles-beliefs,fig:option-skew-belief} show the same
mechanism directly in option-implied volatilities. The curves have opposite slopes at low and high beliefs,
with visible curvature. At $p^\star$, the implied-volatility curve is
hump-shaped rather than flat.
In \Cref{fig:option-skew-belief}, the solid curve is the finite-difference
at-the-money skew obtained from the simulated option prices. The dashed curve
is the short-maturity proxy
$\chi(p)\sigma_\smallertext{S}^{\prime}(p)/\sigma_\smallertext{S}(p)$,
rescaled to the same vertical
range.

\begin{figure}[ht!]
 \centering
 \begin{minipage}[t]{0.49\textwidth}
 \centering
 \includegraphics[width=\textwidth]{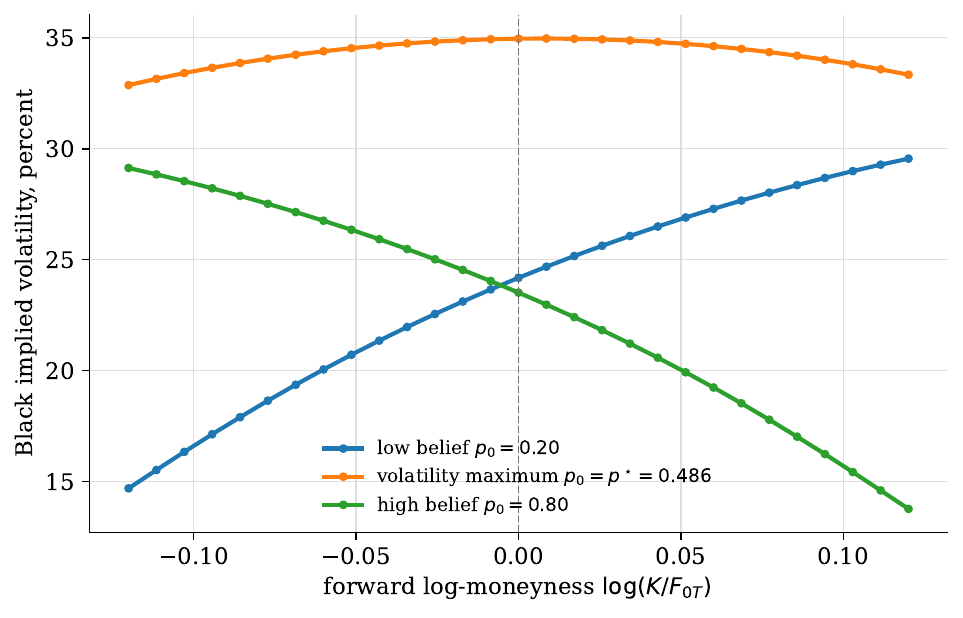}
 \caption{\small Equilibrium implied volatilities at initial beliefs $0.20$, $p^\star$
 and $0.80$, with the marginal short rate and dividend yield included.
 The low-belief smile slopes upwards and the high-belief smile slopes
 downwards; the curve at $p^\star$ is hump-shaped.
 }
 \label{fig:option-smiles-beliefs}
 \end{minipage}
 \hfill
 \begin{minipage}[t]{0.49\textwidth}
 \centering
 \includegraphics[width=\textwidth]{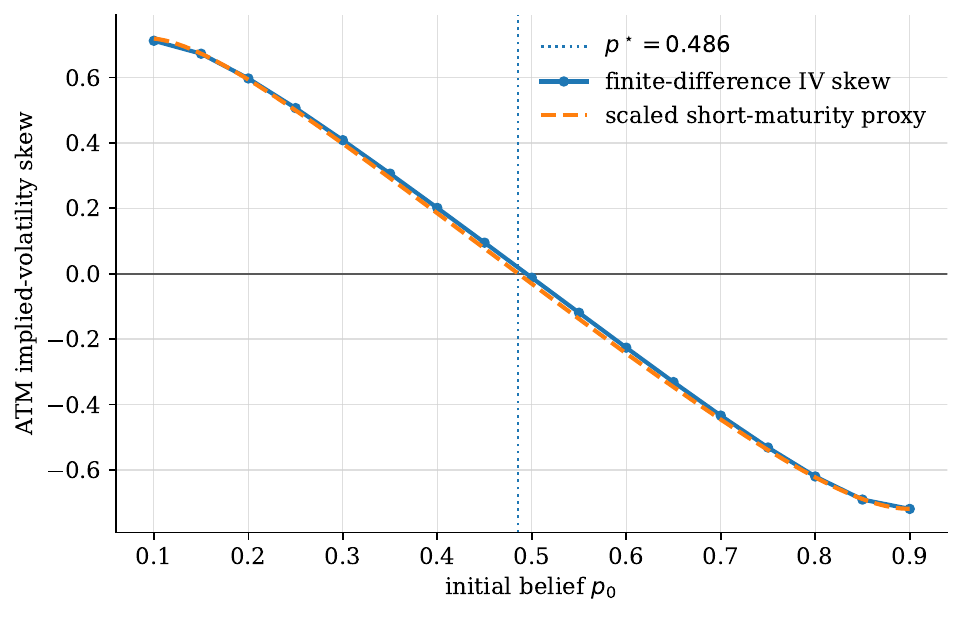}
 \caption{\small At-the-money implied-volatility skew as a function of the initial
 belief. The finite-difference skew from simulated option prices
 tracks the scaled short-maturity proxy
 $\chi(p)\sigma_\smallertext{S}^{\prime}(p)/\sigma_\smallertext{S}(p)$
 and changes sign near the belief $p^\star$ at which
 $\sigma_\smallertext{S}(p)$ is maximised.
 }
 \label{fig:option-skew-belief}
 \end{minipage}
\end{figure}

\medskip

\Cref{fig:shock-decomposition} gives the pathwise interpretation. The
same Brownian innovation has a direct stock-return component and an
information component through beliefs. When belief is low, a positive shock
moves the posterior towards the high-volatility region. When belief is high,
a negative shock does the same.
The numerical implied-volatility skew follows the same state-dependent
pattern; it is not imposed through an exogenous leverage effect.

\begin{figure}[H]
  \centering
  \includegraphics[width=0.70\textwidth]{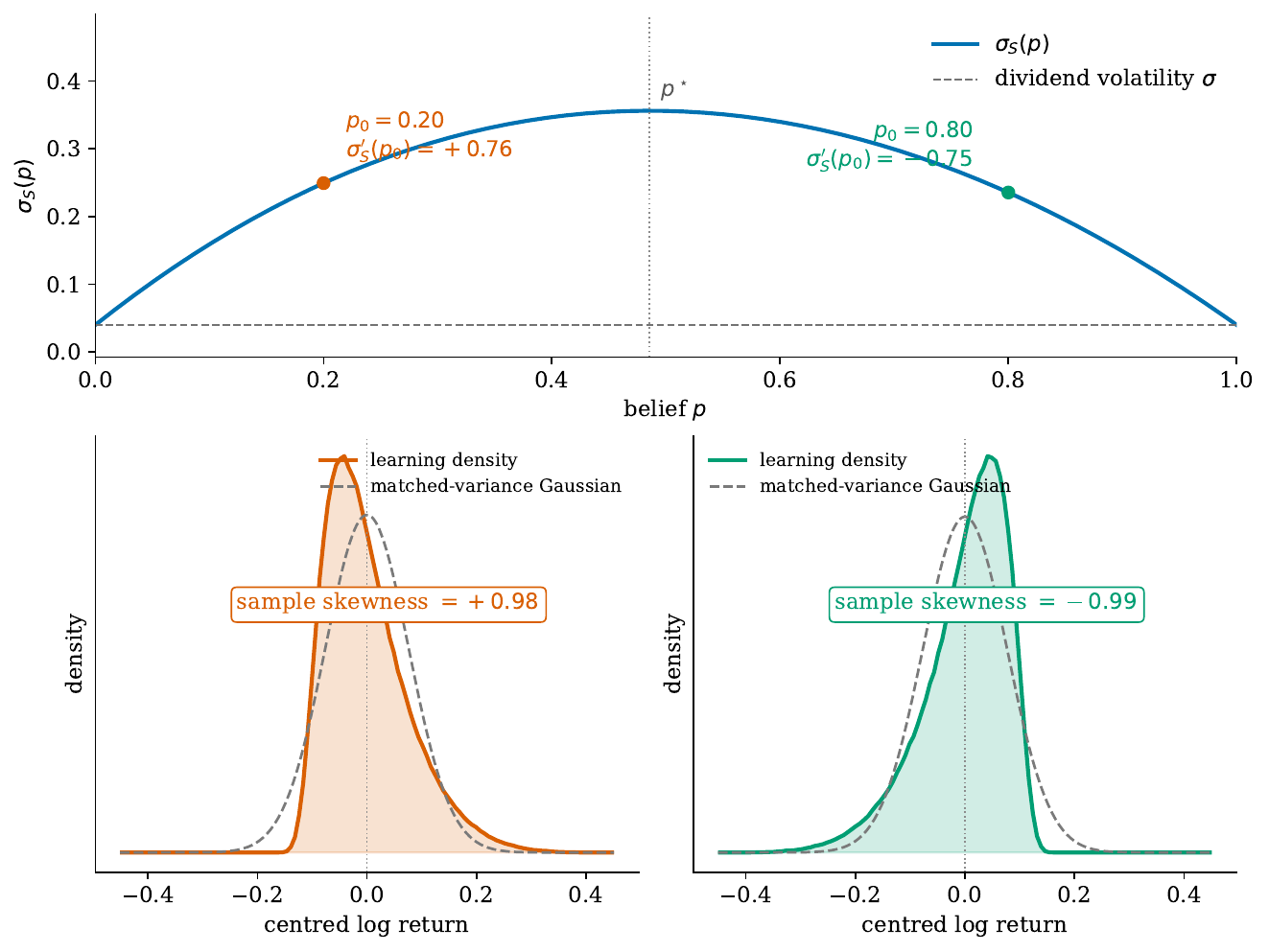}
  \caption{\small  Equilibrium volatility and risk-neutral return distributions.
  The upper panel shows $\sigma_\smallertext{S}(p)$. The lower panels
  show simulated log returns over $\tau=0.10$, centred at their sample
  means, with their sample skewness and a Gaussian benchmark with the
  same variance.
  }
  \label{fig:shock-decomposition}
\end{figure}

\appendix

\section{Proofs for the main equilibrium characterisation}
\label{app:proof-main-equilibrium}

\subsection{Auxiliary equilibrium lemmas}
\label{app:comparison-closure-proofs}

Fix an initial belief $p$. The first result converts the sign of the HJB residual into the comparison
needed for verification. The second passes from zero-residual finite-horizon
dynamics to the infinite-horizon BSDE. Both results use the admissibility conditions in
\Cref{def:admissibility}.

\begin{lemma}[One-sided comparison]
\label{lem:one-sided-comparison}
Suppose that $0<\theta\leq1$.
Let $c\in\mathcal C_a^p$ and let
$(Y^{p,c},Z^{p,c})$ be its utility-BSDE solution. Let $\overline Y$ be positive and
continuous and $\mathbb F^{\smallertext{D},p}$--adapted, and suppose that, on every finite time interval
\begin{equation*}
  \mathrm d\overline Y_t
  =-\big(F_t(c_t,\overline Y_t)+R_t\big)\mathrm dt
   +\overline Z_t \mathrm d\overline W_t^p,
\end{equation*}
where $R$ is $\mathbb F^{\smallertext{D},p}$--progressively
measurable and pathwise locally Lebesgue integrable under $\mathbb P^p$.
If $\rho\in\{-1,1\}$ and
\begin{equation*}
  \rho R_t\geq0,\;\mathrm dt\otimes\mathbb P^p\text{\rm--a.e.},
  \;
  \overline Z\in\mathbb H^1(\mathbb P^p),
  \;
  \liminf_{t\to\smallertext{+}\infty}\overline Y_t=0,
  \;\mathbb P^p\text{\rm--a.s.},
\end{equation*}
then
\begin{equation}
\label{eq:comparison-conclusion}
  \rho\big(Y_t^{p,c}-\overline Y_t\big)\leq0,
  \; t\geq0,
  \; \mathbb P^p\text{\rm--a.s.}
\end{equation}
\end{lemma}

\begin{proof}
Write
\[
  A_t^{p,c}=\int_{\smallertext{0}}^tF_s(c_s,Y_s^{p,c}) \mathrm ds,
  \;
  M_t^{p,c}=\int_{\smallertext{0}}^tZ_s^{p,c} \mathrm d\overline W_s^p.
\]
At time zero, the utility BSDE gives
$A_\infty^{p,c}=Y_{\smallertext{0}}^{p,c}+M_\infty^{p,c}$. By
\Cref{rem:verification-integrand}, $\mathbb F^{\smallertext{D},p}$ is a
completed Brownian filtration, so $\mathcal F_{\smallertext{0}}^{\smallertext{D},p}$ is
trivial up to null sets, and hence
$Y_{\smallertext{0}}^{p,c}$ is deterministic. Since
$Z^{p,c}\in\mathbb H^1(\mathbb P^p)$,
$M^{p,c}$ is uniformly integrable and has a terminal value
$M_\infty^{p,c}\in\mathbb L^1(\mathbb P^p)$. Hence
$A_\infty^{p,c}\in\mathbb L^1(\mathbb P^p)$.
Since $F_s(c_s,Y_s^{p,c})\geq0$, $s\geq 0$, $\mathbb{P}^p\text{\rm--a.s.}$, one has
$A_t^{p,c}\uparrow A_\infty^{p,c}$ and
$A_\infty^{p,c}-A_t^{p,c}\longrightarrow0$, $\mathbb P^p$--almost surely and in
$\mathbb L^1(\mathbb P^p)$. Moreover, the utility BSDE gives
\[
  Y_t^{p,c}
  =
  \big(A_\infty^{p,c}-A_t^{p,c}\big)
  -
  \big(M_\infty^{p,c}-M_t^{p,c}\big),\; t\geq 0,\; \mathbb{P}^p\text{\rm--a.s.}
\]
Because $M_t^{p,c}\longrightarrow M_\infty^{p,c}$,  $\mathbb{P}^p$--almost surely and in
$\mathbb L^1(\mathbb P^p)$, it follows that
$Y_t^{p,c}\longrightarrow0$, $\mathbb{P}^p$--almost surely and in $\mathbb L^1(\mathbb P^p)$.

\medskip

Let
\[
  D_t^\rho=\rho\big(Y_t^{p,c}-\overline Y_t\big),
  \;
  \Xi_t=\rho\big(Z_t^{p,c}-\overline Z_t\big),
  \;
  b_t^\rho
  =\rho\big(
    -F_t\big(c_t,Y_t^{p,c}\big)+F_t\big(c_t,\overline Y_t\big)+R_t
   \big),\; t\geq 0,\; \mathbb{P}^p\text{\rm--a.s.}
\]

Then
\[
  \mathrm dD_t^\rho=b_t^\rho \mathrm dt+\Xi_t \mathrm d\overline W_t^p.
\]
Since $0<\theta\leq1$, the exponent $1-1/\theta$ is non-positive, so
$y\longmapsto F_t(c,y)$ is non-increasing. Fix some $t\geq 0$; on $\{D_t^\rho>0\}$, this
property and the condition $\rho R_t\geq0$ imply $b_t^\rho\geq0$. Indeed,
when $\rho=1$ one has $Y_t^{p,c}>\overline Y_t$, whereas when $\rho=-1$
one has $Y_t^{p,c}<\overline Y_t$; the two cases give the same sign.

\medskip
Let $L^\smallertext{0}(D^\rho)$ denote the local time of $D^\rho$ at zero. The
finite-horizon It\^o--Tanaka formula gives for any $T>0$
\[
  (D_\smallertext{T}^\rho)^{\smallertext{+}}
  =
  (D_t^\rho)^{\smallertext{+}}
  +\int_t^{\smallertext{T}}\mathbf1_{\{D_s^\rho>0\}}b_s^\rho\mathrm ds
  +\int_t^{\smallertext{T}}\mathbf1_{\{D_s^\rho>0\}}\Xi_s\mathrm d\overline W_s^p
  +\frac12\big(L_\smallertext{T}^\smallertext{0}(D^\rho)-L_t^\smallertext{0}(D^\rho)\big),\; t\in[0,T],\; \mathbb{P}^p\text{\rm--a.s.}
\]
The drift integral and the local-time increment are non-negative.
Consequently, for $0\leq t\leq T$
\begin{equation}
\label{eq:pathwise-tanaka}
  (D_t^\rho)^{\smallertext{+}}
  \leq
  (D_\smallertext{T}^\rho)^{\smallertext{+}}
  -\int_t^\smallertext{T}\mathbf1_{\{\smallertext{D}_\smalltext{s}^\smalltext{\rho}>0\}}\Xi_s \mathrm d\overline W_s^p,\; \mathbb{P}^p\text{\rm--a.s.}
\end{equation}
The integrand in the stochastic integral belongs to
$\mathbb H^1(\mathbb P^p)$ because
both $Z^{p,c}$ and $\overline Z$ do. The integral therefore has a $\mathbb P^p$--almost sure and $\mathbb L^1(\mathbb P^p)$ terminal value.

\medskip
Choose continuous versions for which \eqref{eq:pathwise-tanaka} holds for
all real $T$ on one event of probability one. On that event, choose path by
path a sequence $(T_k)_{k\in\mathbb N}$ increasing to $+\infty$ such that $\overline Y_{\smallertext{T}_\smalltext{k}}\longrightarrow0$. The times
$(T_k)_{k\in\mathbb N}$ need not be stopping times because no expectation is taken before the
pathwise limit. Since $Y_{\smallertext{T}_\smalltext{k}}^{p,c}\longrightarrow0$ as well,
$|D_{\smallertext{T}_\smalltext{k}}^\rho|\longrightarrow0$.
Applying \eqref{eq:pathwise-tanaka} with $T=T_k$ and then
letting $k\longrightarrow+\infty$ gives
\[
  (D_t^\rho)^{\smallertext{+}}
  \leq
  -\int_t^{\smallertext{+}\infty}
    \mathbf1_{\{\smallertext{D}_\smalltext{s}^\smalltext{\rho}>0\}}\Xi_s \mathrm d\overline W_s^p,\; t\geq 0,\; \mathbb{P}^p\text{\rm--a.s.}
\]
The right-hand side is an integrable terminal increment of a uniformly
integrable martingale. Since it dominates the non-negative random variable
$(D_t^\rho)^{\smallertext{+}}$, the latter is integrable. Only now do we take expectations,
obtaining $\mathbb E^p[(D_t^\rho)^{\smallertext{+}}]\leq0$. Thus $(D_t^\rho)^{\smallertext{+}}=0$. Applying
the argument first at rational times and then using continuity proves
\eqref{eq:comparison-conclusion} simultaneously for every $t\geq0$.
\end{proof}

\begin{lemma}
\label{lem:zero-residual}
Let $\theta>0$, let $c\in\mathcal C^p$, and suppose that a positive continuous
process $\overline Y$ is
$\mathbb F^{\smallertext{D},p}$--adapted and satisfies, on every finite interval
\begin{equation*}
  \mathrm d\overline Y_t
  =-F_t\big(c_t,\overline Y_t\big) \mathrm dt
   +\overline Z_t \mathrm d\overline W_t^p,
  \;
  \overline Z\in\mathbb H^1(\mathbb P^p),
  \;
  \liminf_{t\to\smallertext{+}\infty}\overline Y_t=0,
  \;\mathbb P^p\text{\rm--a.s.}
\end{equation*}
If $\overline Y_{\smallertext{0}}\in\mathbb L^1(\mathbb P^p)$, then
\[
  \mathbb E^p\bigg[\int_{\smallertext{0}}^{\smallertext{+}\infty}
    F_t\big(c_t,\overline Y_t\big) \mathrm dt\bigg]<+\infty,
  \;
  \overline Y_t\longrightarrow0,
  \;\mathbb{P}^p\text{\rm--almost surely and in }
  \mathbb L^1(\mathbb P^p),
\]
and
\begin{equation}
\label{eq:closed-bsde}
  \overline Y_t
  =\int_t^{\smallertext{+}\infty} F_s(c_s,\overline Y_s) \mathrm ds
   -\int_t^{\smallertext{+}\infty}\overline Z_s \mathrm d\overline W_s^p,\; t\geq0,\; \mathbb P^p\text{\rm--a.s.}
\end{equation}
Thus $(\overline Y,\overline Z)$ is a positive solution of
the utility BSDE associated with $c$ whose martingale
integrand belongs to $\mathbb H^1(\mathbb P^p)$.
If $0<\theta\leq1$, this solution is unique among positive solutions of
\eqref{eq:closed-bsde} whose martingale integrands belong to
$\mathbb H^1(\mathbb P^p)$.
\end{lemma}

\begin{proof}
Let
\[
  M_t=\int_{\smallertext{0}}^t\overline Z_s \mathrm d\overline W_s^p,
  \;
  B_t=\int_{\smallertext{0}}^tF_s(c_s,\overline Y_s) \mathrm ds,\;  t\geq0,\; \mathbb P^p\text{\rm--a.s.}
\]
Since $\theta>0$, the generator is non-negative and $B$ is non-decreasing. The
$\mathbb H^1(\mathbb P^p)$ condition makes $M$ a uniformly integrable martingale with an
almost-sure and $\mathbb L^1(\mathbb P^p)$ limit $M_\infty$. On every finite interval,
$\overline Y=\overline Y_{\smallertext{0}}-B+M$.

\medskip

Choose, path by path, a sequence $(T_k)_{k\in\mathbb N}$ increasing to $+\infty$ with
$\overline Y_{T_\smalltext{k}}\longrightarrow0$. Then
$B_{T_\smalltext{k}}\uparrow B_\infty=\overline Y_{\smallertext{0}}+M_\infty$. Consequently
$B_\infty$ is finite and integrable. Moreover
\[
  \overline Y_t
  =(B_\infty-B_t)-(M_\infty-M_t),\;  t\geq0,\; \mathbb P^p\text{\rm--a.s.},
\]
so $\overline Y_t\longrightarrow 0$, $\mathbb P^p$--almost surely and in
$\mathbb L^1(\mathbb P^p)$. Passing to the limit in
the finite-horizon identity gives \eqref{eq:closed-bsde}.

\medskip

If $0<\theta\leq1$, applying the same It\^o--Tanaka argument to the
difference of two positive solutions, first in one direction and then in the
other, shows that their $Y$ components coincide. Subtracting their dynamics
then shows that their $Z$ components agree
$\mathrm dt\otimes\mathbb P^p$--almost everywhere.
\end{proof}

We now prove the equilibrium-reduction lemma by deriving the realised auxiliary
drift and short-rate restrictions from optimality and clearing.

\begin{proof}[Proof of \Cref{lem:equilibrium-reduction}]
Fix an initial belief $p$ and work on the corresponding model copy.
Define the clearing wealth process by
\[
  X^p
  \coloneqq
  X^{p,\smallertext{S}_{\smalltext{0}}^\smalltext{p},\smallertext{1},\smallertext{D}^\smalltext{p}}.
\]
At the clearing controls $(\pi,c)=(1,D^p)$, the stock process $S^p$
satisfies the same linear wealth equation as $X^p$ and has the same initial
value. By uniqueness
\[
  X_t^p=S_t^p=D_t^pQ_t^p,
  \; t\geq0,\; \mathbb P^p\text{\rm--a.s.}
\]

Let $v$ denote the function representing the agent's value in the
classical $C^2$ state representation. The clearing controls
attain the HJB Hamiltonian. For any $a>0$, linearity of the wealth equation
gives
$X^{p,ax,\pi,ac}=aX^{p,x,\pi,c}$, and the admissibility conditions are
preserved by this scaling. Moreover, since
$1-1/\psi+(1-\gamma)(1-1/\theta)=1-\gamma$, the Epstein--Zin recursion gives
$V^{p,ac}=a^{1-\gamma}V^{p,c}$. Applying the same argument with $a^{-1}$
and taking suprema therefore yields
$v(q,ax,u)=a^{1-\gamma}v(q,x,u)$. Hence
\[
  k(q,u)\coloneqq(1-\gamma)v(q,1,u)>0,
  \;
  v(q,x,u)
  =
  \frac{x^{1-\gamma}}{1-\gamma}k(q,u).
\]

The interior consumption first-order condition at the clearing allocation
therefore becomes
\[
  (X_t^p)^{-\gamma}k(P_t^p,H_t^p)
  =
  \delta(D_t^p)^{-1/\psi}
  \bigl((X_t^p)^{1-\gamma}k(P_t^p,H_t^p)\bigr)^{1-1/\theta},\; \mathrm dt\otimes\mathbb P^p\text{\rm--a.e.}
\]

Using
\[
  \frac{1-\gamma}{\theta}=1-\frac1\psi,
\]
this is equivalent to
\[
  \frac{D_t^p}{X_t^p}
  =
  \delta^\psi
  k(P_t^p,H_t^p)^{-\psi/\theta},\; \mathrm dt\otimes\mathbb P^p\text{\rm--a.e.}
\]
Since $X^p=D^pQ^p$, it follows that
\[
  k(P_t^p,H_t^p)
  =
  \delta^\theta(Q_t^p)^{\theta/\psi},\; \mathrm dt\otimes\mathbb P^p\text{\rm--a.e.}
\]
The equality extends by continuity from
$\mathrm dt\otimes\mathbb P^p$--almost everywhere to every $t\geq0$,
outside one $\mathbb P^p$--null set.

\medskip
Next, write
\[
  (Y^p,Z^p)
  \coloneqq
  (Y^{p,\smallertext{D}^\smalltext{p}},Z^{p,\smallertext{D}^\smalltext{p}}),
\]
for the utility-BSDE solution associated with the clearing consumption.
The state representation and
$1-\gamma+\theta/\psi=\theta$ now yield
\begin{equation}
\label{eq:clearing-utility-reduction}
  Y_t^p
  =
  (1-\gamma)\mathrm e^{-\delta\theta t}
  v(P_t^p,X_t^p,H_t^p)=
  \mathrm e^{-\delta\theta t}
  (X_t^p)^{1-\gamma}k(P_t^p,H_t^p)=
  \delta^\theta\mathrm e^{-\delta\theta t}
  (D_t^p)^{1-\gamma}(Q_t^p)^\theta,\; t\geq 0,\; \mathbb{P}^p\text{\rm--a.s.}
\end{equation}

Let
\[
  \ell(q)
  \coloneqq
  \frac{\partial_p\varphi(q)}{\varphi(q)},
  \;
  d(q)\coloneqq A(q)^\top \ell(q),\; q\in\Delta^n.
\]
It\^o's formula applied to \eqref{eq:clearing-utility-reduction} gives
\begin{equation}
\label{eq:clearing-utility-ito}
  \frac{\mathrm dY_t^p}{Y_t^p}
  =
  \theta
  \bigg(
    \frac{\mathcal H(P_t^p,H_t^p)}{H_t^p}
    -\zeta(P_t^p)
  \bigg)\mathrm dt
  +
  \lambda_\smallertext{U}(P_t^p) \mathrm d\overline W_t^p,\;  t\geq0,\; \mathbb P^p\text{\rm--a.s.},
\end{equation}
where
\begin{equation}
\label{eq:utility-tilt-lambda}
  \lambda_\smallertext{U}(q)
  \coloneqq
  (1-\gamma)\sigma
  -\frac{\theta}{\sigma}d(q),\; q\in\Delta^n.
\end{equation}
On the other hand, \eqref{eq:clearing-utility-reduction} is the utility-BSDE
solution associated with consumption $D^p$. Direct substitution into the
generator gives
\[
  \frac{F_t(D_t^p,Y_t^p)}{Y_t^p}
  =
  \frac{\theta}{Q_t^p},\;  t\geq0,\; \mathbb P^p\text{\rm--a.s.}
\]
The BSDE therefore implies
\begin{equation}
\label{eq:clearing-utility-bsde}
  \frac{\mathrm dY_t^p}{Y_t^p}
  =
  -\frac{\theta}{Q_t^p} \mathrm dt
  +\frac{Z_t^p}{Y_t^p} \mathrm d\overline W_t^p,\;  t\geq0,\; \mathbb P^p\text{\rm--a.s.}
\end{equation}
Uniqueness of the semi-martingale decomposition in
\eqref{eq:clearing-utility-ito} and
\eqref{eq:clearing-utility-bsde} yields
\begin{equation}
\label{eq:reduction-drift-and-volatility}
  \mathcal H(P_t^p,H_t^p)
  =H_t^p\zeta(P_t^p)-\varphi(P_t^p)^{-1},
  \;
  \frac{Z_t^p}{Y_t^p}=\lambda_\smallertext{U}(P_t^p),\; \mathrm dt\otimes\mathbb P^p\text{\rm--almost everywhere.}
\end{equation}

It remains to derive the short-rate restriction. Let $\mu_\smallertext{\rm TR}^p$ and
$\sigma_\smallertext{\rm TR}^p$ denote the
drift and volatility of the stock's total price-plus-dividend return
$(\mathrm dS_t^p+D_t^p\mathrm dt)/S_t^p$. The candidate dynamics give
\begin{gather*}
  \mu_{\smallertext{\rm TR},t}^p
  =
  \overline g(P_t^p)
  +\frac{\mathcal H(P_t^p,H_t^p)+\varphi(P_t^p)^{-1}}{H_t^p}
  +\ell(P_t^p)^\top \big(b(P_t^p)-A(P_t^p)\big)+
  \frac{1}{2\sigma^2}
  \frac{A(P_t^p)^\top\partial_{pp}\varphi(P_t^p)A(P_t^p)}{\varphi(P_t^p)},   \;
  \mathrm dt\otimes\mathbb P^p\text{\rm--a.e.},\\
  \sigma_{\smallertext{\rm TR},t}^p
  =
  \sigma-\frac{d(P_t^p)}{\sigma},   \;
  \mathrm dt\otimes\mathbb P^p\text{\rm--a.e.}
\end{gather*}
At equilibrium, $\pi=1$, and the constraint $\pi\in[0,1]$ means that a small
feasible change has the form
$\pi^\varepsilon=1-\varepsilon\eta$, where $\varepsilon>0$ is small and
$\eta\geq0$. At the current state, the expression for $v$ above gives
$X_t^pv_{xx}/v_x=-\gamma$. Matching the Brownian terms in
$k(P^p,H^p)=\delta^\theta(Q^p)^{\theta/\psi}$ also gives
$A^\top\partial_{xp}v/v_x=(\theta/\psi)d(P_t^p)$.
The one-sided portfolio first-order condition therefore gives

\[
  \mu_{\smallertext{\rm TR},t}^p-r_t^p
  -\frac{\theta}{\psi\sigma}d(P_t^p)
   \sigma_{\smallertext{\rm TR},t}^p
  -\gamma\big(\sigma_{\smallertext{\rm TR},t}^p\big)^2
  \geq0,\;
  \mathrm dt\otimes\mathbb P^p\text{\rm--a.e.}
\]
Substitution of the first identity in
\eqref{eq:reduction-drift-and-volatility} and direct expansion give
\[
  \mu_{\smallertext{\rm TR},t}^p
  -\frac{\theta}{\psi\sigma}d(P_t^p)
   \sigma_{\smallertext{\rm TR},t}^p
  -\gamma\big(\sigma_{\smallertext{\rm TR},t}^p\big)^2
  =\overline r(P_t^p),\;
  \mathrm dt\otimes\mathbb P^p\text{\rm--a.e.}
\]
Putting both expressions together
\[
  \mathcal R(P_t^p,H_t^p)=r_t^p\leq\overline r(P_t^p),\;
  \mathrm dt\otimes\mathbb P^p\text{\rm--a.e.}
\]
Since $p$ was arbitrary, the result holds on every model copy.
\end{proof}

\subsection{Proof of the main theorem}
\label{app:proof-main-theorem}

We are now ready to prove the main result of the paper.

\begin{proof}[{Proof of \Cref{thm:main}}]
\noindent\textbf{Part $1$: necessity.} We organise the argument into five
steps.

\medskip
\noindent\emph{Step $1$: upgrading the clearing-utility terminal condition.} Fix a $\mathfrak C$-equilibrium whose
representative agent's value function admits a classical $C^2$ state
representation and fix an
initial belief $p$.
The proof of \Cref{lem:equilibrium-reduction} identifies the
transformed clearing-utility pair, denoted by
$(\widehat Y^p,\widehat Z^p)$, through
\eqref{eq:clearing-utility-reduction} and
$\widehat Z^p=\lambda_\smallertext{U}(P^p)\widehat Y^p$, where
$\lambda_\smallertext{U}$ is defined in
\eqref{eq:utility-tilt-lambda}. It therefore satisfies
\begin{equation}
\label{eq:clearing-Y-dynamics-revised}
  \mathrm d\widehat Y_t^p
  =-\frac{\theta}{Q_t^p}\widehat Y_t^p \mathrm dt
   +\lambda_\smallertext{U}(P_t^p)\widehat Y_t^p \mathrm d\overline W_t^p.
\end{equation}

Equilibrium admissibility gives
$\widehat Z^p\in\mathbb H^1(\mathbb P^p)$ and only the condition
$\liminf_{t\to\smallertext{+}\infty}\widehat Y_t^p=0$,
$\mathbb P^p$--almost surely. Since $\widehat Y_{\smallertext{0}}^p$ is deterministic and
hence integrable, \Cref{lem:zero-residual} upgrades this to
\begin{equation}
\label{eq:Y-clearing-zero-revised}
  \widehat Y_t^p\longrightarrow0,
  \;\mathbb P^p\text{\rm--almost surely and in }
  \mathbb L^1(\mathbb P^p).
\end{equation}

Thus, the admissibility terminal condition is strengthened to
full almost-sure and $\mathbb L^1$ convergence.

\medskip

\noindent\emph{Step $2$: finding the critical initial value.} Define the density process
\[
  \widehat L_t^p
  \coloneqq
  \mathcal E\bigg(
    \int_{\smallertext{0}}^\cdot\lambda_\smallertext{U}(P_s^p)
    \mathrm d\overline W_s^p
  \bigg)_t,
  \; t\geq0.
\]
Because $\lambda_\smallertext{U}$ is bounded, $\widehat L^p$ is a true
martingale on every finite horizon. For $T>0$, define
$\widehat{\mathbb P}^{p,\smallertext{T}}$ on $\mathcal F_\smallertext{T}^{\smallertext{D},p}$ by
\[
  \frac{\mathrm d\widehat{\mathbb P}^{p,\smallertext{T}}}
       {\mathrm d\mathbb P^p}
  \bigg|_{\smallertext{\mathcal F}_\smalltext{T}^{\smalltext{D}\smalltext{,}\smalltext{p}}}
  =\widehat L_\smallertext{T}^p.
\]
Under $\widehat{\mathbb P}^{p,\smallertext{T}}$
\[
  \widehat W_t^p
  \coloneqq
  \overline W_t^p-
  \int_{\smallertext{0}}^t\lambda_\smallertext{U}(P_s^p)\mathrm ds,
  \; 0\leq t\leq T,
\]
is a Brownian motion, and
\begin{equation}
\label{eq:tilted-filter-revised}
  \mathrm dP_t^p
  =\widehat b(P_t^p) \mathrm dt
   -\frac{A(P_t^p)}{\sigma} \mathrm d\widehat W_t^p,
  \; \text{where}\; 
  \widehat b(p)
  \coloneqq\Lambda^\top p-\frac{\lambda_\smallertext{U}(p)}{\sigma}A(p).
\end{equation}
The coefficients in \eqref{eq:tilted-filter-revised} are Lipschitz-continuous,
so its finite-horizon solution laws are the restrictions of a unique global
solution law on canonical path space. We denote this law and its expectation
by $\widehat{\mathbb P}^p$ and $\widehat{\mathbb E}^p$, and continue to write
$(P^p,\widehat W^p)$ for the coordinate solution and its driving Brownian
motion. Thus global $\widehat{\mathbb P}^p$--statements below refer to this
canonical law; finite-horizon identities are transferred from the original
model through the equivalent measures $\widehat{\mathbb P}^{p,\smallertext{T}}$.

\medskip

Solving
\eqref{eq:clearing-Y-dynamics-revised} gives
\[
  \widehat Y_t^p
  =\widehat Y_{\smallertext{0}}^p\widehat L_t^p
   \exp\bigg(-\theta\int_{\smallertext{0}}^t(Q_s^p)^{-1} \mathrm ds\bigg),\; t\geq 0,\; \mathbb P^p\text{\rm--a.s.}
\]
Consequently
\[
  \frac{\mathbb E^p[\widehat Y_t^p]}{\widehat Y_{\smallertext{0}}^p}
  =\widehat{\mathbb E}^{p}\bigg[
    \exp\bigg(-\theta\int_{\smallertext{0}}^t(Q_s^p)^{-1} \mathrm ds\bigg)
   \bigg], \; t\geq 0.
\]
The left-hand side tends to zero by
\eqref{eq:Y-clearing-zero-revised}. Indeed, with
$K^p\coloneqq\theta\int_{\smallertext{0}}^\cdot(Q_s^p)^{-1}\mathrm ds$,
$\mathrm e^{-\smallertext{K}_\smalltext{t}^\smalltext{p}}\downarrow\mathrm e^{-\smallertext{K}_\smalltext{\infty}^\smalltext{p}}$, so dominated
convergence gives
$\widehat{\mathbb E}^p[\mathrm e^{-\smallertext{K}_\smalltext{\infty}^\smalltext{p}}]=0$ and hence
$K_\infty^p=+\infty$, $\widehat{\mathbb P}^p$--almost surely. Equivalently
\begin{equation}
\label{eq:inverse-Q-diverges}
  \int_{\smallertext{0}}^{\smallertext{+}\infty}(Q_s^p)^{-1} \mathrm ds=+\infty,
  \;
  \widehat{\mathbb P}^{p}\text{\rm--a.s.}
\end{equation}

For $t\geq0$, define
\[
  \Gamma_t^p=\int_{\smallertext{0}}^t\zeta(P_s^p) \mathrm ds,
  \;
  J_t^p
  =\int_{\smallertext{0}}^t
    \mathrm e^{-\smallertext{\Gamma}_\smalltext{s}^\smalltext{p}}\varphi(P_s^p)^{-1} \mathrm ds.
\]
The first identity in
\eqref{eq:on-path-equilibrium-reduction} gives
\[
  \mathrm dH_t^p
  =\big(\zeta(P_t^p)H_t^p-\varphi(P_t^p)^{-1}\big)\mathrm dt.
\]
Variation of constants yields
\begin{equation}
\label{eq:H-explicit-revised}
  H_t^p
  =\mathrm e^{\smallertext{\Gamma}_\smalltext{t}^\smalltext{p}}\big(h(p)-J_t^p\big), \; t\geq 0,\; \widehat{\mathbb P}^p\text{\rm--a.s.}
\end{equation}
Since $H_t^p>0$,
$h(p)-J_t^p=\mathrm e^{-\smallertext{\Gamma}_\smalltext{t}^\smalltext{p}}H_t^p>0$
at every finite time. Moreover, for Lebesgue--almost every $t$
\[
  \frac1{Q_t^p}
  =\frac{\dot J_t^p}{h(p)-J_t^p}
  =-\frac{\mathrm d}{\mathrm dt}\mathrm{log}\big(h(p)-J_t^p\big).
\]
Thus
\[
  \int_{\smallertext{0}}^t(Q_s^p)^{-1}\mathrm ds
  =\mathrm{log}\bigg(\frac{h(p)}{h(p)-J_t^p}\bigg).
\]
By \eqref{eq:inverse-Q-diverges}, the left-hand side tends to infinity;
since $J^p$ is increasing, this forces
\begin{equation}
\label{eq:critical-anchor-revised}
  h(p)=J_\infty^p,
  \;
  \widehat{\mathbb P}^{p}\text{\rm--a.s.}
\end{equation}

Thus, the infinite-horizon condition determines the initial
value $h(p)$ through \eqref{eq:critical-anchor-revised}.

\medskip

\noindent\emph{Step $3$: continuity of $h$.} We next show that this forces the Borel map $h$ to be continuous. The
coefficients $\widehat b$ and $-A/\sigma$ in
\eqref{eq:tilted-filter-revised} are Lipschitz-continuous on the compact simplex.
Extend them to Lipschitz-continuous maps on $\mathbb R^n$. On an auxiliary
filtered probability space
$(\Omega^{\mathrm c},\mathcal F^{\mathrm c},
(\mathcal F_t^{\mathrm c})_{t\geq0},\mathbb P^{\mathrm c})$ carrying a Brownian
motion $W^{\mathrm c}$, let $\mathsf P^x$, for $x\in\Delta^n$, be the unique
strong solution of
\[
  \mathrm d\mathsf P_t^x
  =\widehat b(\mathsf P_t^x)\mathrm dt
   -\frac{A(\mathsf P_t^x)}{\sigma}\mathrm dW_t^{\mathrm c},
  \; \mathsf P_{\smallertext{0}}^x=x.
\]
The law of $\mathsf P^x$ is the law of $P^x$ under
$\widehat{\mathbb P}^x$, so $\mathsf P^x$ remains in $\Delta^n$. Write
$\mathbb E^{\mathrm c}$ for expectation under $\mathbb P^{\mathrm c}$.
Standard BDG and Gr\"onwall estimates give, for every
$T<+\infty$, a finite constant $C_{\smallertext{T}}$, independent of
$(p,q)\in\Delta^n\times\Delta^n$, such that
\begin{equation}
\label{eq:tilted-SDE-stability}
  \mathbb E^{\mathrm c}\bigg[
    \sup_{0\leq s\leq {\smallertext{T}}}
    |\mathsf P_s^p-\mathsf P_s^q|^2
  \bigg]
  \leq C_{\smallertext{T}} |p-q|^2.
\end{equation}

For a continuous function $v$ on $\Delta^n$, write
\[
  \omega_v(\varepsilon)
  =\sup\big\{|v(x)-v(y)|:(x,y)\in\Delta^n\times\Delta^n,\; |x-y|\leq\varepsilon\big\},
  \; \varepsilon>0.
\]
For $x\in\Delta^n$, define the coupled functionals
\[
  \mathsf\Gamma_t^x
  \coloneqq\int_{\smallertext{0}}^t\zeta(\mathsf P_s^x)\mathrm ds,
  \;
  \mathsf J_t^x
  \coloneqq\int_{\smallertext{0}}^t
    \mathrm e^{-\mathsf\Gamma_s^x}\varphi(\mathsf P_s^x)^{-1}\mathrm ds,\; t\geq0.
\]
If
$R_\smallertext{T}^{p,q}\coloneqq
\sup_{0\leq s\leq \smallertext{T}}|\mathsf P_s^p-\mathsf P_s^q|$, then
\begin{align*}
  |\mathsf J_\smallertext{T}^p-\mathsf J_\smallertext{T}^q|
  &\leq T\omega_{1/\varphi}(R_\smallertext{T}^{p,q})
  +\frac{\|1/\varphi\|_\infty T^2}{2}
   \omega_\zeta(R_\smallertext{T}^{p,q}),
\end{align*}
where $1/\varphi$ denotes the map $y\longmapsto\varphi(y)^{-1}$. Since $\varphi$ is continuous and strictly positive on the
compact simplex and $\min_{q\in\Delta^\smalltext{n}}\zeta(q)>0$ by
assumption, set
$\underline\varphi\coloneqq\min_{q\in\Delta^\smalltext{n}}\varphi(q)>0$
and
$\underline\zeta\coloneqq\min_{q\in\Delta^\smalltext{n}}\zeta(q)>0$. Moreover,
uniformly in $p$
\begin{equation*}
  0\leq \mathsf J_\infty^p-\mathsf J_\smallertext{T}^p
  \leq
  \frac{\mathrm e^{-\underline\zeta \smallertext{T}}}
     {\underline\varphi \underline\zeta}.
\end{equation*}

The equality in law in the preceding paragraph and
\eqref{eq:critical-anchor-revised} give
$h(p)=\mathbb E^{\mathrm c}[\mathsf J_\infty^p]$. The stability estimate
\eqref{eq:tilted-SDE-stability} and the uniform continuity of $\zeta$ and
$1/\varphi$ show that
$p\longmapsto\mathbb E^{\mathrm c}[\mathsf J_\smallertext{T}^p]$ is
continuous for every finite $T$. The uniform tail bound therefore makes $h$
the uniform limit of continuous functions, and hence $h\in C(\Delta^n)$.

\medskip

Split now $J_\infty^p$ at a deterministic time $t\ge 0$ and write
\begin{equation}
\label{eq:J-restart-decomposition}
  J_\infty^p
  =J_t^p+\mathrm e^{-\smallertext{\Gamma}_\smalltext{t}^\smalltext{p}}\mathcal J_{t,\infty}^p,
\end{equation}
where
\[
  \mathcal J_{t,\infty}^p
  \coloneqq
  \int_{\smallertext{0}}^{\smallertext{+}\infty}
    \exp\bigg(-\int_{\smallertext{0}}^s\zeta(P_{t+r}^p) \mathrm dr\bigg)
    \varphi(P_{t+s}^p)^{-1} \mathrm ds.
\]
Conditionally on the state $P_t^p$, the Markov property identifies the law
of $\mathcal J_{t,\infty}^p$ with that of $J_\infty^{\smallertext{P}_\smalltext{t}^\smalltext{p}}$. More
precisely, for the Markov kernel the map
$q\longmapsto\widehat{\mathbb P}^{q}[J_\infty^q=h(q)]$ is Borel-measurable and equals
one for every $q\in\Delta^n$. Hence the identity
\eqref{eq:critical-anchor-revised} gives
$\mathcal J_{t,\infty}^p=h(P_t^p)$, $  \widehat{\mathbb P}^{p}\text{\rm--a.s.}$ Substitution in
\eqref{eq:J-restart-decomposition} yields
\begin{equation}
\label{eq:critical-anchor-restart}
  h(p)
  =J_t^p+\mathrm e^{-\smallertext{\Gamma}_\smalltext{t}^\smalltext{p}}h(P_t^p),
  \;
  \widehat{\mathbb P}^{p}\text{\rm--a.s.}
\end{equation}
Taking a common full-probability event for rational $t$ and using continuity
of $J$, $\Gamma$, $P$, and $h$ extends this identity to every time. Comparing
\eqref{eq:critical-anchor-restart} with \eqref{eq:H-explicit-revised}
gives
\begin{equation}
\label{eq:invariant-graph-revised}
  H_t^p=h(P_t^p),
  \; t\geq0,\;   \widehat{\mathbb P}^{p}\text{\rm--a.s.}
\end{equation}
On each finite horizon, the canonical law agrees with the law
under $\widehat{\mathbb P}^{p,\smallertext{T}}$; equivalence then transfers the identity
back to $\mathbb P^p$. Hence $h$ is continuous and the auxiliary process remains on
the invariant graph $H_t^p=h(P_t^p)$.

\medskip

\noindent\emph{Step $4$: $h$ is constant.} It remains to show that $h$ is constant. Fix an arbitrary $p\in\Delta_\circ^n$ and $s>0$. Because $h$, $\zeta$, and $\varphi^{-1}$ are continuous on
the compact simplex, the following constant is finite: 
\[
  C_{\smallertext{0}}
  \coloneqq
  \sup_{q\in\Delta^\smalltext{n}}
  \big|\zeta(q)h(q)-\varphi(q)^{-1}\big|<+\infty.
\]

The first identity in \eqref{eq:on-path-equilibrium-reduction} and
\eqref{eq:invariant-graph-revised} give
\[
  h(P_t^p)-h(p)
  =
  \int_{\smallertext{0}}^t
    \big(
     h(P_s^p)\zeta(P_s^p)-\varphi(P_s^p)^{-1}
    \big)\mathrm ds,\; t\geq 0,\; \widehat{\mathbb P}^p\text{\rm--a.s.}
\]

Consequently
\begin{equation}
\label{eq:h-short-time-revised}
  |h(P_t^p)-h(p)|\leq C_{\smallertext{0}}t, \; t\geq 0,\; \widehat{\mathbb P}^p\text{\rm--a.s.}
\end{equation}
Now we let $\bar p$ denote the unique
solution of
\begin{equation}
\label{eq:anchor-flow-ode}
  \dot{\bar p}(r)=A(\bar p(r)),
  \; \bar p(0)=p.
\end{equation}
For $N\in\mathbb{N}^\star$, let $T_\smallertext{N}\coloneqq s/N$ and define on
$\mathcal F_{\smallertext{T}_\smalltext{N}}^{\smallertext{D},p}$
\begin{equation*}
  \frac{\mathrm d\mathbb Q_\smallertext{N}}{\mathrm d\widehat{\mathbb P}^{p}}
  =\exp\bigg(
    -\sigma N\widehat W_{\smallertext{T}_\smalltext{N}}^p
    -\frac12\sigma^2N^2T_\smallertext{N}
   \bigg).
\end{equation*}
Under $\mathbb Q_\smallertext{N}$,
$W_t^\smallertext{N}=\widehat W_t^p+\sigma Nt$,
$0\leq t\leq T_\smallertext{N}$, is a Brownian motion. For
\[
  \widetilde P_r^\smallertext{N}\coloneqq P_{r/\smallertext{N}}^p,
  \;
  B_r^\smallertext{N}\coloneqq \sqrt N W_{r/\smallertext{N}}^\smallertext{N},
  \; 0\leq r\leq s,
\]
we obtain
\begin{equation*}
  \mathrm d\widetilde P_r^\smallertext{N}
  =\bigg(A\big(\widetilde P_r^\smallertext{N}\big)
      +\frac1N\widehat b\big(\widetilde P_r^\smallertext{N}\big)\bigg)\mathrm dr
   -\frac{A\big(\widetilde P_r^\smallertext{N}\big)}{\sigma\sqrt N} \mathrm dB_r^\smallertext{N}.
\end{equation*}

Write
$E_r^\smallertext{N}\coloneqq\widetilde P_r^\smallertext{N}-\bar p(r)$, choose $L_A<+\infty$ such that
$|A(x)-A(y)|\leq L_A|x-y|$ for every $(x,y)\in\Delta^n\times\Delta^n$, 
and set
$M_A\coloneqq\|A\|_\infty$ and
$M_b\coloneqq\|\widehat b\|_\infty$. Cauchy--Schwarz and Doob's
$L^2$ maximal inequality give, for $0\leq r\leq s$
\begin{align*}
  \mathbb E^{\mathbb Q_\smalltext{N}}\bigg[
    \sup_{0\leq u\leq r}|E_u^\smallertext{N}|^2
  \bigg]
  &\leq 3sL_A^2\int_{\smallertext{0}}^r
    \mathbb E^{\mathbb Q_\smalltext{N}}\bigg[
      \sup_{0\leq v\leq u}|E_v^\smallertext{N}|^2
    \bigg]\mathrm du
    +\frac{3s^2M_b^2}{N^2}
    +\frac{12sM_A^2}{\sigma^2N}.
\end{align*}
Hence Gr\"onwall's inequality yields
\begin{equation}
\label{eq:fast-convergence-revised}
  \mathbb E^{\mathbb Q_\smalltext{N}}\bigg[
    \sup_{0\leq r\leq s}\big|\widetilde P_r^\smallertext{N}-\bar p(r)\big|^2
  \bigg]
  \leq\frac{C_s}{N}.
\end{equation}
Here one may take
$C_s\coloneqq(3s^2M_b^2+12sM_A^2/\sigma^2)
\exp(3s^2L_A^2)$.
In particular, for every $\varepsilon>0$
\[
  \mathbb Q_\smallertext{N}\big[
    |\widetilde P_s^\smallertext{N}-\bar p(s)|>\varepsilon
  \big]
  \leq\frac{C_s}{N\varepsilon^2}\longrightarrow0
\]

Since $\mathbb Q_\smallertext{N}\ll\widehat{\mathbb P}^p$ on
$\mathcal F_{\smallertext{T}_\smalltext{N}}^{\smallertext{D},p}$, applying
\eqref{eq:h-short-time-revised} at $t=T_\smallertext{N}$ and using
$\widetilde P_s^\smallertext{N}=P_{\smallertext{T}_\smalltext{N}}^p$ gives
\begin{equation}
\label{eq:rescaled-h-bound}
  \big|h\big(\widetilde P_s^\smallertext{N}\big)-h(p)\big|
  \leq\frac{C_{\smallertext{0}}s}{N},
  \; \mathbb Q_\smallertext{N}\text{\rm--a.s.}
\end{equation}

Since $h$ is continuous on the compact simplex, it is bounded
and uniformly continuous. Fix $\varepsilon>0$ and choose $\delta>0$ such
that $|h(q)-h(q')|\leq\varepsilon$ whenever $|q-q'|\leq\delta$. Splitting
according to whether $|\widetilde P_s^\smallertext{N}-\bar p(s)|$ exceeds
$\delta$, and then applying Markov's inequality and
\eqref{eq:fast-convergence-revised}, gives
\[
  \mathbb E^{\mathbb Q_\smalltext{N}}\big[
   \big |h\big(\widetilde P_s^\smallertext{N}\big)-h(\bar p(s))\big|
  \big]
  \leq
  \varepsilon
  +
  2\|h\|_\infty
  \mathbb Q_\smallertext{N}\big[
    |\widetilde P_s^\smallertext{N}-\bar p(s)|>\delta
  \big]
  \leq
  \varepsilon
  +
  \frac{2\|h\|_\infty C_s}{N\delta^2}.
\]
For fixed $\varepsilon$ and the corresponding $\delta$, letting
$N\longrightarrow+\infty$ shows that the limit superior of the expectation
is at most $\varepsilon$. Since $\varepsilon$ is arbitrary, this proves
\[
  \mathbb E^{\mathbb Q_\smalltext{N}}\big[
    |h(\widetilde P_s^\smallertext{N})-h(\bar p(s))|
  \big]
  \underset{N\to+\infty}{\longrightarrow}0.
\]

Now we put everything together. For every outcome, the
triangle inequality through the intermediate value
$h(\widetilde P_s^\smallertext{N})$ gives
\[
  |h(\bar p(s))-h(p)|
  \leq
  |h(\bar p(s))-h(\widetilde P_s^\smallertext{N})|
  +
  |h(\widetilde P_s^\smallertext{N})-h(p)|.
\]
The left-hand side is deterministic. Taking
$\mathbb Q_\smallertext{N}$-expectations and using
\eqref{eq:rescaled-h-bound} for the second term gives
\[
  |h(\bar p(s))-h(p)|
  \leq
  \mathbb E^{\mathbb Q_\smalltext{N}}\big[
    |h(\bar p(s))-h(\widetilde P_s^\smallertext{N})|
  \big]
  +\frac{C_{\smallertext{0}}s}{N}.
\]
As $N\longrightarrow+\infty$, the first term on the right converges to zero
as shown above, and the second trivially does as well. The left-hand side does not depend on $N$, so it must be zero.
Because $p\in\Delta_\circ^n$ and $s>0$ were arbitrary, we have
$h(\bar p(s))=h(p)$ for every such $p$ and $s$.

\medskip
Finally, we use the deterministic flow to compare every
interior point $p$ with the same vertex. The solution of \eqref{eq:anchor-flow-ode} is explicit
\begin{equation*}
  \bar p^i(s)
  =\frac{p^i\mathrm e^{-g^\smalltext{i}s}}
      {\sum_{j=1}^np^j\mathrm e^{-g^\smalltext{j}s}},\; s\geq 0,
  \; i\in\{1,\ldots,n\}.
\end{equation*}

Let $i_\star$ be the unique index for which
$g^{i_\smalltext{\star}}=\min_{i\in\{1,\dots,n\}} g^i$. For every
$p\in\Delta_\circ^n$ and
$j\in\{1,\dots,n\}\setminus\{ i_\star\}$, we have
\[
  \frac{\bar p^j(s)}{\bar p^{i_\smalltext{\star}}(s)}
  =\frac{p^j}{p^{i_\smalltext{\star}}}
   \mathrm e^{-(g^\smalltext{j}-g^{\smalltext{i}_\smalltext{\star}})s}
  \underset{s\to\smallertext{+}\infty}{\longrightarrow}0.
\]
Thus $\bar p(s)\longrightarrow e_{i_\smalltext{\star}}$ as $s\longrightarrow+\infty$.
Since $h(\bar p(s))=h(p)$ for every $s>0$,
continuity gives
\[
  h(p)=h(e_{i_\smalltext{\star}}),
  \; p\in\Delta_\circ^n.
\]
Since $\Delta_\circ^n$ is dense in the closed simplex, another use of
continuity gives $h\equiv h_{\smallertext{0}}$ on $\Delta^n$, where
$h_{\smallertext{0}}=h(e_{i_\smalltext{\star}})>0$. \Cref{eq:invariant-graph-revised} then gives
$H^p\equiv h_{\smallertext{0}}$ for every $p\in\Delta^n$. Thus, the initial-value map $h$ and every auxiliary process
$H^p$ are constant.

\medskip

\noindent\emph{Step $5$: the pointwise equilibrium identity.} To identify the remaining equilibrium restriction, fix an arbitrary
$p\in\Delta^n$. Since $H^p\equiv h_{\smallertext{0}}$, the first identity in
\eqref{eq:on-path-equilibrium-reduction} gives
\[
  0
  =\frac{H_t^p-H_{\smallertext{0}}^p}{t}
  =\frac1t\int_{\smallertext{0}}^t
    \big(
     h_{\smallertext{0}}\zeta(P_s^p)-
     \varphi(P_s^p)^{-1}
    \big)\mathrm ds,\; t>0,\; \mathbb P^p\text{\rm--a.s.}
\]
The integrand is continuous at zero and
$P_{\smallertext{0}}^p=p$. Letting $t\downarrow0$
therefore gives
\[
  0=h_{\smallertext{0}}\zeta(p)-\varphi(p)^{-1},
  \; p\in\Delta^n,
\]
which is \eqref{eq:equilibrium-identity-short}. This proves necessity.

\medskip

\noindent\textbf{Part $2$: sufficiency.} We organise the argument into three
steps.

\medskip

\noindent\emph{Step $1$: the candidate value function and feedback controls.}

\medskip

Conversely, suppose that \Cref{eq:equilibrium-identity-short} and all the
conditions in \eqref{eq:converse-conditions} hold. Since
$\mathcal H(p,h_{\smallertext{0}})=0$, the constant process
$H^p\equiv h_{\smallertext{0}}$ on every model copy solves the
auxiliary equation. The uniqueness required in
\Cref{def:class-C} makes it the generated solution.
It remains to verify the candidate value and controls when $u=h_{\smallertext{0}}$.

\medskip
Let
\[
  a_\psi\coloneqq\frac{\theta}{\psi}=\theta+\gamma-1,
  \;
  \ell(p)\coloneqq
  \frac{\partial_p\varphi(p)}{\varphi(p)},
  \;
  d(p)\coloneqq A(p)^\top \ell(p),\; p\in\Delta^n.
\]
For $(p,x,u)\in\Delta^n\times(0,+\infty)^2$, define
\begin{equation*}
  w(p,x,u)
  \coloneqq
  \delta^\theta\frac{x^{1-\gamma}}{1-\gamma}
  \varphi(p)^{a_\psi}u^{a_\psi}.
\end{equation*}
Its derivatives are
\begin{gather*}
  w_x=(1-\gamma)\frac wx,\; 
  w_{xx}=-\gamma(1-\gamma)\frac w{x^2},\;
  w_u=a_\psi\frac wu,\; 
  \partial_pw=a_\psi w\ell,\\
  \partial_{pp}w=a_\psi w
    \bigg(
     \frac{\partial_{pp}\varphi}{\varphi}
     +(a_\psi-1)\ell\ell^\top
    \bigg),\;
  \partial_{xp}w=a_\psi(1-\gamma)\frac wx\ell.
\end{gather*}

At $u=h_{\smallertext{0}}$, direct substitution in \eqref{eq:hjb}, using
\eqref{eq:equilibrium-identity-short} and
\eqref{eq:converse-conditions}, shows that $w$ solves the HJB equation at $u=h_{\smallertext{0}}$ and
that the feedback controls are
\begin{equation*}
  c^\star(p,x,h_{\smallertext{0}})=\frac{x}{h_{\smallertext{0}}\varphi(p)},
  \;
  \pi^\star(p,x,h_{\smallertext{0}})=1,\; (p,x)\in\Delta^n\times(0,+\infty).
\end{equation*}

Indeed, the consumption maximiser follows from the first-order condition.
After division by $xw_x>0$, the portfolio-dependent expression, up to a
term independent of $\pi$, is
\[
  \pi\bigg(
    \mu(p,h_{\smallertext{0}})-\mathcal R(p,h_{\smallertext{0}})
    -\frac{a_\psi d(p)\nu(p)}{\sigma}
  \bigg)
  -\frac{\gamma}{2}\pi^2\nu(p)^2,
\]
where for $p\in\Delta^n$
\[
  \nu(p)\coloneqq \sigma-\frac{d(p)}{\sigma},
  \;
  \mu(p,h_{\smallertext{0}})\coloneqq\overline g(p)
   +\frac{\varphi(p)^{-1}}{h_{\smallertext{0}}}
   +\ell(p)^\top \big(b(p)-A(p)\big)
   +\frac{A(p)^\top\partial_{pp}\varphi(p)A(p)}
               {2\sigma^2\varphi(p)}.
\]
It is concave, and $\pi=1$ maximises it on $[0,1]$ precisely when
\[
  \mathcal R(p,h_{\smallertext{0}})
  \leq
  \mu(p,h_{\smallertext{0}})
  -\frac{a_\psi d(p)\nu(p)}{\sigma}
  -\gamma\nu(p)^2.
\]
Using \eqref{eq:equilibrium-identity-short}, expansion of
$\nu$, and the definition of $\zeta$, the right-hand side is exactly
$\overline r(p)$.

\medskip
Finally, substituting the feedback controls and dividing the remaining HJB
terms by $w$ gives
\begin{align*}
 &(1-\gamma)\overline g
 +\theta\zeta
 +\theta\ell^\top \big(b-(1-\gamma)A\big)
 +\frac{\theta}{2\sigma^2}
   \bigg(
    \frac{A^\top\partial_{pp}\varphi A}{\varphi}
    +(\theta-1)d^2
   \bigg)
 -\frac{\gamma(1-\gamma)\sigma^2}{2}
 =\delta\theta,
\end{align*}
where the equality is the definition of $\zeta$. This proves the HJB
equation and both feedback controls at $u=h_{\smallertext{0}}$. Thus, $w$ solves the HJB equation at $u=h_{\smallertext{0}}$ and yields the candidate
feedback controls $(\pi^\star,c^\star)$.

\medskip

\noindent\emph{Step $2$: verification estimates and admissibility.} For an arbitrary strategy $(\pi,c)\in\mathcal A(x;S^p,r^p)$, abbreviate
$\widehat Y^{p,x,\pi,c}$ and $\widehat Z^{p,x,\pi,c}$ to
$\widehat Y$ and $\widehat Z$, respectively. By definition,
\begin{equation*}
  \widehat Y_t
  \coloneqq
  (1-\gamma)\mathrm e^{-\delta\theta t}
  w(P_t^p,X_t^{p,x,\pi,c},h_{\smallertext{0}}),\; t\geq0,\; \mathbb P^p\text{\rm--a.s.}
\end{equation*}

It\^o's formula gives, on every finite interval,
\begin{equation*}
  \mathrm d\widehat Y_t
  =-\big(F_t(c_t,\widehat Y_t)+\widehat R_t\big) \mathrm dt
   +\widehat Z_t \mathrm d\overline W_t^p,
\end{equation*}
where the martingale integrand is
\begin{equation}
\label{eq:exact-Z}
  \widehat Z_t
  =(1-\gamma)\widehat Y_t
  \bigg(
    \pi_t\sigma
    -\frac{d(P_t^p)}{\sigma}
    \bigg(\pi_t+\frac1{\psi-1}\bigg)
  \bigg),\; \mathrm{d}t\otimes\mathbb P^p\text{\rm--a.e.}
\end{equation}

Define now
\begin{align}
\label{eq:HJB-residual}
  \widehat R_t
  \coloneqq(1-\gamma)\mathrm e^{-\delta\theta t}
  \bigg(
    \delta\theta w-\mathcal L^{\pi_t,c_t}w
    -\frac{\delta (c_t)^{1-1/\psi}}{1-1/\psi}
    \big((1-\gamma)w\big)^{1-1/\theta}
  \bigg)
  \big(P_t^p,X_t^{p,x,\pi,c},h_{\smallertext{0}}\big),\; t\geq0,\; \mathbb P^p\text{\rm--a.s.}
\end{align}

The HJB inequality implies
\begin{equation}
\label{eq:HJB-residual-sign}
  \operatorname{sgn}(1-\gamma)\widehat R_t\geq0,
  \; \mathrm dt\otimes\mathbb P^p\text{\rm--a.e.},
\end{equation}
and the residual vanishes at the feedback controls. The price--dividend ratio is
\[
  Q_t^p=h_{\smallertext{0}}\varphi(P_t^p),\; t\geq0,\; \mathbb P^p\text{\rm--a.s.}
\]
It is bounded above and away from zero. Let $\mathcal M^p$ be the process defined by
\begin{equation*}
  \mathcal M_t^p
  \coloneqq
  \mathrm e^{-\delta\theta t}(D_t^p)^{1-\gamma}(Q_t^p)^\theta,\; t\geq0,\; \mathbb P^p\text{\rm--a.s.}
\end{equation*}

It satisfies
\begin{equation}
\label{eq:M-dynamics-revised}
  \mathrm d\mathcal M_t^p
  =-\frac{\theta}{Q_t^p}\mathcal M_t^p \mathrm dt
   +\lambda_\smallertext{U}(P_t^p)\mathcal M_t^p \mathrm d\overline W_t^p.
\end{equation}

\medskip

Let
\[
  \overline Q
  \coloneqq
  h_{\smallertext{0}}\max_{q\in\Delta^\smalltext{n}}\varphi(q)<+\infty,
  \;
  m\coloneqq\frac{\theta}{\overline Q}>0,
  \;
  K_\smallertext{U}\coloneqq\|\lambda_\smallertext{U}\|_\infty.
\]
The explicit solution of \eqref{eq:M-dynamics-revised} is
\[
  \mathcal M_t^p
  =\mathcal M_{\smallertext{0}}^p
   \exp\bigg(-\int_{\smallertext{0}}^t\frac{\theta}{Q_s^p} \mathrm ds\bigg)
   \mathcal E\bigg(
    \int_{\smallertext{0}}^\cdot\lambda_\smallertext{U}(P_s^p) 
    \mathrm d\overline W_s^p
   \bigg)_t,\; t\geq0,\; \mathbb P^p\text{\rm--a.s.}
\]
The stochastic exponential is a true martingale because its integrand is
bounded. Since $\theta/Q_s^p\geq m$, taking expectations gives
\[
  \mathbb E^p[\mathcal M_t^p]
  \leq\mathcal M_{\smallertext{0}}^p\mathrm e^{-mt}.
\]

For $k\in\mathbb N$ and $t\in[k,k+1]$, set
\[
  \mathcal E_{k,t}^p
  \coloneqq
  \exp\bigg(
    \int_k^t\lambda_\smallertext{U}(P_s^p)\mathrm d\overline W_s^p
    -\frac12\int_k^t\lambda_\smallertext{U}(P_s^p)^2\mathrm ds
  \bigg).
\]
The solution restarted at time $k$ satisfies
$\mathcal M_t^p\leq\mathcal M_k^p\mathcal E_{k,t}^p$, while conditional
Cauchy--Schwarz and Doob's $L^2$ maximal inequality give
\[
  \mathbb E^p\bigg[
    \sup_{k\leq t\leq k+1}\mathcal E_{k,t}^p
    \bigg|\mathcal F_k^{\smallertext{D},p}
  \bigg]
  \leq
  2\mathbb E^p\bigg[
    \big(\mathcal E_{k,k+1}^p\big)^2
    \bigg|\mathcal F_k^{\smallertext{D},p}
  \bigg]^{1/2},
\]
hence
\[
  \mathbb E^p\bigg[
    \sup_{k\leq t\leq k+1}\mathcal M_t^p
     \bigg| \mathcal F_k^{\smallertext{D},p}
  \bigg]
  \leq C_{\smallertext{K}_\smalltext{U}}\mathcal M_k^p,
\]
where one may take
$C_{\smallertext{K}_\smalltext{U}}\coloneqq 2\mathrm e^{\smallertext{K}_\smalltext{U}^\smalltext{2}/2}$.
Indeed, the conditional
second moment of the stochastic exponential over an interval of length one
is bounded by $\mathrm e^{\smallertext{K}_\smalltext{U}^\smalltext{2}}$. Consequently
\begin{align*}
  \mathbb E^p\bigg[
    \bigg(\int_{\smallertext{0}}^{\smallertext{+}\infty}(\mathcal M_t^p)^2 \mathrm dt\bigg)^{1/2}
  \bigg]
  &\leq
  \sum_{k=0}^{\smallertext{+}\infty}
    \mathbb E^p\bigg[\sup_{k\leq t\leq k+1}\mathcal M_t^p\bigg]\leq
  C_{\smallertext{K}_\smalltext{U}}\sum_{k=0}^{\smallertext{+}\infty}\mathbb E^p[\mathcal M_k^p]
  \leq
  \frac{C_{\smallertext{K}_\smalltext{U}}}{1-\mathrm e^{-m}}\mathcal M_{\smallertext{0}}^p.
\end{align*}
Thus there is a constant $C>0$ such that for any $t\geq 0$
\begin{equation}
\label{eq:M-estimates-revised}
  \mathbb E^p[\mathcal M_t^p]
  \leq\mathcal M_{\smallertext{0}}^p\mathrm e^{-mt},
  \;
  \mathbb E^p\bigg[
    \bigg(\int_{\smallertext{0}}^{\smallertext{+}\infty}(\mathcal M_t^p)^2 \mathrm dt\bigg)^{1/2}
  \bigg]
  \leq C\mathcal M_{\smallertext{0}}^p.
\end{equation}

The process $\mathcal M^p$ is a positive super-martingale and therefore has
an almost-sure limit. The first estimate implies that this limit is zero.

\medskip
Let $\lambda_x\coloneqq x/S_{\smallertext{0}}^p$, $x>0$. Define the feedback controls by
\begin{equation*}
  \pi_t^\star\coloneqq1,
  \;
  c_t^{p,\star}=\lambda_xD_t^p,\; t\geq0, \; \mathbb P^p\text{\rm--a.s.}
\end{equation*}
Direct substitution in the wealth equation gives
\begin{equation*}
  X_t^{p,x,\pi^\star,c^{p,\star}}=\lambda_xS_t^p,
  \; t\geq0, \; \mathbb P^p\text{\rm--a.s.}
\end{equation*}
The associated verification process and exact martingale integrand are
\begin{align*}
  \widehat Y_t^{p,x,\pi^\star,c^{p,\star}}
  &=\delta^\theta\lambda_x^{1-\gamma}\mathcal M_t^p,\;
  \widehat Z_t^{p,x,\pi^\star,c^{p,\star}}
  =\lambda_\smallertext{U}(P_t^p)
   \widehat Y_t^{p,x,\pi^\star,c^{p,\star}},\; t\geq0,
   \; \mathbb P^p\text{\rm--a.s.}
\end{align*}
The estimates in \eqref{eq:M-estimates-revised} and boundedness of
$\lambda_\smallertext{U}$ give
$\widehat Z^{p,x,\pi^\star,c^{p,\star}}\in\mathbb H^1(\mathbb P^p)$,
while supermartingale convergence gives
$\liminf_{t\to\smallertext{+}\infty}\widehat Y_t^{p,x,\pi^\star,c^{p,\star}}=0$.
The residual is zero and
$\widehat Y_{\smallertext{0}}^{p,x,\pi^\star,c^{p,\star}}$ is deterministic, so
\Cref{lem:zero-residual} shows that this pair solves the utility BSDE. Since
$0<\theta\leq1$, zero-residual comparison gives uniqueness. Hence
$c^{p,\star}\in\mathcal C_a^p$ and the feedback strategy satisfies all three
conditions in \Cref{def:admissibility}. Thus, the required verification estimates hold and the candidate
feedback controls are admissible for every initial wealth $x>0$.

\medskip
\noindent\emph{Step $3$: optimality and market clearing.} For an arbitrary $(\pi,c)\in\mathcal A(x;S^p,r^p)$, the exact integrand
in \eqref{eq:exact-Z} belongs to $\mathbb H^1(\mathbb P^p)$ and the terminal condition
holds by admissibility. To verify that the residual is pathwise locally Lebesgue-integrable, let $q=1-1/\theta\leq0$. On every
$[0,T]$, the ratio of the positive continuous processes $\widehat Y$ and
$Y^{p,c}$ is bounded above and away from zero, and hence
\[
  \int_{\smallertext{0}}^{\smallertext{T}} F_s\big(c_s,\widehat Y_s\big) \mathrm ds
  =
  \int_{\smallertext{0}}^{\smallertext{T}} F_s\big(c_s,Y_s^{p,c}\big)
    \bigg(\frac{\widehat Y_s}{Y_s^{p,c}}\bigg)^q\mathrm ds
  <+\infty,
  \; \mathbb P^p\text{\rm--a.s.}
\]
The remaining finite-variation coefficient in It\^o's formula is pathwise locally
Lebesgue-integrable, so the residual in \eqref{eq:HJB-residual} is pathwise locally Lebesgue-integrable. The HJB inequality at $u=h_{\smallertext{0}}$, expressed in
\eqref{eq:HJB-residual-sign}, then allows us to apply
\Cref{lem:one-sided-comparison} with
$\rho=\operatorname{sgn}(1-\gamma)$. It yields
\[
  V_t^{p,\pi,c}(x;S^p,r^p)
  \leq
  w(P_t^p,X_t^{p,x,\pi,c},h_{\smallertext{0}}),\; t\geq0, \; \mathbb P^p\text{\rm--a.s.}
\]
At the feedback controls the residual vanishes, so equality holds. This
proves optimality, \eqref{eq:value-short}, and \eqref{eq:controls-short}.
At $x=S_{\smallertext{0}}^p$, one has
$X^{p,S_{\smallertext{0}}^p,\pi^\star,c^{p,\star}}=S^p$, $\pi^\star=1$, and
$c^{p,\star}=D^p$;
both markets therefore clear. This proves sufficiency and completes the
proof.
\end{proof}

The next proposition shows that the strict lower bound on
$\zeta$ is redundant in the irreducible two-state model.

\begin{proposition}[Endogenous positivity in the irreducible two-state model]
\label{prop:two-state-endogenous-zeta}
Let {\rm\Cref{assum:pref}} hold. Suppose that $n=2$,
$\lambda^{\smallertext{1}\smallertext{2}}>0$, and
$\lambda^{\smallertext{2}\smallertext{1}}>0$. If a specification is a
$\mathfrak C$-equilibrium for which the agent's value function admits a
classical $C^2$ state representation, then
\[
  \min_{p\in\Delta^\smalltext{2}}\zeta(p)>0.
\]
Consequently, the necessity conclusion of {\rm\Cref{thm:main}} holds
without separately assuming this lower bound.
\end{proposition}

\begin{proof}
The derivation of \eqref{eq:H-explicit-revised} does not use the lower
bound on $\zeta$. Fix $p_{\smallertext{0}}\in\Delta_\circ^2$. Since
$A(q)\neq0$ for every $q\in\Delta_\circ^2$, the two-state belief diffusion
is non-degenerate in the interior. Its law therefore has full support on
the continuous paths in $\Delta_\circ^2$ starting from
$p_{\smallertext{0}}$; see \citeauthor*{ikeda1989stochastic}
\cite[Theorem~VI.8.1]{ikeda1989stochastic}. In particular, for
every non-empty open set $U\subset\Delta_\circ^2$ and every
$T>0$
\[
  \widehat{\mathbb P}^{p_\smalltext{0}}
  \big[
    P_s^{p_\smalltext{0}}\in U,\; 
    \text{\rm for every}\; s\in[1,T+1]
  \big]
  >0.
\]

By \eqref{eq:H-explicit-revised} and the positivity of $H^{p_\smalltext{0}}$
required in {\rm\Cref{def:class-C}}
\begin{equation}
\label{eq:two-state-finite-J-bound}
  J_t^{p_\smalltext{0}}<h(p_{\smallertext{0}}),
  \; t\geq0,
  \; \widehat{\mathbb P}^{p_\smalltext{0}}\text{\rm--a.s.}
\end{equation}

\medskip

Suppose that $\min_{p\in\Delta^2}\zeta(p)\leq0$. By continuity, for every
$T>0$ there is a non-empty open set
$U_\smallertext{T}\subset\Delta_\circ^2$ on which
$\zeta\leq1/T$. Hence the event
\[
  E_\smallertext{T}
  \coloneqq
  \big\{P_s^{p_\smalltext{0}}\in U_\smallertext{T},\;
  \text{\rm for every}\; s\in[1,T+1]\big\},
\]
has positive $\widehat{\mathbb P}^{p_\smalltext{0}}$-probability. Set
$K\coloneqq\|\zeta\|_\infty$. On $E_\smallertext{T}$
\[
  \Gamma_{1+u}^{p_\smalltext{0}}
  \leq K+\frac{u}{T},
  \; 0\leq u\leq T,
\]
and consequently
\[
  J_{\smallertext{T}+1}^{p_\smalltext{0}}
  \geq
  \frac{\mathrm e^{-\smallertext{K}}}{\|\varphi\|_\infty}
  \int_{\smallertext{0}}^\smallertext{T}\mathrm e^{-u/\smallertext{T}}\mathrm du
  =
  \frac{\mathrm e^{-\smallertext{K}}(1-\mathrm e^{-1})}{\|\varphi\|_\infty}T.
\]
Choose a finite $T$ for which the right-hand side exceeds
$h(p_{\smallertext{0}})$. Since $E_\smallertext{T}$ has positive probability, this
contradicts the almost sure bound
\eqref{eq:two-state-finite-J-bound}. Thus $\zeta>0$ on $\Delta^2$, and
compactness gives $\min_{p\in\Delta^2}\zeta(p)>0$. The necessity proof of
\Cref{thm:main} now applies. In particular, $H$ is constant and its drift
vanishes through $h_{\smallertext{0}}\varphi\zeta=1$.
\end{proof}

\begin{remark}
This argument is specific to two states. There the simplex is
one-dimensional and the scalar belief diffusion is non-degenerate in its
interior, so its support contains paths that remain near any given belief.
For $n>2$, the belief diffusion has only one diffusion direction on the
$(n-1)$-dimensional simplex and need not remain near an arbitrary point.
\end{remark}

\section{Proofs for the two-state Markovian analysis}
\label{app:proof-two-state-analysis}

This appendix contains the proofs of the analytical results stated in
\Cref{sec:markovian-shape}.

\subsection{Existence and uniqueness}

The following lemma establishes positivity and compactness of the linear
resolvent $R$. Existence then follows from Schauder's theorem applied to a
bounded truncation of $\Phi\longmapsto R(\theta\Phi^q)$.

\medskip

\begin{lemma}[Linear degenerate resolvent]
\label{lem:markovian-linear-resolvent}
Assume that $\lambda^{\smallertext{1}\smallertext{2}}>0$ and $\lambda^{\smallertext{2}\smallertext{1}}>0$. For every
$f\in C([0,1])$, there exists a unique function $u_f\in C^1([0,1])\cap C^2((0,1))$ such that
\begin{equation}
\label{eq:linear-degenerate-resolvent}
  -\alpha(p)u_f^{\prime\prime}(p)-\beta(p)u_f^{\prime}(p)+c(p)u_f(p)=f(p),
  \; p\in(0,1),
\end{equation}
and
\begin{equation*}
\begin{aligned}
  -\lambda^{\smallertext{2}\smallertext{1}}u_f^{\prime}(0)+c(0)u_f(0)&=f(0),\\
  \lambda^{\smallertext{1}\smallertext{2}}u_f^{\prime}(1)+c(1)u_f(1)&=f(1).
\end{aligned}
\end{equation*}
If $Rf\coloneqq u_f$, then $R:C([0,1])\longrightarrow C([0,1])$ is linear,
positive, continuous, and compact. It satisfies
\begin{equation*}
 \|Rf\|_\infty
 \le
 \frac{\|f\|_\infty}{c_\smallertext{-}}.
\end{equation*}
Moreover, there exists a constant $C_\smallertext{R}>0$, depending only on the coefficients, such
that
\begin{equation*}
  \|(Rf)^{\prime}\|_\infty\le C_\smallertext{R}\|f\|_\infty.
\end{equation*}
\end{lemma}

\begin{proof}
Let
\[
  a_{\smallertext{0}}\coloneqq\frac{(\Delta g)^2}{2\sigma^2},
  \;
  \alpha(p)=a_{\smallertext{0}}p^2(1-p)^2.
\]
The proof proceeds in four steps.

\medskip

\noindent\emph{Step $1$: regularised problems.}
For $\varepsilon\in(0,1]$, set
$\alpha_\varepsilon=\alpha+\varepsilon$ and consider
\begin{equation*}
\begin{cases}
  -\alpha_\varepsilon u_\varepsilon^{\prime\prime}
  -\beta u_\varepsilon^{\prime}+cu_\varepsilon=f,
    &p\in(0,1),\\
  -\lambda^{\smallertext{2}\smallertext{1}}u_\varepsilon^{\prime}(0)
    +c(0)u_\varepsilon(0)=f(0),\\
  \lambda^{\smallertext{1}\smallertext{2}}u_\varepsilon^{\prime}(1)
    +c(1)u_\varepsilon(1)=f(1).
\end{cases}
\end{equation*}
Since $\alpha_\varepsilon$ is bounded away from zero and $f$ is
continuous, the associated first-order system has a unique $C^1$
solution for each initial pair; its first component is $C^2$.

\medskip

Uniqueness for the boundary-value problem follows from the maximum
principle. A nonzero homogeneous solution cannot have a positive maximum:
at an interior maximum,
$-\alpha_\varepsilon u^{\prime\prime}+cu>0$; at $0$,
$-\lambda^{\smallertext{2}\smallertext{1}}u^{\prime}(0)+c(0)u(0)>0$; and the argument at $1$ is
identical. Applying the same reasoning to $-u$ excludes a negative
minimum.

\medskip

Existence can also be shown directly by shooting. Given
$a=u_\varepsilon(0)$, the left endpoint equation fixes
\[
  u_\varepsilon^{\prime}(0)
  =
  \frac{c(0)a-f(0)}{\lambda^{\smallertext{2}\smallertext{1}}}.
\]
Let $u_{\varepsilon,a}$ be the corresponding initial-value solution and
define its right-boundary residual by
\[
  F_\varepsilon(a)
  \coloneqq
  \lambda^{\smallertext{1}\smallertext{2}}u_{\varepsilon,a}^{\prime}(1)
  +c(1)u_{\varepsilon,a}(1)-f(1).
\]
Linearity of the equation makes $F_\varepsilon$ affine. Its slope is the
right-boundary residual of the homogeneous solution with
$u(0)=1$ and $u^{\prime}(0)=c(0)/\lambda^{\smallertext{2}\smallertext{1}}$. If that slope were
zero, this nonzero solution would satisfy both homogeneous endpoint
relations, contradicting uniqueness. Thus $F_\varepsilon$ has exactly
one zero.

\medskip

The same maximum-principle argument gives
\begin{equation*}
  \|u_\varepsilon\|_\infty
  \leq
  \frac{\|f\|_\infty}{c_\smallertext{-}}.
\end{equation*}

\medskip

\noindent\emph{Step $2$: an $\varepsilon$-uniform derivative bound.}
Set $v_\varepsilon=u_\varepsilon^{\prime}$ and
$h_\varepsilon=cu_\varepsilon-f$. Then
\[
  \|h_\varepsilon\|_\infty
  \leq
  \bigg(1+\frac{c_\smallertext{+}}{c_\smallertext{-}}\bigg)\|f\|_\infty
  \eqqcolon H_f,
\]
and the endpoint equations imply
\[
  |v_\varepsilon(0)|
  \leq\frac{H_f}{\lambda^{\smallertext{2}\smallertext{1}}},
  \;
  |v_\varepsilon(1)|
  \leq\frac{H_f}{\lambda^{\smallertext{1}\smallertext{2}}}.
\]
Choose $r\in(0,1/4)$ and $b_{\smallertext{0}},b_{\smallertext{1}}>0$ such that
$\beta\geq b_{\smallertext{0}}$ on $[0,r]$ and
$-\beta\geq b_{\smallertext{1}}$ on $[1-r,1]$. On the left interval,
\[
  \alpha_\varepsilon v_\varepsilon^{\prime}
  +\beta v_\varepsilon=h_\varepsilon.
\]
Variation of constants, with
\[
  E_\varepsilon(p,s)
  \coloneqq
  \exp\bigg(
    -\int_s^p\frac{\beta(z)}{\alpha_\varepsilon(z)} \mathrm dz
  \bigg),
\]
gives
\[
  |v_\varepsilon(p)|
  \leq
  |v_\varepsilon(0)|
  +
  H_f\int_{\smallertext{0}}^p
    \frac{E_\varepsilon(p,s)}{\alpha_\varepsilon(s)} \mathrm ds
  \leq
  H_f\bigg(\frac1{\lambda^{\smallertext{2}\smallertext{1}}}+\frac1{b_{\smallertext{0}}}\bigg).
\]
Reflection about $p=1/2$ yields the analogous bound on $[1-r,1]$.
On $[r,1-r]$, $\alpha$ is bounded away from zero. The mean-value
theorem supplies one point at which $v_\varepsilon$ is bounded by a
multiple of $\|u_\varepsilon\|_\infty$, and variation of constants from
that point bounds $v_\varepsilon$ on the whole interior interval.
Consequently, for a coefficient-dependent constant $C_\smallertext{R}$,
\begin{equation}
\label{eq:uniform-C1-regularised}
  \|u_\varepsilon^{\prime}\|_\infty
  \leq
  C_\smallertext{R}\|f\|_\infty,
  \;
  \varepsilon\in(0,1].
\end{equation}

\medskip

\noindent\emph{Step $3$: the degenerate limit and endpoint traces.}
The bounds above make $(u_\varepsilon)$ uniformly bounded and
equi-Lipschitz. The Arzelà--Ascoli theorem therefore gives a sequence
$\varepsilon_n\downarrow0$ and $u\in C([0,1])$ such that
\[
  u_{\varepsilon_n}\longrightarrow u
  \;\text{uniformly on }[0,1].
\]
On every $J\Subset(0,1)$, the equation bounds
$u_{\varepsilon_n}^{\prime\prime}$ uniformly. Arzelà--Ascoli applied to
the derivatives gives, after a diagonal subsequence,
$u_{\varepsilon_n}\to u$ in $C^1(J)$; substituting into the equation
then gives convergence in $C^2(J)$. Hence
\[
  u\in C([0,1])\cap C^2((0,1))
\]
and $u$ solves \eqref{eq:linear-degenerate-resolvent} in the interior.
If $f\in C^k_{\mathrm{loc}}((0,1))$, differentiation of the uniformly
elliptic interior equation gives
$u\in C^{k+2}_{\mathrm{loc}}((0,1))$. No higher regularity is asserted
here when the forcing term $f$ is merely continuous.

\medskip

To identify the derivative trace at $0$, set $g_f=cu-f$ and
\[
  \rho_{\smallertext{0}}(p)
  \coloneqq
  \exp\bigg(
    -\int_p^r\frac{\beta(s)}{\alpha(s)} \mathrm ds
  \bigg).
\]
Because $\alpha(p)\sim a_{\smallertext{0}}p^2$ and
$\beta(p)\to\lambda^{\smallertext{2}\smallertext{1}}>0$,
there exist $C_\rho,K_\rho>0$ such that
$\rho_{\smallertext{0}}(p)\leq C_\rho\mathrm e^{-K_\rho/p}$ near zero. Rewriting the equation as
$\alpha u^{\prime\prime}+\beta u^{\prime}=g_f$ gives
\[
  (\rho_{\smallertext{0}}u^{\prime})^{\prime}
  =
  \rho_{\smallertext{0}}\frac{g_f}{\alpha}.
\]
Continuity of $u$ forces
$\lim_{p\downarrow0}\rho_{\smallertext{0}}(p)u^{\prime}(p)=0$; otherwise
$|u^{\prime}|$ would dominate the nonintegrable function
$1/\rho_{\smallertext{0}}$. Therefore
\[
  u^{\prime}(p)
  =
  \frac1{\rho_{\smallertext{0}}(p)}
  \int_{\smallertext{0}}^p
    \rho_{\smallertext{0}}(s)\frac{g_f(s)}{\alpha(s)} \mathrm ds.
\]
Since $\rho_{\smallertext{0}}^{\prime}=\rho_{\smallertext{0}}\beta/\alpha$, l'Hôpital's rule yields
\[
  u^{\prime}(0+)
  =
  \frac{c(0)u(0)-f(0)}{\lambda^{\smallertext{2}\smallertext{1}}}.
\]
The reflected argument at $1$ gives
\[
  u^{\prime}(1-)
  =
  \frac{f(1)-c(1)u(1)}{\lambda^{\smallertext{1}\smallertext{2}}}.
\]
Thus $u\in C^1([0,1])\cap C^2((0,1))$ and satisfies both endpoint
relations.

\medskip

\noindent\emph{Step $4$: properties of the resolvent.}
First, the maximum-principle argument from Step $1$, now using the
degenerate endpoint relations, proves uniqueness. Hence every convergent
subsequence of $u_\varepsilon$ has the same limit $u=Rf$, and the whole
family converges.

\medskip

Second, if $f\geq0$, a negative minimum of $u$, whether interior or at
an endpoint, contradicts the corresponding equation. Thus $R$ is
positive. Applying the maximum principle to $u$ and $-u$ gives
\[
  \|Rf\|_\infty\leq\frac{\|f\|_\infty}{c_\smallertext{-}}.
\]
Linearity follows from uniqueness, and this estimate proves continuity of
$R$ in the supremum norm.

\medskip

Finally, the uniform estimate
\eqref{eq:uniform-C1-regularised} passes to the limit:
\[
  \|(Rf)^{\prime}\|_\infty
  \leq C_\smallertext{R}\|f\|_\infty.
\]
The image of the unit ball of $C([0,1])$ is therefore uniformly bounded
and equi-Lipschitz. A final application of the Arzelà--Ascoli theorem
shows that this image is relatively compact in $C([0,1])$, so $R$ is
compact.
\end{proof}

\begin{lemma}[Regularity scale for the linear degenerate resolvent]
\label{lem:linear-regularity-scale}
Assume $\lambda^{\smallertext{1}\smallertext{2}}>0$ and $\lambda^{\smallertext{2}\smallertext{1}}>0$, let $k\geq1$ be an
integer, and, for $f\in C([0,1])$, let $u_f=Rf$ be the solution from
{\rm\Cref{lem:markovian-linear-resolvent}}.
\begin{enumerate}
  \item[$(i)$] If $f\in C^k((0,1))$, then
  $u_f\in C^{k+2}((0,1))$. In particular, $u_f$ is smooth on
  $(0,1)$ whenever $f$ is smooth there, and $u_f$ is real
  analytic on $(0,1)$ whenever $f$ is real analytic there.
  \item[$(ii)$] If $f\in C^k([0,1])$, then
  $u_f\in C^{k+1}([0,1])$. Moreover, set
  \[
    h_f\coloneqq cu_f-f,\;
    h_{f,0}\coloneqq h_f,\;
    \beta_j\coloneqq j\alpha^{\prime}+\beta,
  \]
  and, for $j\geq1$, define
  \begin{equation}
  \label{eq:hfj-def}
    h_{f,j}
    \coloneqq
    h_f^{(j)}
    -\sum_{i=2}^{j}\binom{j}{i}
      \alpha^{(i)}u_f^{(j+2-i)}
    -\sum_{i=1}^{j}\binom{j}{i}
      \beta^{(i)}u_f^{(j+1-i)}.
  \end{equation}
  Then
  \[
    \lambda^{\smallertext{2}\smallertext{1}}u_f^{(j+1)}(0)=h_{f,j}(0),
    \;
    \lambda^{\smallertext{1}\smallertext{2}}u_f^{(j+1)}(1)=-h_{f,j}(1),
    \; j\in\{0,\dots,k\}.
  \]
\end{enumerate}
\end{lemma}

\begin{proof}
\emph{Interior regularity.}
On $(0,1)$, $\alpha>0$, and
\eqref{eq:linear-degenerate-resolvent} can be rearranged as
\begin{equation}
\label{eq:interior-rearranged}
  u_f^{\prime\prime}
  =
  \frac{cu_f-\beta u_f^{\prime}-f}{\alpha}.
\end{equation}
Suppose $u_f\in C^j((0,1))$ for some
$j\in\{2,\dots,k+1\}$. The right-hand side of
\eqref{eq:interior-rearranged} then belongs to
$C^{j-1}((0,1))$, and hence $u_f\in C^{j+1}((0,1))$.
Starting from $u_f\in C^2((0,1))$ and iterating proves
$u_f\in C^{k+2}((0,1))$.

\medskip

If $f$ is real analytic, the right-hand side of
\eqref{eq:interior-rearranged}, viewed as a function of
$(p,u_f,u_f^{\prime})$, is real analytic near every interior point.
The local analytic-solvability theorem for ordinary differential
equations, together with uniqueness of the Cauchy problem, therefore
implies that $u_f$ is real analytic on $(0,1)$; see, for example,
\cite[Chapter~1]{coddington1955theory}.

\medskip

\emph{Regularity at the endpoints.}
We prove by induction over $j\in\{0,\dots,k\}$ that
$u_f^{(j+1)}$ extends continuously to $[0,1]$ and satisfies the
displayed endpoint identities. The case $j=0$ is Step~$3$ of the
proof of \Cref{lem:markovian-linear-resolvent}.

\medskip

Let $j\geq1$, suppose the assertion holds up to order $j-1$, and set
$v\coloneqq u_f^{(j)}$. Differentiating
\[
  \alpha u_f^{\prime\prime}+\beta u_f^{\prime}=h_f
\]
exactly $j$ times and collecting the two highest derivatives gives
\begin{equation}
\label{eq:leibniz-reduced}
  \alpha v^{\prime\prime}+\beta_jv^{\prime}=h_{f,j},
  \; p\in(0,1).
\end{equation}
All derivatives of $u_f$ appearing in $h_{f,j}$ have order at most
$j$. The induction hypothesis and $f\in C^k([0,1])$ therefore give
$h_{f,j}\in C([0,1])$. Since
$\alpha^{\prime}(0)=\alpha^{\prime}(1)=0$,
\[
  \beta_j(0)=\lambda^{\smallertext{2}\smallertext{1}}>0,
  \;
  \beta_j(1)=-\lambda^{\smallertext{1}\smallertext{2}}<0.
\]
Thus \eqref{eq:leibniz-reduced} has the same boundary structure as the
equation used to identify the first derivative traces.

\medskip

To make this explicit at $0$, choose $r_j>0$ such that
$\beta_j>0$ on $[0,r_j]$, and set
\[
  \rho_j(p)
  \coloneqq
  \exp\bigg(
    -\int_p^{r_j}\frac{\beta_j(s)}{\alpha(s)}\mathrm ds
  \bigg),
  \; p\in(0,r_j].
\]
Because $\alpha(p)$ vanishes quadratically and
$\beta_j(0)=\lambda^{\smallertext{2}\smallertext{1}}>0$, there exist
constants $C_j,K_j>0$ such that
$\rho_j(p)\leq C_j\mathrm e^{-K_j/p}$ near $0$. Hence
$\rho_j/\alpha$ is integrable there, whereas $1/\rho_j$ is not.
Equation \eqref{eq:leibniz-reduced} yields
\[
  (\rho_jv^{\prime})^{\prime}
  =
  \rho_j\frac{h_{f,j}}{\alpha}.
\]
The limit of $\rho_jv^{\prime}$ at $0$ must be zero; otherwise
$v^{\prime}$ would have a nonintegrable fixed-sign singularity,
contradicting continuity of $v$ at the endpoint. Therefore
\[
  v^{\prime}(p)
  =
  \frac1{\rho_j(p)}
  \int_{\smallertext{0}}^p
    \rho_j(s)\frac{h_{f,j}(s)}{\alpha(s)}\mathrm ds.
\]
Since $\rho_j^{\prime}=\rho_j\beta_j/\alpha$, l'Hôpital's rule gives
\[
  v^{\prime}(0+)
  =
  \frac{h_{f,j}(0)}{\beta_j(0)}
  =
  \frac{h_{f,j}(0)}{\lambda^{\smallertext{2}\smallertext{1}}}.
\]
Thus $u_f^{(j+1)}$ extends continuously to $0$ and satisfies the
claimed identity there. Applying the same argument to $u_f(1-\cdot)$ gives
\[
  \lambda^{\smallertext{1}\smallertext{2}}u_f^{(j+1)}(1)=-h_{f,j}(1).
\]
This closes the induction and proves the result.
\end{proof}

\begin{proof}[Proof of \Cref{thm:markovian-existence}]
Since
\[
 q-1=-\frac1\theta<0,
\]
we may choose constants $0<m<M$ such that
\[
 \theta m^{q-1}>c_\smallertext{+},
 \;
 \theta M^{q-1}<c_\smallertext{-}.
\]
Let
\[
  \mathcal T(y)\coloneqq \min\{M,\max\{m,y\}\},
  \;
  N(y)\coloneqq \theta\mathcal T(y)^q.
\]
The truncation $\mathcal T$ is Lipschitz continuous with values in
$[m,M]$, and $y\longmapsto\theta y^q$ is Lipschitz continuous on that
interval. Hence $N$ induces a continuous map on $C([0,1])$ that
maps bounded sets into bounded sets. Since $R$ is continuous and
compact, define the continuous compact map
\begin{align*}
  T:C([0,1])&\longrightarrow C([0,1]),\\
  v&\longmapsto R(N(v)),
\end{align*}
where $R$ is the linear resolvent of \Cref{lem:markovian-linear-resolvent}.
If $N_\smallertext{+}\coloneqq \sup_{y\in\mathbb R}N(y)$, then positivity of $R$ and the
resolvent bound give
\[
 0\le T v\le \frac{N_\smallertext{+}}{c_\smallertext{-}}.
\]
Hence $T$ maps the closed convex set
\[
  \mathcal K
  \coloneqq 
  \bigg\{
    v\in C([0,1]):
    0\le v\le \frac{N_\smallertext{+}}{c_\smallertext{-}}
  \bigg\},
\]
into itself. Schauder's fixed-point theorem gives a function
$\Phi\in\mathcal K$ such that $T\Phi=\Phi$.

\medskip

We show that the truncation is inactive. Suppose first that
$\min_{[0,1]}\Phi<m$, and let $p_{\smallertext{0}}$ be a point where the minimum is
attained. If $p_{\smallertext{0}}\in(0,1)$, then
$\Phi^{\prime}(p_{\smallertext{0}})=0$ and $\Phi^{\prime\prime}(p_{\smallertext{0}})\ge0$. Thus the truncated equation gives
\[
 N(\Phi(p_{\smallertext{0}}))
 =
 -\alpha(p_{\smallertext{0}})\Phi^{\prime\prime}(p_{\smallertext{0}})+c(p_{\smallertext{0}})\Phi(p_{\smallertext{0}})
 \le
 c_\smallertext{+}\Phi(p_{\smallertext{0}})
 <
 c_\smallertext{+}m.
\]
If $p_{\smallertext{0}}=0$, then $\Phi^{\prime}(0)\ge0$, and the left boundary trace gives
\[
 N(\Phi(0))
 =
 -\lambda^{\smallertext{2}\smallertext{1}}\Phi^{\prime}(0)+c(0)\Phi(0)
 \le
 c_\smallertext{+}m.
\]
The case $p_{\smallertext{0}}=1$ is identical, since a minimum at the right endpoint has
$\Phi^{\prime}(1)\le0$. In all cases we obtain
\[
 N(\Phi(p_{\smallertext{0}}))\le c_\smallertext{+}m.
\]
But $\Phi(p_{\smallertext{0}})<m$ implies $N(\Phi(p_{\smallertext{0}}))=\theta m^q$, contradicting
$\theta m^q>c_\smallertext{+}m$.

\medskip

Similarly, suppose that $\max_{[0,1]}\Phi>M$, and let $p_{\smallertext{1}}$ be a point
where the maximum is attained. At an interior maximum,
$\Phi^{\prime}(p_{\smallertext{1}})=0$ and $\Phi^{\prime\prime}(p_{\smallertext{1}})\le0$, and hence
\[
 N(\Phi(p_{\smallertext{1}}))
 =
 -\alpha(p_{\smallertext{1}})\Phi^{\prime\prime}(p_{\smallertext{1}})+c(p_{\smallertext{1}})\Phi(p_{\smallertext{1}})
 \ge
 c_\smallertext{-}\Phi(p_{\smallertext{1}})
 >
 c_\smallertext{-}M.
\]
At $p_{\smallertext{1}}=0$ one has $\Phi^{\prime}(0)\le0$, while at $p_{\smallertext{1}}=1$ one has
$\Phi^{\prime}(1)\ge0$, so the same inequality follows from the corresponding
boundary trace. Since $\Phi(p_{\smallertext{1}})>M$, however,
\[
 N(\Phi(p_{\smallertext{1}}))=\theta M^q<c_\smallertext{-}M,
\]
a contradiction. Therefore $m\le\Phi\le M$, and $N(\Phi)=\theta\Phi^q$.
Thus $\Phi$ solves
\eqref{eq:two-state-Phi-ode}--\eqref{eq:two-state-Phi-boundary}.

\medskip

It remains to record the sharper bounds. Let $p_*$ be a minimum point of
$\Phi$. Repeating the preceding minimum argument, but now with the
untruncated equation, gives
\[
 \theta\Phi(p_*)^q
 \le
 c_\smallertext{+}\Phi(p_*).
\]
Since $q-1=-1/\theta<0$, this implies
\[
 \Phi(p_*)\ge \bigg(\frac{\theta}{c_\smallertext{+}}\bigg)^\theta.
\]
The maximum argument gives
\[
  \theta\Phi(p^*)^q
  \ge
  c_\smallertext{-}\Phi(p^*),
\]
at a maximum point $p^*$, and therefore
\[
 \Phi(p^*)\le \bigg(\frac{\theta}{c_\smallertext{-}}\bigg)^\theta.
\]
Consequently, $\Phi\in C^1([0,1])\cap C^2((0,1))$ is a positive solution satisfying
\[
  \bigg(\frac{\theta}{c_\smallertext{+}}\bigg)^\theta
  \le
  \Phi(p)
  \le
  \bigg(\frac{\theta}{c_\smallertext{-}}\bigg)^\theta,
  \; p\in[0,1].
\]

\medskip

Finally, because $\Phi>0$ and $\theta>0$, the transformation $\varphi=\Phi^{1/\theta}$ is well defined and reversible. Moreover
\[
  \Phi^{\prime}
  =
  \theta\varphi^{\theta-1}\varphi^{\prime},
  \;
  \Phi^{\prime\prime}
  =
  \theta\varphi^{\theta-1}
  \bigg(
    \varphi^{\prime\prime}
    +(\theta-1)\frac{(\varphi^{\prime})^2}{\varphi}
  \bigg).
\]
Dividing the semilinear equation and its endpoint relations by
$\theta\varphi^{\theta-1}$ yields
\eqref{eq:two-state-varphi-ode}--\eqref{eq:two-state-varphi-boundary}.
The bounds for $\varphi$ follow by taking the $1/\theta$-power in
\eqref{eq:Phi-existence-bounds}.
\end{proof}

\begin{proof}[Proof of \Cref{cor:markovian-existence-uniqueness}]
Existence follows from \Cref{thm:markovian-existence}, while uniqueness
follows from \Cref{thm:markovian-uniqueness}.
\end{proof}

\begin{proof}[Proof of \Cref{thm:markovian-uniqueness}]
Suppose that $\varphi_{\smallertext{1}}$ and $\varphi_{\smallertext{2}}$ are two positive solutions of
\eqref{eq:two-state-varphi-ode}--\eqref{eq:two-state-varphi-boundary}. Set
\[
 q\coloneqq 1-\frac1\theta,
 \;
 \Phi_i\coloneqq \varphi_i^\theta,\; i\in\{1,2\}.
\]
Then both $\Phi_{\smallertext{1}}$ and $\Phi_{\smallertext{2}}$ solve
\eqref{eq:two-state-Phi-ode}--\eqref{eq:two-state-Phi-boundary}. It is
therefore enough to prove uniqueness for the transformed equation. Define
\[
 \rho(p)\coloneqq \frac{\Phi_{\smallertext{1}}(p)}{\Phi_{\smallertext{2}}(p)},
 \;
 M\coloneqq \max_{p\in[0,1]}\rho(p).
\]
Since the solutions are continuous and positive, $M$ is finite.
We claim that $M\le 1$. Suppose, to the contrary, that $M>1$, and let
$p_{\smallertext{0}}\in[0,1]$ be a maximiser of $\rho$. Set
\[
 U(p)\coloneqq \Phi_{\smallertext{1}}(p)-M\Phi_{\smallertext{2}}(p).
\]
Then $U\le0$ on $[0,1]$ and $U(p_{\smallertext{0}})=0$. Write
\[
 \mathcal L[\Phi]
 \coloneqq
 -\alpha\Phi^{\prime\prime}
 -
 \beta\Phi^{\prime}
 +
 c\Phi.
\]
Since $q<1$ and $M>1$, the quantity
$\theta\Phi_{\smallertext{2}}(p_{\smallertext{0}})^q\big(M^q-M\big)$ is negative. At each possible
location of $p_{\smallertext{0}}$, the relevant interior or endpoint equation forces this
expression to be nonnegative.

\medskip

\emph{Interior maximiser.}
If $p_{\smallertext{0}}\in(0,1)$, then $U^{\prime}(p_{\smallertext{0}})=0$ and $U^{\prime\prime}(p_{\smallertext{0}})\le 0$, so
\[
 \mathcal L[\Phi_{\smallertext{1}}](p_{\smallertext{0}}) - M\mathcal L[\Phi_{\smallertext{2}}](p_{\smallertext{0}})
 =
 -\alpha(p_{\smallertext{0}})U^{\prime\prime}(p_{\smallertext{0}}) \ge 0.
\]
By the transformed equation, the left-hand side is also the
negative quantity displayed above, a contradiction.

\medskip
\emph{Boundary maximiser at $p=0$.}
If $p_{\smallertext{0}}=0$, then $U(0)=0$, $U\le 0$ on $[0,1]$, and hence
$U^{\prime}(0)\le 0$. Subtracting $M$ times the left endpoint condition for
$\Phi_{\smallertext{2}}$ from that for $\Phi_{\smallertext{1}}$ yields
\[
 -\lambda^{\smallertext{2}\smallertext{1}}U^{\prime}(0)
 = \theta\Phi_{\smallertext{2}}(0)^{q}\big(M^{q}-M\big).
\]
The left-hand side is nonnegative because $\lambda^{\smallertext{2}\smallertext{1}}\ge0$ and
$U^{\prime}(0)\le0$, while the right-hand side is negative; a
contradiction.

\medskip
\emph{Boundary maximiser at $p=1$.}
If $p_{\smallertext{0}}=1$, then $U(1)=0$, $U\le0$ on $[0,1]$, and the left derivative
satisfies $U^{\prime}(1)\ge0$. The right endpoint condition gives
\[
 \lambda^{\smallertext{1}\smallertext{2}}U^{\prime}(1)
 =
 \theta\Phi_{\smallertext{2}}(1)^{q}\big(M^{q}-M\big).
\]
Again the left-hand side is nonnegative because $\lambda^{\smallertext{1}\smallertext{2}}\ge0$ and
$U^{\prime}(1)\ge0$, while the right-hand side is negative; a
contradiction.

\medskip
Therefore $M\le 1$, and so $\Phi_{\smallertext{1}}\le \Phi_{\smallertext{2}}$. Reversing the roles
of the two solutions yields $\Phi_{\smallertext{2}}\le\Phi_{\smallertext{1}}$, hence $\Phi_{\smallertext{1}}\equiv
\Phi_{\smallertext{2}}$. Since $\varphi=\Phi^{1/\theta}$, the positive solution
$\varphi$ is also unique.
\end{proof}

\begin{proof}[Proof of \Cref{thm:markovian-smoothness}]
Let $\Phi$ be the unique positive solution from
\Cref{cor:markovian-existence-uniqueness}. By
\eqref{eq:Phi-existence-bounds},
\[
  \Phi\geq
  \bigg(\frac{\theta}{c_\smallertext{+}}\bigg)^\theta>0.
\]
Set $f_\Phi\coloneqq\Psi(\Phi)=\theta\Phi^q$. The semilinear equation
and uniqueness in \Cref{lem:markovian-linear-resolvent} imply
\[
  \Phi=Rf_\Phi.
\]
Because $x\longmapsto\theta x^q$ is smooth and real analytic on
$(0,+\infty)$, the forcing term $f_\Phi$ has the same regularity as
$\Phi$.

\medskip

In the interior, the equation can be written as
\[
  \Phi^{\prime\prime}
  =
  \frac{c\Phi-\beta\Phi^{\prime}-\theta\Phi^q}{\alpha}.
\]
Its right-hand side is real analytic as a function of
$(p,\Phi,\Phi^{\prime})$ on
$(0,1)\times(0,+\infty)\times\mathbb R$. The standard analytic
regularity theorem for ordinary differential equations therefore shows
that $\Phi$ is real analytic on $(0,1)$.

\medskip

For regularity at the boundary, the existence theorem first gives
$\Phi\in C^1([0,1])$, and therefore
$f_\Phi\in C^1([0,1])$.
\Cref{lem:linear-regularity-scale} then yields
$\Phi=Rf_\Phi\in C^2([0,1])$. Consequently
$f_\Phi\in C^2([0,1])$, and another application of the lemma gives
$\Phi\in C^3([0,1])$. Iterating proves
$\Phi\in C^\infty([0,1])$.

\medskip

For the data $f=f_\Phi$, the functions $h_{f,j}$ from
\eqref{eq:hfj-def} coincide with the functions $\mathcal E_j$ in
\eqref{eq:hj-def}. The endpoint identities in
\Cref{lem:linear-regularity-scale} therefore give
\eqref{eq:endpoint-hierarchy} at every order. Finally, positivity of
$\Phi$ and smoothness of $x\longmapsto x^{1/\theta}$ on
$(0,+\infty)$ imply the stated properties of
$\varphi=\Phi^{1/\theta}$.
\end{proof}

\begin{remark}[Endpoint compatibility hierarchy]
\label{rem:ventcel-hierarchy}
The relations \eqref{eq:endpoint-hierarchy} form a Ventcel-type
hierarchy generated by the quadratic vanishing of $\alpha$ at
$\{0,1\}$. In particular,
$\alpha^{\prime}(0)=\alpha^{\prime}(1)=0$, so the drift coefficient
of every differentiated equation retains the inward endpoint values
\[
  \beta_j(0)=\lambda^{\smallertext{2}\smallertext{1}},
  \;
  \beta_j(1)=-\lambda^{\smallertext{1}\smallertext{2}}.
\]
The derivative-trace argument can therefore be repeated at every order.
The order-two relations are exactly the endpoint identities used in the
curvature proof. No estimates on higher derivatives of the regularised
solutions, uniform in the regularisation parameter, are needed: all
higher regularity is obtained a posteriori from the degenerate equation.
\end{remark}

\subsection{Monotonicity and curvature}

\begin{proof}[Proof of \Cref{thm:markovian-monotonicity}]
Define the comparison function
\[
 \varphi_{\mathrm{cmp}}(p) \coloneqq \frac{\theta}{c(p)}.
\]
Suppose first that $\eta>0$. Since $c^{\prime}(p)=-\eta<0$, the function
$\varphi_{\mathrm{cmp}}$ is
strictly increasing on $[0,1]$.

\medskip
\emph{Step $1$: every interior local extremum lies on a
specific side of $\varphi_{\mathrm{cmp}}$.}
Suppose $p_{\smallertext{0}}\in(0,1)$ is a local maximum of~$\varphi$. Because
$\Phi=\varphi^{\theta}$ is increasing in~$\varphi$, $p_{\smallertext{0}}$ is
also a local maximum of~$\Phi$, so $\Phi^{\prime}(p_{\smallertext{0}})=0$ and
$\Phi^{\prime\prime}(p_{\smallertext{0}})\le 0$. Evaluating
\eqref{eq:two-state-Phi-ode} at~$p_{\smallertext{0}}$ gives
\[
 c(p_{\smallertext{0}})\Phi(p_{\smallertext{0}})
 \le -\alpha(p_{\smallertext{0}})\Phi^{\prime\prime}(p_{\smallertext{0}}) + c(p_{\smallertext{0}})\Phi(p_{\smallertext{0}})
 = \theta\Phi(p_{\smallertext{0}})^{q},
\]
which, using $q=1-1/\theta$ and $\Phi=\varphi^{\theta}$, is equivalent
to $\varphi(p_{\smallertext{0}})\le \varphi_{\mathrm{cmp}}(p_{\smallertext{0}})$. The dual inequality
$\varphi(p_{\smallertext{0}})\ge \varphi_{\mathrm{cmp}}(p_{\smallertext{0}})$ holds at any interior local minimum.

\medskip
The same one-sided arguments at the endpoints, using the boundary
conditions, give $\varphi(0)\le \varphi_{\mathrm{cmp}}(0)$ if $p=0$ is a one-sided local
maximum, and the analogous statements at $p=1$ and for one-sided minima.

\medskip
\emph{Step $2$: the inequalities and the monotonicity
of $\varphi_{\mathrm{cmp}}$ rule out non-monotonicity.}
Suppose, for contradiction, that $\varphi$ is not non-decreasing.
Then there exists a one-sided local maximum at some
$x\in[0,1)$ followed by a one-sided local minimum at some
$y\in(x,1]$ with $\varphi(x)>\varphi(y)$. By Step~$1$,
\[
 \varphi(x) \le \varphi_{\mathrm{cmp}}(x)
 < \varphi_{\mathrm{cmp}}(y) \le \varphi(y),
\]
contradicting $\varphi(x)>\varphi(y)$. Hence $\varphi$ is
non-decreasing.

\medskip
\emph{Step $3$: ruling out flat segments.}
If $\varphi$ were constant on a nontrivial interval, the same would
hold for~$\Phi$, and \eqref{eq:two-state-Phi-ode} would reduce to
$c(p)\Phi=\theta\Phi^{q}$, \emph{i.e.}\
$\varphi=\theta/c(p)=\varphi_{\mathrm{cmp}}(p)$. But
$\varphi_{\mathrm{cmp}}$ is strictly increasing, so $\varphi$ cannot be constant on any
nontrivial subinterval. Therefore $\varphi$ is
strictly increasing.

\medskip

If $\eta<0$, then $\varphi_{\mathrm{cmp}}$ is strictly decreasing, while the local-extremum
inequalities in Step~1 remain unchanged. If $\varphi$ were not
non-increasing, there would be a one-sided local minimum at
$x\in[0,1)$, followed by a one-sided local maximum at
$y\in(x,1]$, with $\varphi(x)<\varphi(y)$. Step~1 would then give
\[
  \varphi(x)\ge\varphi_{\mathrm{cmp}}(x)
  >\varphi_{\mathrm{cmp}}(y)\ge\varphi(y),
\]
a contradiction. Flat intervals are ruled out exactly as above because
$\varphi=\varphi_{\mathrm{cmp}}$ on any such interval and
$\varphi_{\mathrm{cmp}}$ is strictly decreasing.
Hence $\varphi$ is strictly decreasing.
\end{proof}

\medskip

\begin{proof}[Proof of \Cref{prop:theta-one-affine-solution}]
For $\theta=1$, equation \eqref{eq:two-state-varphi-ode} becomes
$-\alpha(p)\varphi^{\prime\prime}-\beta(p)\varphi^{\prime}+c(p)\varphi=1$. We seek an
affine solution $\varphi(p)=a+bp$. Since $\varphi^{\prime\prime}=0$, the equation
reduces to $-\beta(p)b + c(p)(a+bp) = 1$. Imposing the two endpoint
conditions \eqref{eq:two-state-varphi-boundary} yields the linear system
\[
 \begin{pmatrix}
 c(0) & -\lambda^{\smallertext{2}\smallertext{1}} \\
 c(1) & c(1)+\lambda^{\smallertext{1}\smallertext{2}}
 \end{pmatrix}
 \begin{pmatrix} a \\ b \end{pmatrix}
 = \begin{pmatrix} 1 \\ 1 \end{pmatrix},
\]
whose solution is the pair $(a,b)$.

\medskip
It remains to verify the interior equation. Writing
$c(p)=c(0)-\eta p$ and
$\beta(p)=\lambda^{\smallertext{2}\smallertext{1}}(1-p)-\lambda^{\smallertext{1}\smallertext{2}}p+\eta p(1-p)$ and expanding,
the quadratic terms in $p$ cancel, so
$-\beta(p)b+c(p)(a+bp)$ is an affine function of $p$. Since this
affine function equals $1$ at $p=0$ and at $p=1$, it equals $1$ for
every $p\in[0,1]$. Uniqueness then follows from \Cref{thm:markovian-uniqueness}.
\end{proof}

\medskip

For the curvature argument, differentiate \eqref{eq:two-state-Phi-ode}
twice. Recall that $\Psi(x)=\theta x^q$. Its derivatives are
\[
 \Psi^{\prime}(x) = (\theta-1)x^{-1/\theta},
 \;
 \Psi^{\prime\prime}(x) = \frac{1-\theta}{\theta}x^{-1-1/\theta}.
\]
Note that $\operatorname{sgn}(\Psi^{\prime\prime})
= \operatorname{sgn}(1-\theta)$,
and that
$\Psi^{\prime}(\Phi) = (\theta-1)/\varphi$
since $\Phi=\varphi^{\theta}$. Differentiating \eqref{eq:two-state-Phi-ode} once produces
\[
 -\alpha\Phi^{\prime\prime\prime} - (\alpha^{\prime}+\beta)\Phi^{\prime\prime}
 + \big(c - \beta^{\prime} - \Psi^{\prime}(\Phi)\big)\Phi^{\prime}
 = -c^{\prime}\Phi.
\]
Differentiating once more, and using the identities $c^{\prime\prime}\equiv 0$ and
$\beta^{\prime\prime}=2c^{\prime}$ (both immediate from the explicit forms of $c$ and
$\beta$), the equation for $V\coloneqq \Phi^{\prime\prime}$ reads
\begin{equation}
\label{eq:V-equation}
 -\alpha V^{\prime\prime}
 -(2\alpha^{\prime}+\beta)V^{\prime}
 + \big(c-\alpha^{\prime\prime}-2\beta^{\prime}-\Psi^{\prime}(\Phi)\big)V
 = \Psi^{\prime\prime}(\Phi)(\Phi^{\prime})^2.
\end{equation}
The zero-order coefficient is the \emph{curvature condition}
\[
 m_{\theta}(p)
 \coloneqq c(p) - \alpha^{\prime\prime}(p) - 2\beta^{\prime}(p) - \frac{\theta-1}{\varphi(p)}.
\]
At the endpoints, $\alpha(0)=\alpha(1)=0$ and $\alpha^{\prime}(0)=\alpha^{\prime}(1)=0$,
while $\beta(0)=\lambda^{\smallertext{2}\smallertext{1}}$ and $\beta(1)=-\lambda^{\smallertext{1}\smallertext{2}}$, so the
endpoint forms of \eqref{eq:V-equation} are
\begin{equation}
\label{eq:V-equation-endpoints}
 -\lambda^{\smallertext{2}\smallertext{1}}V^{\prime}(0) + m_{\theta}(0)V(0)
 = \Psi^{\prime\prime}(\Phi(0))(\Phi^{\prime}(0))^{2},
 \;
 \lambda^{\smallertext{1}\smallertext{2}}V^{\prime}(1) + m_{\theta}(1)V(1)
 = \Psi^{\prime\prime}(\Phi(1))(\Phi^{\prime}(1))^{2}.
\end{equation}

\begin{proof}[Proof of \Cref{thm:markovian-curvature}]
We treat the three cases in turn.

\medskip
\emph{Case $0<\theta<1$.}
Then $\Psi^{\prime\prime}(\Phi)>0$, so the right-hand side of
\eqref{eq:V-equation} is nonnegative. We claim $V\ge 0$ on
$[0,1]$. Suppose, for contradiction, that $V$ has a negative minimum
at some $p_{\smallertext{0}}\in[0,1]$.

\medskip
If $p_{\smallertext{0}}\in(0,1)$, then $V^{\prime}(p_{\smallertext{0}})=0$ and $V^{\prime\prime}(p_{\smallertext{0}})\ge 0$, so the
left-hand side of \eqref{eq:V-equation} at $p_{\smallertext{0}}$ equals
\[
 -\alpha(p_{\smallertext{0}})V^{\prime\prime}(p_{\smallertext{0}}) + m_{\theta}(p_{\smallertext{0}})V(p_{\smallertext{0}}) \;<\; 0,
\]
since $m_{\theta}(p_{\smallertext{0}})>0$ and $V(p_{\smallertext{0}})<0$, while the right-hand side is
nonnegative; contradiction. At a boundary minimum
$p_{\smallertext{0}}=0$ we have $V^{\prime}(0)\ge 0$, so the first endpoint identity
in \eqref{eq:V-equation-endpoints} has a negative left-hand
side and a nonnegative right-hand side; contradiction. The case
$p_{\smallertext{0}}=1$ is symmetric, with $V^{\prime}(1)\le 0$.

\medskip
Hence $V=\Phi^{\prime\prime}\ge 0$ on $[0,1]$. From
$\varphi=\Phi^{1/\theta}$ we obtain
\begin{equation}
\label{eq:phi-second-from-Phi}
 \varphi^{\prime\prime}
 = \frac{1}{\theta}\Phi^{1/\theta-1}\Phi^{\prime\prime}
 + \frac{1}{\theta}\bigg(\frac{1}{\theta}-1\bigg)
 \Phi^{1/\theta-2}(\Phi^{\prime})^{2}.
\end{equation}

\medskip

For $0<\theta<1$ we have $1/\theta-1>0$ and $\Phi^{\prime\prime}\ge 0$, so both
terms in \eqref{eq:phi-second-from-Phi} are nonnegative; thus
$\varphi^{\prime\prime}\ge 0$ and $\varphi$ is convex.

\medskip
\emph{Case $\theta>1$.}
Now $\Psi^{\prime\prime}(\Phi)<0$, so the right-hand side of
\eqref{eq:V-equation} is nonpositive. By the symmetric maximum-principle
argument (positive maximum of $V$ at an interior or boundary
point), one obtains $V=\Phi^{\prime\prime}\le 0$ on $[0,1]$. Since
$1/\theta-1<0$, both terms in \eqref{eq:phi-second-from-Phi} are
nonpositive, so $\varphi^{\prime\prime}\le 0$ and $\varphi$ is concave.

\medskip
\emph{Case $\theta=1$.} This has been treated separately in \Cref{prop:theta-one-affine-solution}.
\end{proof}

\begin{proof}[Proof of \Cref{prop:primitive-curvature-conditions}]
The condition $m_\theta>0$ in \Cref{thm:markovian-curvature} depends on
the solution $\varphi$ through the term $(\theta-1)/\varphi$. The
following estimates eliminate this dependence for the unique positive
Markovian solution.
Set
\[
 \lambda_\Sigma \coloneqq
 \lambda^{\smallertext{1}\smallertext{2}}+\lambda^{\smallertext{2}\smallertext{1}},
 \;
 k \coloneqq \frac{(\Delta g)^2}{2\sigma^2},
 \;
 c_{\smallertext{0}} \coloneqq c(0),
 \;
 c_{\smallertext{1}} \coloneqq c(1) = c_{\smallertext{0}} - \eta.
\]
A direct calculation gives
\[
 \alpha^{\prime\prime}(p) = 2k\big(1-6p+6p^{2}\big),
 \;
 \beta^{\prime}(p) = -\lambda_\Sigma + \eta(1-2p),
\]
so
\[
 \mathcal D(p) \coloneqq c(p) - \alpha^{\prime\prime}(p) - 2\beta^{\prime}(p)
 = c_{\smallertext{0}} - 2\eta + 2\lambda_\Sigma - 2k
 + (3\eta + 12k)p - 12kp^{2}.
\]
The function $\mathcal D$ is a concave quadratic in $p$, and under
$\eta>0$ its minimum on $[0,1]$ is attained at $p=0$, with value
$\mathcal D(0)=c_{\smallertext{0}} - 2\eta + 2\lambda_\Sigma - 2k$.

\medskip
Since $\varphi$ is the positive Markovian solution,
\Cref{thm:markovian-monotonicity} applies. Step~1 of its proof yields the
maximum-principle bounds
\begin{equation}
\label{eq:phi-mp-bounds}
 \frac{\theta}{c_{\smallertext{0}}} \le \varphi(p) \le \frac{\theta}{c_{\smallertext{1}}},
 \; p\in[0,1].
\end{equation}

\medskip

Combining $\mathcal D(p)\ge \mathcal D(0)$ with \eqref{eq:phi-mp-bounds} we obtain

\medskip
\emph{Convex regime $(0<\theta<1)$.}
Then $(\theta-1)/\varphi \le 0$ and, by
\eqref{eq:phi-mp-bounds}, $1/\varphi\ge c_{\smallertext{1}}/\theta$. Therefore
\[
 m_\theta(p)
 \ge \mathcal D(0) + \frac{(1-\theta)c_{\smallertext{1}}}{\theta}
 = \frac{c_{\smallertext{1}}}{\theta} + 2\lambda_\Sigma - \eta - 2k.
\]

\medskip

Hence a sufficient primitive condition for convexity is
\[
 \frac{c(1)}{\theta} + 2(\lambda^{\smallertext{1}\smallertext{2}}+\lambda^{\smallertext{2}\smallertext{1}})
 > (1-\gamma)(g^{\smallertext{1}}-g^{\smallertext{2}}) + \frac{(g^{\smallertext{1}}-g^{\smallertext{2}})^{2}}{\sigma^{2}}.
\]
Under this condition $m_\theta(p)>0$ along the Markovian solution on
$[0,1]$, and \Cref{thm:markovian-curvature} implies that $\varphi$ is
convex.

\medskip

\emph{Concave regime $(\theta>1)$.}
Now $(\theta-1)/\varphi \le (\theta-1)c_{\smallertext{0}}/\theta$, so
\[
 m_\theta(p)
 \ge \mathcal D(0) - \frac{(\theta-1)c_{\smallertext{0}}}{\theta}
 = \frac{c_{\smallertext{0}}}{\theta} - 2\eta + 2\lambda_\Sigma - 2k.
\]
The corresponding sufficient primitive condition for concavity is
\[
 \frac{c(0)}{\theta} + 2(\lambda^{\smallertext{1}\smallertext{2}}+\lambda^{\smallertext{2}\smallertext{1}})
 > 2(1-\gamma)(g^{\smallertext{1}}-g^{\smallertext{2}}) + \frac{(g^{\smallertext{1}}-g^{\smallertext{2}})^{2}}{\sigma^{2}}.
\]
Under this condition $m_\theta(p)>0$ along the Markovian solution on
$[0,1]$, and \Cref{thm:markovian-curvature} implies that $\varphi$ is
concave.

\medskip

The same maximum-principle mechanism therefore yields convexity for
$\theta<1$, affinity for $\theta=1$, and concavity for $\theta>1$,
under conditions that involve only the structural parameters
$(\gamma,g^{\smallertext{1}},g^{\smallertext{2}},\sigma,\lambda^{\smallertext{1}\smallertext{2}},\lambda^{\smallertext{2}\smallertext{1}},\kappa,\theta)$.
\end{proof}

\section{Discounted characteristic functions and Fourier pricing}
\label{app:characteristic-functions}

The risk-neutral measure in this appendix is the finite-horizon measure
defined under \Cref{assum:risk-neutral-martingale},
with $r_f$ equal to the marginal short rate
\eqref{eq:marginal-short-rate-option}. We retain the notation
and suppression convention of \Cref{sec:option-pricing} for the fixed
$\iota(p_{\smallertext{0}})$-model copy.

\medskip

Since the coefficients in \eqref{eq:option-PDE} depend on $p$ but not
on $X$, the pricing problem is translation-invariant in log stock price.
Let $\mathrm i\coloneqq\sqrt{-1}$. For $u\in\mathbb C$, define the
discounted characteristic function of the
log return by
\begin{equation*}
 \Upsilon(u,p,\tau)
 \coloneqq 
 \mathbb E^{\mathbb Q}_{0,p}
 \Bigg[
 \exp\bigg(
 -\int_{\smallertext{0}}^\tau r_f(p_s)\mathrm{d}s
 \bigg)
 \exp\big(
 \mathrm i uX_\tau
 \big)
 \Bigg].
\end{equation*}
The normalisation $X_{\smallertext{0}}=0$ is without loss of generality because only log
returns enter the transform. Under the standing conditions of
\Cref{sec:option-pricing}, the functions $r_f$, $q_\smallertext{S}$,
$\sigma_\smallertext{S}$, $\chi$, and $\widetilde B$ are bounded on
$[0,1]$. The log return therefore has exponential moments of every order on
each finite horizon, so this expectation is finite for every
$u\in\mathbb C$ and finite $\tau$.

\medskip

\begin{theorem}[Characteristic-function PDE]
\label{thm:characteristic-PDE}
For each $u\in\mathbb C$, suppose the map
$(\tau,p)\longmapsto\Upsilon(u,p,\tau)$ belongs to
$C^{1,2}((0,+\infty)\times(0,1))$. Then it satisfies
\begin{equation}
\label{eq:characteristic-PDE}
\begin{aligned}
 \Upsilon_\tau
 &=
 \frac12\chi^2(p)\Upsilon_{pp}
 +
 \big(
 \widetilde B(p)
 +
 \mathrm i u\sigma_\smallertext{S}(p)\chi(p)
 \big)\Upsilon_p
 +
 \bigg(
 \mathrm i u\big(r_f(p)-q_\smallertext{S}(p)\big)
 -
 \frac12(u^2+\mathrm i u)\sigma_\smallertext{S}^2(p)
 -
 r_f(p)
 \bigg)\Upsilon,
\end{aligned}
\end{equation}
with initial condition
\begin{equation*}
 \Upsilon(u,p,0)=1.
\end{equation*}
If $\Upsilon$,
$\Upsilon_\tau$, and
$\Upsilon_p$ admit continuous one-sided traces and
\[
 \chi^2(p)\Upsilon_{pp}(u,p,\tau)\longrightarrow0
\]
as $p\downarrow0$ and as $p\uparrow1$, then the endpoint traces satisfy
\begin{align*}
 \Upsilon_\tau(u,0,\tau)
 &=\lambda^{\smallertext{2}\smallertext{1}}\Upsilon_p(u,0,\tau)
   +\mathcal K(u,0)\Upsilon(u,0,\tau),\\
 \Upsilon_\tau(u,1,\tau)
 &=-\lambda^{\smallertext{1}\smallertext{2}}\Upsilon_p(u,1,\tau)
   +\mathcal K(u,1)\Upsilon(u,1,\tau),
\end{align*}
where
\[
 \mathcal K(u,p)
 \coloneqq
 \mathrm i u\big(r_f(p)-q_\smallertext{S}(p)\big)
 -\frac12(u^2+\mathrm i u)\sigma_\smallertext{S}^2(p)-r_f(p).
\]
\end{theorem}

\begin{proof}
Apply the generator in \Cref{thm:option-PDE} to functions of the
form
\[
 \mathfrak f_{u,v}(X,p)=\mathrm{e}^{\mathrm i u\smallertext{X}}v(p).
\]
Then $\partial_X\mathfrak f_{u,v}=\mathrm i u\mathfrak f_{u,v}$,
$\partial_{XX}\mathfrak f_{u,v}=-u^2\mathfrak f_{u,v}$, and
$\partial_{Xp}\mathfrak f_{u,v}
=\mathrm i u\mathrm{e}^{\mathrm i u\smallertext{X}}v^{\prime}(p)$. Using
\[
  b_\smallertext{X}(p)=r_f(p)-q_\smallertext{S}(p)-\frac12\sigma_\smallertext{S}^2(p),
\]
the coefficient multiplying $v(p)$ is
\[
  \mathrm i u b_\smallertext{X}(p)
  -
  \frac12u^2\sigma_\smallertext{S}^2(p)
  -
  r_f(p)
  =
  \mathrm i u\big(r_f(p)-q_\smallertext{S}(p)\big)
  -
  \frac12(u^2+\mathrm i u)\sigma_\smallertext{S}^2(p)
  -
  r_f(p).
\]
The coefficient multiplying $v^{\prime}(p)$ is
$\widetilde B(p)+\mathrm i u\sigma_\smallertext{S}(p)\chi(p)$, which gives
\eqref{eq:characteristic-PDE}. At the endpoints, $\chi=0$ and
$\widetilde B(0)=\lambda^{\smallertext{2}\smallertext{1}}$,
$\widetilde B(1)=-\lambda^{\smallertext{1}\smallertext{2}}$. Taking the one-sided traces gives the two
endpoint equations.
\end{proof}

Two special values will be used in the call-price formula. At the current
state $(p_t,T-t)$, the zero-coupon bond price is
\[
 \Upsilon(0,p_t,T-t)
 =
 \mathbb E^{\mathbb Q}
 \Bigg[
 \exp\bigg(
 -\int_t^{\smallertext{T}} r_f(p_s)\mathrm{d}s
 \bigg)
 \Bigg|\Fc_t^\smallertext{D}
 \Bigg],
\]
while
\[
 \Upsilon(-\mathrm i,p_t,T-t)
 =
 \frac{1}{S_t}
 \mathbb E^{\mathbb Q}
 \Bigg[
 \exp\bigg(
 -\int_t^{\smallertext{T}} r_f(p_s)\mathrm{d}s
 \bigg)
 S_T\Bigg|\Fc_t^\smallertext{D}
 \Bigg]
\]
is the discounted value of one terminal stock share, normalised by the
current stock price. Let $m\coloneqq \log(K/S)$. Define the normalised
characteristic functions
\begin{equation*}
 \Upsilon_{\smallertext{1}}(u,p,\tau)
 \coloneqq 
 \frac{\Upsilon(u-\mathrm i,p,\tau)}{\Upsilon(-\mathrm i,p,\tau)},
 \;
 \Upsilon_{\smallertext{2}}(u,p,\tau)
 \coloneqq 
 \frac{\Upsilon(u,p,\tau)}{\Upsilon(0,p,\tau)}.
\end{equation*}
Define
\begin{equation}
\label{eq:Pi-j-option}
 \Pi_j(m,p,\tau)
 \coloneqq 
 \frac12
 +
 \frac1\pi
 \int_{\smallertext{0}}^{\smallertext{+}\infty}
 \mathrm{Re}
 \bigg[
 \frac{
 \mathrm{e}^{-\mathrm i u m}\Upsilon_j(u,p,\tau)
 }{\mathrm i u}
 \bigg]\mathrm{d}u,
 \; j\in\{1,2\}.
\end{equation}
The integral is understood as an improper Fourier integral.

\begin{proposition}[Fourier representation of call prices]
\label{prop:fourier-call-price}
At continuity points of the two tilted return distributions, and whenever
the improper integrals in \eqref{eq:Pi-j-option} converge, the European call
price is
\begin{equation}
\label{eq:fourier-call-price}
 C^{\mathrm{lev}}(S,p,\tau;K)
 =
 S\Upsilon(-\mathrm i,p,\tau)\Pi_{\smallertext{1}}(m,p,\tau)
 -
 K\Upsilon(0,p,\tau)\Pi_{\smallertext{2}}(m,p,\tau).
\end{equation}
\end{proposition}

\begin{proof}
Fix $0\leq t<T$, set $\tau\coloneqq T-t$, and let
\[
 \Delta X_{t,T}\coloneqq X_T-X_t=\log(S_T/S_t).
\]
Then
\[
 (S_\smallertext{T}-K)^{\smallertext{+}}
 =
 S_t \mathrm{e}^{\smallertext{\Delta} \smallertext{X}_{\smalltext{t}\smalltext{,}\smalltext{T}}}
 \mathbf 1_{\{\smallertext{\Delta} \smallertext{X}_{\smalltext{t}\smalltext{,}\smalltext{T}}>m\}}
 -
 K\mathbf 1_{\{\smallertext{\Delta} \smallertext{X}_{\smalltext{t}\smalltext{,}\smalltext{T}}>m\}},
 \;
 m=\log(K/S_t).
\]
Therefore
\[
\begin{aligned}
 C_t
 &=
 S_t
 \mathbb E^{\mathbb Q}
 \Big[
 \mathrm{e}^{-\int_\smalltext{t}^\smalltext{T} r_\smalltext{f}(p_\smalltext{s})\mathrm{d}s}
 \mathrm{e}^{\smallertext{\Delta} \smallertext{X}_{\smalltext{t}\smalltext{,}\smalltext{T}}}
 \mathbf 1_{\{\smallertext{\Delta} \smallertext{X}_{\smalltext{t}\smalltext{,}\smalltext{T}}>m\}}
 \Big|\Fc_t^\smallertext{D}\Big]
 -
 K
 \mathbb E^{\mathbb Q}
 \Big[
 \mathrm{e}^{-\int_\smalltext{t}^\smalltext{T} r_\smalltext{f}(p_\smalltext{s})\mathrm{d}s}
 \mathbf 1_{\{\Delta X_{t,T}>m\}}
 \Big|\Fc_t^\smallertext{D}\Big].
\end{aligned}
\]
The normalising constants in the two terms are respectively
\[
 \Upsilon(-\mathrm i,p_t,T-t)
 =
 \mathbb E^{\mathbb Q}
 \Big[
 \mathrm{e}^{-\int_\smalltext{t}^\smalltext{T} r_\smalltext{f}(p_\smalltext{s})\mathrm{d}s}
 \mathrm{e}^{\Delta X_{t,T}}
 \Big|\Fc_t^\smallertext{D}\Big],\; \text{
and}\;
 \Upsilon(0,p_t,T-t)
 =
 \mathbb E^{\mathbb Q}
 \Big[
 \mathrm{e}^{-\int_\smalltext{t}^\smalltext{T} r_\smalltext{f}(p_\smalltext{s})\mathrm{d}s}
 \Big|\Fc_t^\smallertext{D}\Big].
\]
After normalisation, $\Upsilon_{\smallertext{1}}$ and
$\Upsilon_{\smallertext{2}}$ are characteristic
functions under the stock-tilted and $T$-forward measures,
respectively. At continuity points, the standard Fourier
inversion formula for recovering a distribution function from its
characteristic function
\cite[pages~481--482]{gilpelaez1951note} gives
\eqref{eq:Pi-j-option} for the event
$\{\Delta X_{t,\smallertext{T}}>m\}$. The same
inversion-based option decomposition is used by
\citeauthor*{david2002option}
\cite[Proposition~4]{david2002option}. Substitution yields
\eqref{eq:fourier-call-price}.
\end{proof}

\bibliography{bibliography}
\end{document}